\documentclass[a4paper,11pt,leqno]{amsart}
\usepackage{mathtools}
\usepackage{a4wide,array}
\usepackage{cite}
\usepackage{hyperref}
\usepackage{amsmath,amssymb,amsthm}
\usepackage[english]{babel}
\usepackage[applemac]{inputenc}
\usepackage{stmaryrd}
\usepackage[cyr]{aeguill}
\usepackage{esint}
\usepackage[usenames,dvipsnames]{color}
\usepackage{xargs}
\usepackage{enumitem}
\usepackage[normalem]{ulem}
\setlist{nosep} 
\usepackage{xcolor}

\newcommand{\N}{\mathbb{N}}

\newcommand{\R}{\mathbb{R}}
\newcommand{\C}{\mathbb{C}}

\newtheorem{theo}{Theorem}
\newtheorem{prop}{Proposition}[section]
\newtheorem{lem}[prop]{Lemma}
\newtheorem{coro}[prop]{Corollary}

\newtheorem{defi}[prop]{Definition}

\theoremstyle{plain}

\numberwithin{equation}{section}

\def\t0{\rightarrow 0} 
\def\ti{\rightarrow \infini} 

\newcommand{\f}{\frac}
\newcommand{\infini}{\infty}
\newcommand{\ep}{\varepsilon}

\newcommand{\hal}{\frac{1}{2}}
\renewcommand{\(}{\left(}
\renewcommand{\)}{\right)}

\def\div{\mathrm{div} \, } 
\def\1{\mathbf{1}} 

\def \mc{\mathcal}
\def \p{\partial}
\def \ep{\epsilon}
\def\nab{\nabla}
\renewcommand{\epsilon}{\varepsilon}

\def \carr{C} 
\def \Nn{\mc{N}}
\def \D{\mc{D}}

\def\Lip{\mathrm{Lip}_1} 
\def \Lploc{L^p_{\mathrm{loc}}}

\def\dconfig{d_{\config}}

\def \PNbeta{\P_{N, \beta}} 
\def \PgN2{\mathbf{P}_{N,2}} 

\def \bQpN{\bar{\mathfrak{Q}}_{N}} 

\def\Esp{\mathbf{E}} 

\def\lc{\left\langle}
\def\rc{\right\rangle}

\def \fbarbeta{\overline{\mathcal{F}}_{\beta}} 

\def \A{\mathcal{A}} 
\def \conf{\mathrm{Conf}} 
\def \Eloc{E^{\mathrm{loc}}}

\def \K{\mathcal{K}}

\def \dist{d}

\def \dist{\mathrm{dist}}

\def\XXint#1#2#3{{\setbox0=\hbox{$#1{#2#3}{\int}$}
     \vcenter{\hbox{$#2#3$}}\kern-.5\wd0}}

\def \Pelec{\Pst^{\mathrm{elec}}}
\def \bPelec{\bPst^{\mathrm{elec}}}

\def \WN{w_N} 
\def \XN{\vec{X}_N} 
\def \XpN{\vec{X}'_N}
\def \YN{\vec{Y}_N}

\def \YpN{\vec{Y}'_N}
\def \Np{N_{+}} 
\def \Nm{N_{-}}

\def \Box{\Lambda} 
\def \PNbeta{\mathbb{P}^{\beta}_{N}} 
\def \ZNbeta{Z_{N,\beta}} 

\def\config{\mathcal{X}^0} 

\def\configs{\mathcal{X}} 
\def \C{\mathcal{C}}
\def \Cs{\mathcal{C}}
\def \Cp{\mathcal{C}^+}
\def \Cm{\mathcal{C}^-}

\def\probas{\mathcal{P}} 
\def\Pst{P} 
\def\bPst{\bar{P}} 
\def \Pelec{P^{\mathrm{elec}}} 

\def\ro{r}

\def\cW{\mathcal{W}^o}
\def\dW{\mathbb{W}^o}
\def\tdW{\widetilde{\mathbb{W}}^o}
\def\bdW{\overline{\mathbb{W}}^o}

\def\Detat{\mathbb{W}^{\star}_{\tau}}
\def\Deta{\mathbb{W}^{\star}}
\def\tDeta{\widetilde{\mathbb{W}}^{\star}}
\def\bDeta{\overline{\mathbb{W}}^{\star}}
\def\tDetat{\widetilde{\mathbb{W}}^{\star}_{\tau}}
\def\bDetat{\overline{\mathbb{W}}^{\star}_{\tau}}

\def\mup{\mu^+}

\def\mum{\mu^-}

\def\muNp{\mu^+_N}
\def\muNm{\mu^-_N}
\def\rhop{\rho^+}
\def\rhom{\rho^-}

\def \Pst{P} 
\def \bPst{\bar{P}} 
\def \Psts{P} 
\def \bPsts{\bar{P}} 
\def \bPstsp{\bar{P}^{+}}
\def \bPstsm{\bar{P}^{-}}
\def \bPN{\bar{P}_N}

\def \Ent{\mathrm{Ent}}   
\def \ERS{\mathsf{ent}} 
\def \bERS{\overline{\mathsf{ent}}} 

\def \Leb{\mathbf{Leb}}
\def \Poisson{\mathbf{\Pi}}
\def \B{\mathbf{B}} 

\def\pb{p}
\def\qb{q}
\def\iN{i_N}

\def\H{H}
\def \nuppN{\hat{\nu}^+_N}
\def \nupmN{\hat{\nu}^-_N}
\def \nupN{\nu^+_N}
\def \numN{\nu^-_N}

\def \indic{\mathbf{1}}
\def \Poissons{\Poisson^{s}}

\def \tWs{\widetilde{\mathbb{W}}}
\def \bWs{\overline{\mathbb{W}}}
\def \bPstsx{\bar{P}^{x}}
\def \VNr{V_{N,r}}
\def \Er{E_r}
\def \Cscr{\C^{\rm{scr}}}
\def \Escr{E^{\rm{scr}}}
\def \Escrr{\Escr_{r}}

\def \bPsN{\bar{P}_{N}}
\def \VpN{V'_{N,r}}
\def\VN{V'_N}
\def \KNbeta{K_{N,\beta}}
\def \Phiscr{\Phi^{\rm{scr}}}

\def \Ccp{\mathbf{C}^{+}}
\def \Ccm{\mathbf{C}^{-}}
\def \Ccs{\mathbf{C}}
\def \BM{\dot{\mathfrak{M}}}
\def \HM{\widehat{\mathfrak{M}}}
\def \Amod{A^{\rm{mod}}}
\def \Cmods{\C^{\rm{mod}}}
\def \Cmodp{\C^{\rm{mod}, +}}
\def \Cmodm{\C^{\rm{mod}, -}}
\def \Emod{E^{\rm{mod}}}

\def \mNR{m_{N,R}}

\def \Cabss{\C^{\rm{abs,s}}}

\def \Aabs{A^{\rm{abs}}}
\def \Phiscr{\Phi^{\rm{scr}}}

\def \Phimod{\Phi^{\rm{mod}}}

\def \FR{F_{R}}
\def \OR{\mc{O}_{R}}
\def \G{\Gamma}

\def \Pelec{P^{\rm{elec}}}
\def \bPelec{\bar{P}^{\rm{elec}}}
\def \conf{\rm{Conf }}
\def \bQpN{\bar{\mathfrak{Q}}^{\beta}_{N}} 
\def \probinv{\probas_{\rm{inv}}} 
\def \Elec{\mathsf{Elec}}
\def \bPgotNb{\overline{\mathfrak{P}}^{\beta}_N}
\def \Pstsp{P^{+}}
\def \Pstsm{P^{-}}
\def \fbarbeta{\overline{\mathcal{F}}_{\beta}} 
\def \fbsc{\overline{\mathcal{F}}^{\rm{sc}}_{\beta}}
\def \SN{\vec{S}_{2N}}
\def \bQ{\bar{Q}}

\newcommand{\np}{n^{+}}
\newcommand{\nm}{n^{-}}
\newcommand{\nintp}{n^{+, \epsilon}_{\mathrm{int}}}
\newcommand{\nintm}{n^{-, \epsilon}_{\mathrm{int}}}
\newcommand{\nint}{n^{\epsilon}_{\mathrm{int}}}

\newcommand{\Pru}{\mathrm{Pr}}

\def \BNtau{\Box_{N,\tau}^n}

\title{Large Deviations for the Two-Dimensional Two-Component Plasma}
\author{Thomas Lebl\'e, Sylvia Serfaty, Ofer Zeitouni, with an appendix by Wei Wu}
\address[Thomas Leblé]{Sorbonne Universit\'es, UPMC Univ. Paris 06, CNRS, UMR 7598, Laboratoire Jacques-Louis Lions, 4, place Jussieu 75005, Paris, France.}
\email{leble@ann.jussieu.fr}
\address[Sylvia Serfaty]{Sorbonne Universit\'es, UPMC Univ. Paris 06, CNRS, UMR 7598, Laboratoire Jacques-Louis Lions, 4, place Jussieu 75005, Paris, France.
 \newline \& Institut Universitaire de France \newline \& 
Courant Institute, New York University, 251 Mercer st, New York, NY 10012, USA.} 
\email{serfaty@ann.jussieu.fr}

\address[Ofer Zeitouni]{Department of Mathematics,
Weizmann Institute of Science, POB 26, Rehovot 76100, Israel
\newline \& Courant Institute, New York University, 251 Mercer Street, New York,
NY 10012, USA.}
\email{ofer.zeitouni@weizmann.ac.il}
\address[Wei Wu]{Courant Institute, New York University, 
251 Mercer st, New York, NY 10012, USA \newline
\& New York University - Shanghai, 1555 Century Ave, Pudong Shanghai, CN 200122,
China.} 
\email{weiwu@cims.nyu.edu}
\begin{document}

\begin{abstract}
We derive a large deviations principle 
for the two-dimensional two-component plasma in a box. 
As a consequence,
we obtain a variational representation for the free energy, and also show
that the macroscopic empirical measure of either positive or negative charges 
converges to the uniform measure. An appendix,
 written by Wei Wu, discusses applications to the supercritical complex Gaussian
multiplicative chaos.
\end{abstract}
\maketitle

\section{Introduction}
\subsection{General setting}
The two-dimensional two-component plasma is a standard ensemble of statistical mechanics, in which $N$ particles of positive charge and $N$ particles of negative charge  interact logarithmically in the plane, cf \cite{Frohlich,GunPan,DeutschLavaud,forrester}. The associated Gibbs measure at inverse temperature $\beta>0$  is given by 
\begin{equation}
\label{def:Gibbs}
d\PNbeta(\XN, \YN) :=\frac{1}{\ZNbeta} e^{- \frac{\beta}{2} \WN(\XN, \YN)}d\XN \, d\YN,
\end{equation}
where $\ZNbeta$ is the normalizing constant, i.e. the  \textit{partition function}
\begin{equation}
\label{def:ZNbeta}
\ZNbeta := \int_{\Box^{2N}} e^{-\frac{\beta}{2} \WN(\XN, \YN)} d\XN d\YN,
\end{equation}
and we have written 
\begin{equation}
\label{def:Ham} \WN(\XN, \YN) := \sum_{1\leq i \neq j \leq N} - \log|x_i-x_j| - \log |y_i - y_j| + \sum_{1 \leq i,j \leq N} \log |x_i - y_j|
\end{equation} for 
any $N \geq 1$ and any $N$-tuples $\XN = (x_1, \dots, x_N)$ and $\YN = (y_1, \dots, y_N)$ of points in, say, the unit cube $\Box := [0,1]^2$ of $\R^2$. The notation $ d\XN \, d\YN$ refers to the Lebesgue measure on $\Box^{2N}$.
The choice of ${\beta}/{2}$ 
instead of $\beta$ in the exponent of \eqref{def:Gibbs} is made in order to match the existing literature. In physical terms, $\WN(\XN, \YN)$ computes the two-dimensional electrostatic (or logarithmic) interaction of the point charges $(x_1, \dots, x_N)$ and $(y_1, \dots, y_N)$, the former carrying a $+1$ charge and the latter a $-1$ charge. 

We are interested in proving a  Large Deviation Principle (LDP) 
on the Gibbs measure $\PNbeta$, which is inspired by 
\cite{LebSer}, where
such a result was obtained for the one-component plasma in 
arbitrary dimension. 
The (say, two-dimensional) one-component plasma corresponds to a system 
of point charges which all have {\it same sign}  and interact logarithmically, 
but that need to be confined by some external 
potential, acting in effect like a slowly varying neutralizing 
(opposite) charge distribution.  

Motivations for studying two-component plasmas are numerous. Besides its intrinsic interest as a toy model for classical electrons and ions, 
it is also related to the  so-called Sine-Gordon model: the grand canonical partition function of the two component plasma can be related to the Euclidean version of the sine-Gordon partition function, and the Coulomb gas on a lattice   itself related to the XY model and the Kosterlitz-Thouless phase transition (see the review \cite{spencer} and references therein). 
 Another motivation, which will be described in more details 
 in the appendix, is the connection with the partition function of 
 ``complex multiplicative Gaussian chaos" (cf. \cite{LRV}), which may be 
 formally written as
$\int e^{i \beta h(x)} \, dx$, where $h(x)$ is the Gaussian Free Field. 
This question is 
itself related to height functions of dimer models and 
to
the Lee-Yang theorem for the XY model.
It turns out that computing moments of $e^{i\beta h}$ makes the 
Gibbs measure of the two-component plasma appear 
and, as described in the appendix,
our results yield a rate of decay in terms of $\beta $ of the 
tails of this partition function.

Due to the presence of point charges of opposite signs, 
the system is unstable at low temperature because the 
thermal excitation does not compensate the energetical trend for   configurations with $- \infty$ energy, in which at least two particles of opposite signs collide. The domain of stability of the system (i.e. the range of the parameter $\beta$ for which the integral in \eqref{def:ZNbeta} converges, so that \eqref{def:Gibbs} makes sense) was found to be $\beta < 2$ in \cite{DeutschLavaud}, together with a first bound on $\ZNbeta$.  A more accurate estimate and the existence of the thermodynamic limit were proven in \cite{Frohlich} by Euclidean quantum field techniques, and \cite{GunPan} by classical methods (the question of obtaining it by classical methods was apparently first raised in \cite{sm}). 
The result can be summarized as follows.
\begin{prop}[Gunson-Panta, \cite{GunPan}] \label{prop:FGP}
For any $\beta < 2$,
\begin{equation} \label{expZN1}
\log \ZNbeta = \frac{\beta}{2} N \log N + C_{\beta}N + o(N),
\end{equation}
with a constant $C_{\beta}$ and an error $o(N)$ both depending on $\beta$. 
\end{prop}
Later, a number of other results, such as the asymptotics of 
two-point correlation functions and  the
thermodynamic properties of the two-component plasma,
were obtained in the physics literature. 
We refer to the review \cite{samaj} and references therein. 

For any $\XN, \YN$, let $\muNp$ and $\muNm$ be the empirical measures associated to the positive and negative charges
$$
\muNp := \frac{1}{N} \sum_{i=1}^N \delta_{x_i}, \quad \muNm  := \frac{1}{N} \sum_{i=1}^N \delta_{y_i}.
$$
A natural question is to ask for the large $N$ behavior of these 
\textit{macroscopic} quantities in the space $\probas(\Box)$ of 
probability measures on $\Box$. At the microscopic level, 
we may also ask whether the point process induced by $\PNbeta$ has 
a typical behavior. In this paper, we give an LDP for 
a 
spatially averaged
 microscopic behavior and as a consequence we show that both empirical measures $\muNp, \muNm$ converge a.s. to the uniform measure on $\Lambda$.
We remark that in case $\beta$ scales as $1/N$, a LDP and Gaussian
fluctuations limits
for these quantities are derived in \cite{BoG}; the techniques 
needed to handle the scaling in this paper
are completely different.

\subsection{Main result} Before stating our main result, we need to 
introduce some notation and concepts, which are all ``signed" 
versions of those introduced in \cite{SS2d, LebSer} 
(in the latter papers,
 all the charges have the same sign).  We  denote by $\configs$  the set of 
 locally finite 
signed point configurations with the topology of local convergence (for more details see Section \ref{sec:PrelimNot}).  If $(\XN, \YN)$ is a pair of $N$-tuples of points in the square $\Box$, we may see it as an element of the space $\configs$ by associating to $\XN$ (resp. $\YN$) the point configuration $\nupN := \sum_{i=1}^N \delta_{x_i}$ (resp. $\numN := \sum_{i=1}^N \delta_{y_i}$). When starting from $(\XN, \YN)$, we first rescale the associated finite signed configurations by a factor $\sqrt{N}$ to get 
$$
\nuppN := \sum_{i=1}^N \delta_{\sqrt{N} x_i} \quad \nupmN  :=  \sum_{i=1}^N \delta_{\sqrt{N} y_i},
$$
and we then define the map
\begin{equation}
\label{def:iN}
\begin{array}{rcl}
\iN : (\R^2)^N \times (\R^2)^N & \longrightarrow & \probas(\Box \times \configs) \\
(\XN, \YN) & \longmapsto & \bar{P}_{(\XN,\YN)}:= \displaystyle \int_{\Box} \delta_{(x, \theta_{\sqrt{N}x} \cdot (\nuppN, \nupmN))} dx
\end{array}
\end{equation}
where 
for any Borel space $X$,
$\probas(X)$ denotes  the set of  Borel probability measures on $X$ and
$\theta_{\lambda}$ denotes the action of translation by a 
vector $\lambda\in \R^2$, that is,  for $\nu\in \probas(\R^2)$, we have
$\theta_\lambda \nu(A)=\nu(A+\lambda)$ for any measurable set $A$.  
 The variable $x$ is a ``tag" that is keeping track of the point $x \in \Box$  
 around which the configuration was blown-up, and this way we build from 
 any signed  point configuration the law of a 
 ``tagged signed point process", $\bar{P}_{(\XN,\YN)}$. 
The laws of these signed point processes will 
easily be shown to be tight, and 
any accumulation point
 as $N \ti$ is a stationary probability measure on $\Lambda \times \config$
 (i.e. the law of a stationary tagged point process)
whose first marginal is the Lebesgue measure on $\Box$.
We will generally denote with bars the quantities corresponding to tagged point processes and without bars the quantitites corresponding to non-tagged point processes.

Throughout we will always consider  the subset 
\[
\left\lbrace \bPsts \in \probas(\Box \times \configs), \bPsts(A \times \configs) = \Leb(A), \forall A \text{ Borel}\right\rbrace,
\]
and continue, with some abuse of notation, to denote it by $\probas(\Box \times \configs)$. This assumption allows us to consider the disintegration probability measures
$\bPstsx \in \probas(\configs)$ for any $x \in \Lambda$.
We denote by $\probinv(\configs)$ the set of stationary laws of signed point processes, and we denote by  $\probinv(\Box \times \configs)$ the set of stationary laws of tagged signed point processes, that is those $\bPsts \in \probas(\Box \times \configs)$ so that the corresponding disintegration measure $\bPstsx$ is stationary for Lebesgue-a.e. $x \in \Box$. Finally we denote by $\probas_{\rm{inv},1}(\Lambda \times \configs)$ the set of $\bPsts \in \probas(\Box \times \configs)$ such that $\bPsts$ has total intensity $1$ (i.e. there is, in average, one point of each sign per unit volume).

In Section \ref{sec:RenormEnerg} we will define an interaction energy functional $\tWs$ on the space $\probinv(\configs)$. It can be understood as the expectation of the infinite-volume limit of the logarithmic interaction in the system of charges described by the signed configurations. We then define the interaction  energy of  $\bPsts \in \probinv(\Box \times \configs)$ as
\begin{equation} \label{def:bWs1}
\bWs(\bPsts) := \int_{\Box} \tWs(\bPstsx) dx.
\end{equation}

Next, we define the \textit{specific relative entropy} of the law of a signed point process as the infinite-volume limit of the usual relative entropy with respect to a reference measure. 
\begin{defi} Let $\Psts \in \probinv(\configs)$. 
The relative specific entropy $\ERS[\Psts]$ with respect to the signed Poisson point process of uniform intensity $1$ is given by
\begin{equation} \label{def:ERS1}
\ERS[\Psts] := \lim_{R \ti} \frac{1}{R^2} \Ent\left( \Psts_{R} | \Poissons_{R} \right), 
\end{equation}
where $\Psts_{R}$ denotes the restriction of $\Psts$ to $\carr_R:= [-R/2,R/2]^2$, and 
\[
\Ent( \mu | \nu) =
\begin{cases} \int \log \frac{d\mu}{d\nu} d\mu & \text{ if $\mu$ is absolutely continuous with respect to $\nu$, } \\
+ \infty & \text{ otherwise}
\end{cases}
\]
 is the usual relative entropy. The reference measure is the law of a signed Poisson point process $$\Poissons := \Poisson^1 \otimes \Poisson^1,$$ which is nothing but the law of two independent Poisson point processes of intensity $1$.
\end{defi}
The  good definition of such an infinite-volume relative entropy $\ERS$ is known in the “non-signed” case where one deals with standard point processes (see e.g. \cite{seppalainen}, and \cite{LebSer} for an extension to the case of tagged point processes), and we recast its properties in the setting of signed point processes in Section \ref{sec:ERS}. We may then define the specific relative entropy of $\bPsts \in \probinv(\Box \times \configs)$ as
\begin{equation} \label{def:bERSa}
\bERS(\bPsts) := \int_{\Box} \ERS[\bPstsx] dx.
\end{equation}

Using \eqref{def:bWs1} and \eqref{def:bERSa}, 
we introduce
the function $\fbarbeta$ defined  on the space 
$\probas_{\rm{inv},1}(\Lambda \times \configs)$, 
\begin{equation}
\label{def:fbarbeta} \fbarbeta(\bPsts) := 
\begin{cases}
\frac{\beta}{2} \bWs(\bPsts)  + \bERS[\bPsts] & \text{ if } \bERS[\bPsts] < + \infty, \\
+ \infty & \text{ otherwise.}
\end{cases}
\end{equation}
We let $\fbsc$ be the lower semi-continuous regularization of $\fbarbeta$ on $\probas_{\rm{inv},1}(\Lambda \times \configs)$ i.e.
\begin{equation} \label{def:fbsc}
\fbsc(\bPsts) := \lim_{\epsilon \t0} \inf_{B(\bPsts, \epsilon)} \fbarbeta.
\end{equation}
In particular it is standard that $\fbarbeta$ and $\fbsc$ have the same infimum on $\probas_{\rm{inv},1}(\Lambda \times \configs)$.
We remark that we do not know whether $\fbarbeta$ is lower semi-continuous, and in particular we do not rule out the possibility that
$\fbarbeta=\fbsc$.

When $\beta < 2$ is fixed, we let $\bPsN$ be the random variable $\iN(\XN, \YN)$ (as in \eqref{def:iN})  when $(\XN, \YN)$ are sampled according to $\PNbeta$, and we let $\bPgotNb$ be its law. In other terms $\bPgotNb$ is the push-forward of $\PNbeta$ by $\iN$.

We may now state our main result.
\begin{theo} \label{theo:LDP}
The sequence $\{\bPgotNb\}_N$ satisfies a Large Deviations Principle at speed $N$ with good rate function given by
$$\fbsc - \inf_{\probas_{\rm{inv},1}(\Lambda \times \configs)} \fbarbeta.$$
\end{theo}
As a consequence we obtain the following expansion for $\log \ZNbeta$ with $\beta < 2$:
\begin{coro} \label{coro:ZN}
For any $\beta < 2$ it holds
\begin{equation} \label{expanZN}
\log \ZNbeta = \frac{\beta}{2} N \log N - \left(\inf_{\probas_{\rm{inv}, 1}(\Lambda, \configs)} \fbarbeta\right) N + o(N),
\end{equation}
where the term $o(N)$ depends on $\beta$.
\end{coro} 
In comparison with Proposition \ref{prop:FGP}, we now have a characterization of the constant $C_{\beta}$ in front of the term of order $N$. We also have
information on the asymptotic behavior of the empirical measures.
\begin{theo} \label{theo:mesure}
The sequence $\{(\muNp, \muNm)\}_N$ converges $\PNbeta$-a.s. to $\Leb_{\Lambda} \otimes \Leb_{\Lambda}$, where $\Leb_{\Lambda}$ is the uniform probability measure on $\Lambda$.
\end{theo}
We emphasize that in the present case, in contrast with the one-component case, the optimal macroscopic distribution of the points cannot be deduced from the leading order behavior of the system. Indeed, a leading order LDP 
(see Section \ref{sec-first-order})
only shows that $\mu_N^+ $ and $\mu_N^-$ have the same limit. The next order analysis of Theorem \ref{theo:LDP} allows to identify {\it at the same time} the macroscopic distribution of the particles, and their microscopic behavior. 

\subsection{Interpretation and method}
\paragraph{\textbf{Comparison with the one-component case.}}
To explain these results and their proof, it is useful to return to the case of the one-component plasma, as was studied in \cite{LebSer} using tools introduced in \cite{SS2d,RougSer,PetSer} (see also \cite{serfatyZur}). 
We recall that in  the one-component plasma, the particles have 
the 
same (positive) sign and are confined by an external potential, which can be shown to act like a neutralizing negative diffuse background charge. 
It is well-known that   the macroscopic distribution of the particles can be identified by a leading order LDP (at speed $N^2$) to be the  so-called {\it equilibrium measure}, uniquely determined by the confining potential (cf. \cite{hiaipetz,bz}). Then the  next order  LDP analysis performed in \cite{LebSer} allows to identify the behavior of the particles at the microscopic scale as minimizing a certain rate function, which is of similar nature to \eqref{def:fbarbeta}.  This analysis relied on expressing the logarithmic interaction energy via the {\it electric field}, or {\it electric potential} that the set of charges generated.  A crucial tool was then the so-called {\it screening result}, which allows, via the electric field formulation,  to localize the interaction energy in microscopic boxes which can then be seen as (essentially) independent and non-interacting.

The main difference between the one-component and two-component cases is that in the one-component case, the interaction energy has a sign and is bounded below, hence no configuration can give a large contribution to the partition function. In the two-component case, the interaction energy is not bounded below, and it is only thanks to the \textit{entropy term}, corresponding to the volume in phase-space, that configurations with very negative energy do not weigh too much in the partition function. Another heuristic way of saying this is that the Lebesgue measure (in phase-space) behaves like a ``Lebesgue repulsion" which prevents particle of opposite signs from getting too close to each other too often, and this is only true because $\beta <2$, i.e. when this Lebesgue repulsion is strong enough. In other words, energy and volume considerations always have to be worked with jointly, and we always need to exploit the fact that $\beta<2$. 

\paragraph{\textbf{Positive part of the energy and dipole contributions.}}
Because the number of positively charged particles and the number of negatively charged particles are the same, it is natural to see a configuration as a set of {\it dipoles} of particles of opposite sign, matched by  nearest neighbor pairing (or minimal matching, also called minimal connection).  
It is tempting to try to prove directly an  LDP  on the pairing, 
but 
we have not been able to do so. 
Instead, we exploit the idea of matching (borrowed 
from \cite{GunPan})
in conjunction with simple computations originating in  \cite{SS2d,RougSer}, which relies on the expression of the interaction energy via the electric potential generated by the system of charges. We rewrite the interaction energy as the sum of a positive part and a part corresponding only to 
 nearest neighbor interactions, which can be thought of as a ``dipole contribution". More precisely for each $x_i$ or $y_i$, we define 
\[
r(x_i)= \min \left( 1 , \hal \min_{j\neq i} |x_i-x_j|, \hal \min_{j} |x_i-y_j|\right)
\]
to be the (half) nearest neighbor distance truncated at $1$. We then prove the identity, valid for every pair of $N$-tuples $\XN, \YN$ 
with distinct coordinates:
\begin{equation} \label{rewritingWpN0}
\WN(\XN, \YN) =  \frac{1}{2\pi} \int_{\R^2} |\nabla \VNr|^2 + \sum_{i=1}^N \log \ro(x_i) + \sum_{i=1}^N \log \ro(y_i) .
\end{equation}
where 
\[\VNr= \log * \left( \sum_{i=1}^N \delta_{x_i}^{(r(x_i))}-\sum_{i=1}^N \delta_{y_i}^{(r(y_i))}\right)
\]
and where $\delta_x^{(\eta)}$ denotes the uniform  measure of mass $1$ on the sphere of center $x$ and radius $\eta$ (for $x \in \R^2$ and $\eta > 0$). 

The identity \eqref{rewritingWpN0} is similar to  the “electrostatic inequality”  found in \cite{GunPan} (however we do not discard the positive term as they do), and it allows us to use the method of \cite{GunPan}, which controls the negative (and possibly unbounded) dipole contributions. The analysis of \cite{GunPan} exploits in a quantitative way the fact that $\beta<2$ and that the Lebesgue repulsion dominates the dipole attraction.  

Similarly, the interaction energy $\bWs$ is the sum of two terms, one positive part corresponding to the large $N$ limit (after blow-up at the scale $\sqrt{N}$) of the positive quantity $ \frac{1}{2\pi} \int_{\R^2} |\nabla \VNr|^2 $, and a negative part corresponding to the large $N$ limit of the nearest-neighbor contributions 
$\sum_{i=1}^N \log \ro(x_i) + \sum_{i=1}^N \log \ro(y_i)$.

\paragraph{\textbf{Interpretation of Theorem \ref{theo:LDP}.}}
The way to read our result at the microscopic scale is to say that the Gibbs measure must concentrate on minimizers of $\fbarbeta$. Unfortunately, we do not know whether a minimizer can be shown to be unique, but in any case it reduces to a minimization problem (with the structure of a \textit{free energy} as in statistical physics)
which should identify some optimal random signed point processes. For comparison, in the two-dimensional one-component case, the analogous result allows to say that the law of the well-known Ginibre point process minimizes the rate function for a certain value of $\beta$. 
We may interpret the minimization of $\fbarbeta$ heuristically as follows. The term $\bWs(\bPsts)$  in \eqref{def:fbarbeta} favors signed configurations which minimize the logarithmic interaction, hence we expect that it favors 
 short dipoles (and as such, is clearly not bounded below).  On
the contrary, the \textit{specific relative entropy} in \eqref{def:fbarbeta} 
favors disorder, and thus tend to ``separate" the dipole points. When $\beta<2$, the sum of the two terms can be shown to be bounded below. The competition between the two terms depends of course on the value of $\beta$. When $\beta $ is 
 small (i.e. the temperature is 
 large) then the entropy term (or ``thermal agitation") dominates, whereas 
as $\beta$ gets larger and approaches $2$ the dipole attraction gets 
stronger, until the system can no longer sustain (or spontaneously generate) 
dipoles.  \\
 
\paragraph{\textbf{Plan of the paper.}}
The rest of the paper is organized as follows: in Section \ref{sec:PrelimNot}, we gather all the definitions and notation that we use, and we present the rewriting of the energy which isolates the dipole contribution. 
In Section \ref{sec:prelimstudy} we use the results of \cite{GunPan} to prove 
that $\fbsc$ is a good rate function and we establish the main result, 
postponing some of the main proofs. In  Section \ref{appGP} we recall, 
for the reader's convenience, the main computations of \cite{GunPan} 
(on which we rely heavily). In Section \ref{secLB} we 
prove the LDP lower bound, and in Section \ref{secUB} we prove the 
LDP upper bound. 
Section \ref{sec-first-order} is devoted to the proof of the large deviations
principle for the empirical measures $(\mu_N^+,\mu_N^-)$,
at the leading order speed
(which is $N^2$).
\\

\paragraph{\textbf{Acknowledgments:}} Part of this project was carried 
out during visits of Thomas Lebl\'{e} and Sylvia Serfaty to the Weizmann Institute.
They would like to warmly thank the institute for its hospitality. 
The work of Sylvia Serfaty was supported
by the Institut Universitaire de France.
The work of Ofer Zeitouni
was supported by an Israel Science Foundation 
grant. Wei Wu thanks Hubert Lacoin and Chuck Newman for useful discussions.

\section{Definitions, notation, and preliminary results}\label{sec:PrelimNot}
\subsection{General notation}
If $X$ is a topological space we denote by $\probas(X)$ the set of Borel probability measures on $X$, and if $P$ is in $\probas(X)$ we denote by $\Esp_{P} \left[ \cdot \right]$ the expectation under $P$. We endow the space $\probas(X)$ of Borel probability measures on $X$ with the Dudley distance:
\begin{equation} \label{DudleyDistance}
d_{\probas(X)}(P_1, P_2) = \sup \left\lbrace \int F (dP_1 - dP_2) |\  F \in \Lip(X) \right\rbrace
\end{equation}
where $\Lip(X)$ denotes the set of functions $F : X \rightarrow \R$ that are $1$-Lipschitz with respect to $d_X$ and such that $\|F\|_{\infty} \leq 1$. It is well-known that the distance $d_{\probas(X)}$ metrizes the topology of weak convergence on $\probas(X)$. We denote by $C^0(X)$ (resp. $C^0_c(X)$) the space of continuous functions (resp. with compact support on $X$), and by $C^0_{b}(X)$ the set of continuous, bounded functions on $X$. 

If $R > 0$ we let $\carr_R := [-R/2, R/2]^2$ be a square of center $0$ and sidelength $R$. If $U$ is a Borel subset of $\R^2$ we denote by $|U|$ its Lebesgue measure (or area). If $(X, d_X)$ is a metric space, $x \in X$ and $r > 0$ we denote by $B(x,r)$ the closed ball of center $x$ and radius $r$ for $d_X$. In $\R^2$ we will use the notation $D(x,r)$ for the closed disk of center $x$ and radius $r$. If $A \subset X$ we denote by $\mathring{A}$ its interior and by $\bar{A}$ its closure.

\paragraph{\textbf{(Signed) point configurations.}}
If $A$ is a Borel set of $\R^2$ we denote by $\config(A)$ the set of locally finite point configurations in $A$ or equivalently the set of non-negative, purely atomic Radon measures on $A$ giving an integer mass to singletons (see \cite{dvj}). We let $\config := \config(\R^2)$. We endow the sets $\config(A)$ (for $A$ Borel) with the topology induced by the topology of weak convergence of Radon measure (also known as vague convergence or convergence against compactly supported continuous functions).

If $B$ is a compact subset of $\R^2$ we endow $\config(B)$ with the following distance: 
\begin{equation} \label{defdistanceconfigB}
d_{\config(B)}(\C_1,\C_2) := \sup \left\lbrace \int F (d\C_1 - d\C_2) |\  F \in \Lip(B) \right\rbrace.
\end{equation}
Similarly we endow $\config := \config(\R^2)$ with the following distance:
\begin{equation} \label{dconfig}
\dconfig(\C_1,\C_2) := \sum_{k \geq 1} \f{1}{2^k} \left( 
\f{d_{\config(\carr_k)}(\C_1,\C_2)}{(\C_1(\carr_k) + 
\C_2 (\carr_k)) \vee 1 }  
  \right).
\end{equation}

We now define the analogue  
of $\config(A)$ and $\config$ in the setting of signed point configuration. 
\begin{defi} \label{def:pointconfig}
A signed point configuration in $\R^2$ (resp. in $A$) 
is defined as an element 
$\Cs=(\Cp,\Cm)$
of $\configs := \config \times \config$ 
(resp. $\configs(A) := \config(A) \times \config(A)$).
We say that $\Cs$ is \textit{simple} if the mass of each singleton is exactly one and the supports of 
$\Cp$ and $\Cm$ are disjoint.
\end{defi}

We endow these product spaces with the product topology and the 
usual $1$-product metric (the sum of distances componentwise). 
We will sometimes abuse notation and write $\int f d\Cs$ as the integral of a test function $f$ against the signed measure $d\Cp - d\Cm$.
For $A \subset \R^2$ a measurable subset we let $|\Cs|(A)$ be the total number of points in $A$ i.e. $|\Cs|(A) := \Cp(A) + \Cm(A)$.

\begin{defi} \label{def:pruning}
We define the ``pruning"
 map $\Pru : \configs \to \configs$ by associating to any 
signed point configuration $\Cs = (\Cp, \Cm)$ the Jordan decomposition 
of the signed measure $\Cp - \Cm$. 
\end{defi}
The effect of $\Pru$ is to 
remove any dipole which would not be felt at the level of the signed measure. 
For example,
if $\Cs = (\delta_{x_0} + \delta_{x_1}, \delta_{x_0} + \delta_{x_2})$ with $x_i
\in \R^2$ distinct,
we have $\Pru(\Cs) = (\delta_{x_1}, \delta_{x_2})$.

The additive group $\R^2$ acts on $\config$ by translations $\{\theta_t\}_{t \in \R^2}$: if $\C = \{x_i, i \in I\} \in \config$ we let 
\begin{equation}\label{actiontrans}
\theta_t \cdot \C := \{x_i - t, i \in I\}.
\end{equation}
We extend this action (while keeping the same notation) to the setting of signed configurations by acting on both components $\Cp$ and $\Cm$ of a given $\Cs \in \configs(\R^2)$.

For any integer $N$ we identify a configuration $\C$ with $N$ points with 
an unordered
$N$-tuples of points in $\R^2$, which we still denote by $\C$.
Denoting by $\pi_N$ the projection from
$(\R^2)^N$ to unordered $N$-tuples in $\R^2$, for a set $A$ of configurations with $N$ points we write
$\Leb^{\otimes N}(A)=\Leb^{\otimes N}(\pi_N^{-1}(A))$.

\paragraph{\textbf{Random tagged signed point configurations.}}
The space $\probas(\Box \times \configs)$ can be viewed as the space of laws of tagged signed point configurations, where we keep as a tag the point $x \in \Box$ around which the signed configuration was blown up. It is equipped with the topology of weak convergence of measures on $\Box \times \configs$. Throughout we will always consider the subset 
\[
\left\lbrace \bPsts \in \probas(\Box \times \configs), \bPsts(A \times \configs) = \Leb(A), \forall A \text{ Borel}\right\rbrace,
\]
and continue, with some abuse of notation, to denote it by $\probas(\Box \times \configs)$. This assumption allows us to consider the disintegration probability measures
$\bPstsx \in \probas(\configs)$ for any $x \in \Lambda$, which satisfy by definition
$$
\int F(x,\Cs) d\bPsts(x,\Cs) = \int_{\Box} \left(\int F(x, \Cs) d\bPstsx(\Cs)\right) dx,
$$
for any function $F \in C^0_{b}(\Lambda \times \configs)$. We refer to \cite[Section 5.3]{AGS} for a proof of the existence of disintegration measures. 

\paragraph\textbf{{Intensity.}}
To any $\Psts \in \probas(\configs)$ we may associate two probability measures $\Pstsp, \Pstsm$  on $\config$ as the push-forwards of $\Psts$ by the two canonical projections of $\configs$ on $\config$, namely $(\Cp, \Cm) \mapsto \Cp$ and $(\Cp, \Cm) \mapsto \Cm$.

Let $\Pst \in \probas(\config)$. If there exists a measurable function $\rho_{1,P}$ such that for any function $\varphi \in C^0_c(\R^2)$ we have
\begin{equation} \label{defintensite}
\Esp_{\Pst}\left[\sum_{x\in \C} \varphi(x)\right] = \int_{\R^2} \rho_{1,P}(x) \varphi(x) dx,
\end{equation}
then we say that $\rho_{1,P}$ is the one-point correlation function (or \textit{intensity}) of $P$. For $m \geq 0$ we say that $P$ is of intensity $m$ when the function $\rho_{1,P}$ of \eqref{defintensite} exists and satisfies $\rho_{1,P} \equiv m$. 

When $\bPst \in \probas(\Lambda \times \config)$ we let $\rho_{\bPst}$ be the intensity measure of $\bPst$ defined by $\rho_{\bPst}(x) := \rho_{1,\bPst^{x}}$. If $\bPsts \in \probas(\Lambda \times \configs)$ we let $\rhop_{\bPsts}$, $\rhom_{\bPsts}$ be the respective intensity measure of $\bPstsp$, $\bPstsm$.
We denote by $\probas_{\rm{inv},1}(\Box \times \configs)$ the set of all $\bPsts \in \probinv(\Box \times \configs)$ such that 
$$ \int_{\Box} \rhop_{\bPsts} = \int_{\Box} \rhom_{\bPsts} = 1.$$

\subsection{Rewriting the interaction energy} Here we adapt computations from \cite{RougSer,PetSer} to rewrite the interaction energy in terms of the electric potential generated by the points, seen as charges (this is also the analogue of the 
``electrostatic inequality''
of \cite{GunPan}). Comparing with \cite{RougSer,PetSer}, instead of using a fixed (small) truncation distance we use  the nearest neighbor distance.
\paragraph{\textbf{Truncation of the logarithmic interaction.}}
Following \cite{PetSer} we define $\delta_p^{(\eta)}$ to be 
the normalized surface measure on $\p D(p, \eta)$ (it coincides with the Dirac mass at $p$ if $\eta=0$). 
We will also need the notion of truncated logarithmic kernel defined for  $1>\eta>0$ and  $x \in \R^2$ by 
\begin{equation}
\label{def:feta} f_\eta(x)= \left(-\log|x| - \log(\eta)\right)_+,
\end{equation} and by $f_\eta\equiv 0$ if $\eta=0$.
We note that the function  $f_\eta$ vanishes outside the disk $D(0, \eta)$ and satisfies that
 \begin{equation}\label{def:de}
 \frac{1}{2\pi} \div (\nabla f_\eta) + \delta_0=\delta_0^{(\eta)}.
\end{equation}

\paragraph{\textbf{Nearest-neighbor distance.}}
If $\XN, \YN$ are two $N$-tuples of points, we define the nearest-neighbor (half-)distance as
 \begin{equation} \label{def:rofini}
 \ro(x_i) =  \min\( 1,\hal \min_{j\neq i} |x_i-x_j|, \hal \min_{j} |x_i-y_j|\),
 \end{equation}
 for any $i = 1, \dots, N$, and similarly for $\ro(y_i)$.
 
\paragraph{\textbf{Rewriting of the energy functional.}}
Let $\XN, \YN$ be two $N$-tuples of points in $\Box$. 
 We let $\VNr$ be the electric potential generated by $\mc C$,
after ``smearing out" 
the charges on a distance $\ro$.
More precisely,
\begin{equation} \label{def:HN}
\VNr := 2\pi (-\Delta)^{-1} \left(\sum_{i=1}^N \delta^{(\ro(x_i))}_{x_i} - \sum_{i=1}^N \delta^{(\ro(y_i))}_{y_i}\right),
\end{equation}
where $2\pi (-\Delta)^{-1}$ is the convolution by $\log$. By the properties of $\delta^{(\eta)}_p$, this is the same as setting 
\begin{equation} \label{2edefh}
\VNr (x) = \sum_{i=1}^N \left( - \log |x-x_i| - f_{\ro(x_i)}(x-x_i) \right)  +\sum_{i=1}^N \left( \log |x-y_i|+ f_{\ro(y_i)} (x-y_i)\right).
\end{equation}
We also write
	\begin{equation}
		\label{def:VN0}
		V_{N,0}:=\sum_{i=1}^N 
		\left( - \log |x-x_i| +
	 \log |x-y_i|\right).
	\end{equation}

The next crucial identity expresses the fact that the interaction energy $\WN(\XN,\YN)$ can be computed using the electric potential $\VNr$ (more precisely its gradient, the truncated electric field) even if the interaction has been truncated at distance $\ro$.
\begin{lem} \label{lem:WIPP}
Let $\XN, \YN$ be such that the associated global point configuration 
$\mc C$ is simple.
Then,
\begin{equation} \label{rewritingWpN}
\WN(\XN, \YN) =  \frac{1}{2\pi} \int_{\R^2} |\nabla \VNr|^2 + \sum_{i=1}^N \log \ro(x_i) + \sum_{i=1}^N \log \ro(y_i) .
\end{equation}
\end{lem}
\begin{proof}This can be seen as a simple application of Newton's theorem.
Since the system is globally neutral (there are $N$ positive charges and $N$ negative charges), the electric potential $\VNr$ decays like $1/|x|$ as $|x| \ti$, and $\nab \VNr$ decays like $1/|x|^2$. 
We may thus integrate by parts and find that 
the boundary terms tend to zero, and therefore,
using \eqref{2edefh}, 
we obtain
\begin{multline*}
\frac{1}{2\pi} \int_{\R^2} |\nabla \VNr|^2  = \int_{\R^2} -\frac{1}{2\pi} \Delta \VNr \VNr 
= \sum_{i=1}^N \int_{\R^2}
 \VNr \left(\delta_{x_i}^{(\ro(x_i))} -  \delta_{y_i}^{(\ro(y_i))} \right)
   \\ = \sum_{i,j=1}^N \left(-\log |x-x_j|-f_{\ro(x_j)} (x-x_j) +\log |x-y_j|+f_{\ro(y_j)} (x-y_j)\right) \left(\delta_{x_i}^{( \ro(x_i))} -  \delta_{y_i}^{(\ro(y_i))}\right).
   \end{multline*}
Next, we use the fact that the disks $D(x_i, \ro(x_i))$ and 
$D(y_i, \ro(y_i))$ are disjoint by the
definition of $\ro$, and that for any $p, \eta$, the measure $\delta_p^{(\eta)}$ is supported on $\p D(p, \eta)$ where $f_\eta(x-p)$ vanishes, to obtain 
   \begin{multline}
\frac{1}{2\pi} \int_{\R^2} |\nabla \VNr|^2  = 
 \sum_{i\neq j} \left( - \log |x-x_j|\delta_{x_i}^{(\ro(x_i))} -\log|x-y_j|\delta_{y_i}^{(\ro(y_i))} \right) 
 \\+ \sum_{i,j} \left( \log |x-x_j|\delta_{y_i}^{(\ro(y_i))}  +\log |x-y_j| \delta_{x_i}^{(\ro(x_i))} \right)
   -\sum_{i=1}^N \left(\log \ro(x_i)+\log \ro(y_i)\right).
   \end{multline}
In addition, by Newton's theorem (or by the mean-value property), the average of $\log |x-p|$ over any disk $D(q,\eta)$ not containing $p$ is $\log |x-q|$. Since $\delta_p^{(\eta)}$ is precisely the uniform measure of mass $1$ on $\p D(p, \eta)$, and using again the fact that the 
disks
$D(x_i, \ro(x_i))$ and $D(y_i, \ro(y_i))$ are disjoint, we conclude that 
    \begin{multline}
\frac{1}{2\pi} \int_{\R^2} |\nabla \VNr|^2  = 
 \sum_{i\neq j}\left(- \log |x_i-x_j| -\log|y_i-y_j|  \right)
 \\+\sum_{i,j} \left( \log |y_i-x_j|  +\log |x_i-y_j|  \right)
   -\sum_{i=1}^N \left( \log \ro(x_i)+\log \ro(y_i) \right),
   \end{multline}
which is the desired result.
\end{proof}

\subsection{Blow-up coordinates} \label{sec:blowup}
In view of Lemma \ref{lem:WIPP}, 
and since the finite point configurations $\XN, \YN$ are simple $\PNbeta$-a.s.,
we may rewrite the Gibbs measure as the probability measure whose density with respect to the Lebesgue measure on $\Box^{2N}$ is given by
\begin{equation}\label{rwP}
\frac{1}{\ZNbeta} \exp\left(-\frac{\beta}{2}\left( \frac{1}{2\pi} \int_{\R^2} |\nab \VNr|^2  +\sum_{i=1}^N \log \ro(x_i)+ \log \ro(y_i) \right) \right).
\end{equation}
We rescale the finite configurations by a factor $\sqrt{N}$ and use a prime symbol for the new quantities. In particular we let 
$\XpN = (x'_1, \dots, x'_N)$ with $x'_i = \sqrt{N} x_i$ (and the same for $\YpN$), and we have of course $\ro(x'_i) = \sqrt{N} \ro(x_i)$. We let $\VpN$ be the electric potential generated by the rescaled point configuration after truncation $$
\VpN := \VNr\left(\frac{\cdot}{\sqrt{N}}\right),
$$
and 
\begin{equation}\label{defvn}\VN:= V_{N,0} 
\left(\frac{\cdot}{\sqrt{N}}\right).\end{equation}
We have, by 
a change of variables
$$\int_{\R^2} |\nab \VpN|^2= \int_{\R^2} |\nab \VNr|^2,$$ whereas the nearest-neigbor distance term scales as
$$ \sum_{i=1}^N \log \ro(x_i)+ \log \ro(y_i) = \sum_{i=1}^N \log \ro'(x_i)+ \log \ro'(y_i) + N \log N.$$
Combining these identities with \eqref{expZN1} and \eqref{rwP} we may write the Gibbs measure as the probability measure whose density with respect to the Lebesgue measure on $\Box^{2N}$ is given by
\begin{equation} \label{rwP2}
\frac{1}{\KNbeta} \exp \left(-\frac{\beta}{2}\left( \frac{1}{2\pi} \int_{\R^2} |\nab \VpN|^2  +\sum_{i=1}^N \log \ro(x'_i)+ \log \ro(y'_i) \right)\right),
\end{equation}
where the new normalizing constant $\KNbeta$ satisfies $\log \KNbeta = C_{\beta}N + o(N)$,
 with $C_\beta$ as in \eqref{expZN1}.

The computations in the last two subsections serve as a motivation for the upcoming definition of the interaction energy for the infinite configurations which arise after taking the limit $N \to \infty$.  We note in particular that $\VNr$ solves 
\[
-\Delta \VNr= 2\pi \left( \sum_{i=1}^N \delta_{x_i}^{(r(x_i))}- \sum_{i=1}^N \delta_{y_i}^{(r(y_i))}\right),
\]
a relation which will pass to the limit (in the sense of distributions)  as $N \to \infty$. The electric field associated to the potential $\VNr$ is $\nab \VNr$ and its divergence is $\Delta \VNr$. 

\subsection{Interaction energy for signed configurations}
\label{sec:RenormEnerg}
\paragraph{\textbf{Electric fields and electric processes.}}
We may thus define the class $\Elec$ of “electric” vector fields on $\R^2$ to be the set of vector fields $E$ which belong to $\Lploc(\R^2, \R^2)$ for some $p < 2$ and satisfy
\begin{equation}\label{eqclam}
- \div E = 2\pi  (\Cp-\Cm) \quad  \text{ in } \R^2
\end{equation}
for some signed point configuration $\Cs = (\Cp, \Cm)$. When $E$ satisfies \eqref{eqclam} we write $E~\in~\Elec(\Cs)$ and say that $E$ is \textit{compatible} with $\Cs$. We note that two elements of $\Elec(\Cs)$ differ by a divergence-free vector field.
Unless stated otherwise, we endow the space $\Lploc(\R^2, \R^2)$ with the 
weak topology.  If $E \in \Elec$ we let $\mathrm{Conf}({E})$ be the underlying signed point configuration i.e. the signed point configuration $\Cs = (\Cp, \Cm)$ where $(\Cp, \Cm)$ is the Jordan decomposition of $\frac{-1}{2\pi} \div E$. 
In particular, if $\Cs$ is a signed point configuration with
$\C = \Pru(\C)$ (the latter is implied if $\Cs$ is simple),
 and such that $E \in \Elec(\C)$, then $\Cs = \mathrm{Conf}(E)$.

We define an electric field law as an element of $\probas(\Lploc(\R^{2}, \R^{2}))$ where $p<2$, concentrated on $\Elec$. It will usually be denoted by $\Pelec$.  We say that $\Pelec$ is stationary when it is invariant under the (push-forward by) translations $\theta_x \cdot E = E(\cdot - x)$ for any $x \in \R^2$. A tagged electric field law is an element of $\probas(\Box \times \Lploc(\R^{2}, \R^{2}))$ whose first marginal is the normalized Lebesgue measure on $\Lambda$ and whose disintegration slices are electric field laws. It will be denoted by $\bPelec$. We say that $\bPelec$ is stationary if for a.e.~$z \in \Box$, the disintegration measure $\bar{P}^{\mathrm{elec},z}$ is stationary.

\paragraph{\textbf{Nearest-neighbor distance and truncation.}}
If $\Cs = (\Cp, \Cm)$ is a signed point configuration we define the nearest-neighbor (half-)distance as
 \begin{equation} \label{def:roinfini}
 \ro(x) =  \min\(1 ,\hal \min_{x' \in \Cp, x' \neq x} |x - x'|, \hal \min_{y \in \Cm} |x-y|\),
 \end{equation}
for any $x \in \Cp$ and we define similarly $\ro(y)$ for $y \in \Cm$.

For any electric field $E$ we define its truncation
\begin{equation}\label{defeeta}
\Er :=  E - \sum_{p \in \Cp} \nabla f_{\ro(p)}(x-p) + \sum_{p \in \Cm} \nabla f_{ \ro(p)}(x-p),
\end{equation}
where $\C = \mathrm{Conf}(E)$ and $\ro$ is defined above, 
computed with respect to $\mathrm{Conf}(E)$. This is the ``infinite configuration" equivalent of the truncated electric field $\nabla \VNr$ as in \eqref{rewritingWpN}. We may now define the interaction energy of an admissible electric field in a similar fashion as what arises in Lemma \ref{lem:WIPP}.

A control on the $L^2$-norm of $E_r$ can be 
translated into a bound of the $L^p$-norm on $E$ as follows.
\begin{lem} \label{lem:L2Lp}
Let $\C$ be a 
point configuration, let $E \in \Elec(\C)$ and let $R > 0$. We have
\begin{equation}
\label{L2Lp}
\|E\|_{L^p(C_R)} \leq  L_1 \|E_r\|_{L^2(C_R)} + L_2 |\C|(C_{R+1})
\end{equation}
for some $L_1 > 0$ depending only on $R$ and $p$ and some  universal constant $L_2$.
\end{lem}
\begin{proof}
From the definition \eqref{defeeta} we get
\[
|E| \leq  |E_r| + \sum_{q \in \Cp} |\nabla f_{\ro(q)}(x-q)| + \sum_{q \in \Cm} |\nabla f_{ \ro(q)}(x-q)|,
\]
and \eqref{L2Lp} follows by using the triangle inequality 
for the $L^p$-norm, 
 H\"{o}lder's inequality which yields $\|E_r\|_{L^p(C_R)} \leq 
L_1 \|E_r\|_{L^2(C_R)}$ for some 
{$L_1$} depending on $R$ and $p$, and by observing that the terms $\|\nabla f_{\ro(q)}\|_{L^p}$ are uniformly bounded by some $L_2 > 0$ and that 
the number of such terms is
bounded by 
the total number of points of $\C$ in $C_{R+1}$.
\end{proof}

\subsubsection{Positive part of the energy}
For any $E \in \Elec$ we define $\cW(E)$ as
\begin{equation}
\label{def:cW}
\cW(E) := \limsup_{R \to \infty} \frac{1}{R^2} \int_{\carr_R} |\Er|^2,
\end{equation} 
and we call it the ``positive part" of the energy of $E$. 
Recall that the truncation $E_r$ of $E$ at nearest-neighbor distance is 
defined with respect to the ``minimal" underlying signed point 
configuration $\mathrm{Conf}(E)$.

Next, if $\Cs$ is a signed point configuration we let $\dW(\Cs)$ be the infimum of $\cW(E)$ over the electric fields $E$ compatible with $\Cs$
\begin{equation}
\label{def:dW}
\dW(\Cs) := \inf \left\lbrace \cW(E), E \in \Elec(\Cs) \right\rbrace.
\end{equation}
We emphasize that the definition of $\dW(\Cs)$ proceeds by 
considering the energy of associated electric fields, and 
that the truncation of an electric field is defined with respect to the 
``minimal" underlying signed point configuration $\mathrm{Conf}(E)$. 
In particular if $\Pru(\C_1) = \Pru(\C_2)$ then $\dW(\C_1) = \dW(\C_2) = 
\dW(\Pru(\C_1))$, where $\Pru$ is the pruning map 
of Definition \ref{def:pruning}.

\subsubsection{Negative part of the energy} 
If $\chi : \R^2 \to \R$ is a 	nonnegative measurable function with compact support and $\Cs$ a signed point configuration we let
\begin{align*}
\Deta(\chi, \Cs) & := - \int \chi(x) \log(\ro(x)) (d\Cp + d\Cm)(x)\\
\Deta(\Cs) & := \limsup_{R \ti} \frac{1}{R^2} \Deta(\1_{\carr_R}, \Cs).
\end{align*}
Similarly, for any $0 < \tau < 1$, we let
\begin{align*}
\Detat(\chi, \Cs) & := - \int \chi \log(\ro(x) \vee \tau) (d\Cp + d\Cm)(x)\\
\Detat(\Cs) & := \limsup_{R \ti} \frac{1}{R^2} \Detat(\1_{\carr_R}, \Cs).
\end{align*}
The function $\Deta$ is a non-negative quantity, corresponding to the expected ``dipole contribution" to the energy. We will call $-\Deta$ the negative part of the energy. It can be obtained as the monotone limit of $\Detat(\Cs).$

We could now try to define the interaction energy of a signed point configuration as the difference $\dW(\Cs) - \Deta(\Cs)$. However it turns out that for good definition it is preferable to consider such an object at the level of signed point processes, as we do below.  

\subsubsection{A compatibility lemma}

\begin{lem} \label{lem:compatibility} Let $\{E_N\}_N$ be a sequence of electric fields, let $E \in \Elec$ and let $\C \in \configs$. Assume that
\begin{enumerate}
\item $\{E^{(N)}\}_N$ converges to $E$ weakly in $\Lploc(\R^2, \R^2)$.
\item $\{\mathrm{Conf}(E^{(N)})\}_N$ converges to $\C$ in $\configs$.
\end{enumerate}
Then $\mathrm{Conf}(E) = \Pru(\C)$, and in particular $E$ is an electric field.
Moreover,
\begin{equation} \label{lsciENr}
\liminf_{N \ti} \int \chi |E^{(N)}_r|^2 \geq \int \chi |E_r|^2
\end{equation}
 for any smooth, compactly supported, nonnegative function $\chi$.
\end{lem}

\begin{proof}
Let $B_R$ be the ball of radius $R$,  and let $f \in C^0_b(W^{-1,p}(B_R))$. 
One may check that $(C_0^0(B_R))^* $ embeds continuously into $W^{-1,p}(B_R)$ (indeed,  $W^{-1,p}(B_R)$ is the  dual of the Sobolev space $W^{1,q}_0(B_R)$ where $q$ is the conjugate exponent to $p$, and the latter embeds  continuously into $C_0^0(B_R)$ since $q>2$). It thus follows that $f$ is also bounded and continuous on $(C_0^0(B_R))^*.$
The convergence of $\mathrm{Conf}(E^{(N)})$ to $\C$ and the fact that $\C$ is locally finite thus  imply that
\begin{equation} \label{convergencecontreffort} 
 \lim_{N \ti} f\left(-\frac{1}{2\pi} \div E^{(N)}\right) =  f(\Cp - \Cm), \quad \forall f \in C^0_b(W^{-1,p}(B_R)).
\end{equation}

Since the function $E \mapsto -\frac{1}{2\pi} \div E$ is continuous from $\Lploc(B_R)$ to $W^{-1,p}(B_R)$, we may  use the first assumption to get
\begin{equation}\label{convergencecontreffort2} 
\lim_{N \ti} f\left(-\frac{1}{2\pi} \div E^{(N)}\right)  =  f\left(-\frac{1}{2\pi} \div E\right).
\end{equation}Since this is true for all $f \in C^0_b(W^{-1,p}(B_R))$ and for all $R>1$, we conclude that 
 $-\frac{1}{2\pi} \div E = \Cp - \Cm$ and thus $\mathrm{Conf}(E) = \Pru(\C)$.

We next 
prove \eqref{lsciENr}. We may assume that the left-hand side is finite, 
otherwise there is nothing to prove. This implies  that, up to an
extraction of a subsequence, $\sqrt{\chi} E^{(N)}_r$  converges weakly in $L^2(\R^2, \R^2)$ to  some vector field, which we claim can only be $\sqrt{\chi} E_{r'}$ where $r'$ denotes the nearest-neighbour distance computed in $\C$ (and not in $\mathrm{Conf}(E)$).  Indeed, it suffices to check that $\sqrt{\chi} E^{(N)}_r$ converges to $\sqrt{\chi} E_{r'}$ weakly in $L^p$.  In view of the first assumption and of \eqref{defeeta}, it suffices to check that 
$$\sum_{p \in \C^{(N),+}} \nabla f_{r(p)} (x-p)- \sum_{q \in \C^{(N),-}}  \nabla f_{r(q)}(x-q)
\rightharpoonup \sum_{p \in \Cp} \nabla f_{r(p)} (x-p) - \sum_{q \in \Cm}  \nabla f_{r(q)}(x-q)$$ weakly in $L^p_{\mathrm{loc}}(\R^2)$, where $\C^{(N)} = (\C^{(N), +}, \C^{(N),-}) := \mathrm{Conf}(E^{(N)})$. But from \eqref{def:feta}, we have $\nab f_\eta(x-p)=-\frac{x-p}{|x-p|^2}\indic_{|x-p|<\eta}$ (and $0$ if $\eta=0$), which is continuous  in $L^p$ with respect to both $p$ and $\eta$.  
So,
the stated convergence follows 
from the second assumption and the definition of $r'$,  using also the fact that the point configurations  are locally finite. We conclude that $\sqrt{\chi} E^{(N)}_r$ converges weakly in $L^2(\R^2)$ as claimed, 
so by the lower semi-continuity of the $L^2$ norm,
$$
\liminf_{N \ti} \int \chi |E^{(N)}_r|^2 \geq \int \chi |E_{r'}|^2.$$
We may finally observe that since $r' < r$  (the nearest neighbor distances are smaller in $\C$ than in $\Pru(\C)$) we have $|E_{r'}|^2 \geq |E_r|^2$ pointwise, which concludes the proof. 
\end{proof}
At the level of laws of electric fields, the result of Lemma \ref{lem:compatibility} translates into
\begin{lem} \label{lem:compatibility2}
Let $\{P_N\}_{N}$ be a sequence of random signed point processes and $\{\Pelec_N\}_{N}$ be a sequence of laws of electric fields. Let $P$ be  a random signed point process and $\Pelec$ be a law of electric fields. Let us assume that:
\begin{enumerate}
\item For any $N \geq 1$, the push-forward of $\Pelec_N$ by $\conf$ coincides with $P_N$.
\item The sequence $\{P_N\}_{N}$ converges to $P$ as $N \ti$ in $\probas(\configs)$.
\item The sequence $\{\Pelec_N\}_{N}$ converges to $\Pelec$ as $N \ti$ in $\probas(\Lploc(\R^2, \R^2))$.
\end{enumerate}
Then we have
\begin{enumerate}
\item The push-forward of $\Pelec$ by $E \mapsto -\frac{1}{2\pi} \div E$ is concentrated on signed measures of the type $\Cp - \Cm$ for some $(\Cp, \Cm)$ in $\configs$. In particular $\Pelec$ is concentrated on $\Elec$.
\item The push-forward of $\Pelec$ by $\conf$ coincides with the push-forward of $P$ by $\Pru$. 
\end{enumerate}
Moreover for any smooth, compactly supported, nonnegative function $\chi$ we have
\begin{equation} \label{lsciENr1}
\liminf_{N \ti} \Esp_{\Pelec_N} \left[ \int \chi |E_r|^2\right] \geq \Esp_{\Pelec} \left[\int \chi |E_r|^2\right].
\end{equation}

\end{lem}

\subsection{Process level energy} \label{sec:energypp}
\paragraph{\textbf{Energy of signed point processes.}}
Let $\Pst$ (resp. $\bPst$) be the law of a signed point process (resp. of a tagged signed point process). We define 
\begin{align*}
& \tdW(\Psts) := \Esp_{\Psts} \left[ \dW(\Cs) \right], \quad \bdW(\bPsts) := \Esp_{\bPsts} \left[ \dW(\Cs) \right] \text{ (positive part of the energy)}\\
& \tDeta(\Psts) := \Esp_{\Psts} \left[  \Deta(\Cs) \right], \quad \bDeta(\bPsts) := \Esp_{\bPsts} \left[  \Deta(\Cs) \right] \text{ (contribution of dipoles)} \\
& \tDetat(\Psts) := \Esp_{\Psts} \left[ \Detat(\Cs) \right], 
\quad \bDetat(\bPsts) := \Esp_{\bPsts} \left[  \Detat(\Cs) \right] 
\text{ (dipoles at truncated
distance $\geq \tau$)}.
\end{align*}

Finally we define the interaction energy of $\Psts$ (resp. $\bPsts$) 
\begin{equation} \label{def:tWS}
\tWs(\Psts) := \tdW(\Psts) - \tDeta(\Psts) \ \text{ and }\  \bWs(\bPsts) := \bdW(\bPsts) - \bDeta(\bPsts).
\end{equation}
It is yet unclear whether the right-hand sides in \eqref{def:tWS} have an actual meaning, because it could be the difference of two infinite quantities. However we will see in Section \ref{sec:prelimstudy} that for a certain class of point processes, which are the only candidates for describing the microscopic behavior, the quantity $\tDeta(\bPsts)$ is in fact finite.

\paragraph{\textbf{Stationary lifting with minimal energy.}}
The following useful lemma shows that we may associate to any stationary tagged point process the law of a \textit{stationary} tagged electric field which is compatible with it and whose energy is minimal.
\begin{lem} \label{lem:compatible}
Assume $\bPsts \in \probinv(\Lambda \times \configs)$ is such that 
$\bdW(\bPsts)$ is finite, and such that signed point configuration 
are $\bPsts$-a.s. simple.
Then there exists a law $\bPelec$ of tagged electric fields such that 
\begin{enumerate}
\item The push-forward of $\bPelec$ by $(z,E) \mapsto (z, \mathrm{Conf}({E}))$ is $\bPsts$.
\item  We have
\begin{equation} \label{compatible}
\bdW(\bPsts) = \Esp_{\bPelec} \left[\cW\right].
\end{equation}
\end{enumerate}
\end{lem}
\begin{proof}
The proof of Lemma \ref{lem:compatible} is very similar to that
 of \cite[Lemma 2.12]{LebSer} and we only sketch it here. For
 simplicity we present the argument in the non-tagged case, 
the extension to random tagged signed point processes being straightforward.
Let $\epsilon > 0$ be fixed. For any simple signed point configuration $\Cs$ 
such that $\dW(\Cs)$ is finite, by definition of $\dW(\Cs)$ we 
may find an electric field $E^{(\C, \ep)}$ 
such that $\cW(E^{(\C, \ep)}) \leq \dW(\Cs) + \epsilon$. For any $k \geq 1$ we let
\[
P^{\rm{elec}}_{(\C, \epsilon, k)} := \frac{1}{|C_k|} \int_{C_k} 
\delta_{\theta_x \cdot E^{(\C, \ep)}}\ dx,
\]
which is an electric field law (as defined in Section \ref{sec:RenormEnerg}). For any $m \geq 1$ we have (using Fubini's theorem)
\[
\int_{C_m} |E_r|^2 dP^{\rm{elec}}_{(\C, \epsilon, k)} (E) \leq \frac{1}{|C_k|} \int_{C_{m+k}} |E^{(\C, \ep)}_r|^2,  
\]
and the right-hand side is bounded as $k \ti$ by the
finiteness of $\cW(E^{(\C, \ep)})$. Using Lemma~\ref{lem:L2Lp} 
and a standard compactness result in $L^p$-spaces, it follows that the sequence $\{P^{\rm{elec}}_{(\C, \epsilon, k)}\}_{k}$ is tight in $\probas(\Lploc(\R^2, \R^2))$ (for the weak $\Lploc$ topology). We let $P^{\rm{elec}}_{(\C, \epsilon)}$ be a limit point as $k \ti$.
Set
\[
P_{(\C, k)} := \frac{1}{|C_k|} \int_{C_k} 
\delta_{\theta_x \cdot \C}\ dx
\]
Since the signed point configurations are simple $P$-a.s. we see that $P_{(\C, k)}$ is the push-forward of $P^{\rm{elec}}_{(\C, \epsilon, k)}$ by $\conf$ for $P$-a.e. $\C$. Moreover the ergodic theorem implies that for $P$-a.e. $\C$, the sequence $\{P_{(\C, k)}\}_{k}$ converges to a stationary signed point process $P_{\C}$. 
By Lemma \ref{lem:compatibility} we conclude
that the push-forward of $P^{\rm{elec}}_{(\C, \epsilon)}$ by $\conf$ coincides with the push-forward of $P_{\C}$ by $\Pru$ and that the energy is lower semi-continuous in the following sense
\begin{multline} \label{lsciL22}
\Esp_{P^{\rm{elec}}_{(\C, \epsilon)}} \left[ \frac{1}{|C_m|} \int_{C_m} |E_r|^2 \right] \leq \liminf_{k \ti} \Esp_{P^{\rm{elec}}_{(\C, \epsilon, k)}} \left[ \frac{1}{|C_m|} \int_{C_m} |E_r|^2 \right]  \\ \leq \liminf_{k \ti}  \frac{1}{|C_k|} \int_{C_{m+k}} |E^{(\C, \ep)}_r|^2  \leq \cW(E^{(\C, \ep)}).
\end{multline}
Define next $P^{\rm{elec}}_{\epsilon} := \int P^{\rm{elec}}_{(\C, \epsilon)} dP(\C)$, or by duality
\[
\int f dP^{\rm{elec}}_{\epsilon} := \int \left(\int f dP^{\rm{elec}}_{(\C, \epsilon)}\right) dP(\C),
\]
for any test function $f \in C^0(\Lploc(\R^2, \R^2))$. It is not difficult to check that the electric field law 
$P^{\rm{elec}}_{\epsilon}$ is stationary and that its push-forward by $\conf$ is 
$P$ (because $\C = \Pr(\C)$, $P$-a.s.). Moreover we get from \eqref{lsciL22}, for any $m \geq 1$,
\[
\Esp_{P^{\rm{elec}}_{\epsilon}} \left[\frac{1}{|C_m|} \int_{C_m} |E_r|^2 \right] \leq  \Esp_{P}\left[\cW(E^{(\C, \ep)})\right] \leq \dW(P) + \epsilon.
\]
Letting $\epsilon \t0$, we may thus find a limit point $\Pelec$ of $\{P^{\rm{elec}}_{\epsilon}\}$ that is still compatible with $P$ and such that $\Esp_{\Pelec} [\cW] \leq \dW(P)$. The converse inequality is always true by 
 definition of $\dW$.
\end{proof}

We also obtain the following useful lower semi-continuity property.
\begin{lem} \label{lem:LSCIWo}
The map $\bPsts \mapsto  \bdW(\bPsts)$ is lower semi-continuous on $\probinv(\Lambda \times \configs)$.
\end{lem}
\begin{proof}
Let $\{\bPsts_k\}_{k}$ be a sequence in $\probinv(\Lambda \times \configs)$ converging to some stationary $\bPsts$, and such that $\liminf_{k \ti} \bdW(\bPsts_k)$ is finite. Up to the extraction of a subsequence,
we may assume that the $\liminf$ is a limit. For any $k \geq 1$ we may apply Lemma \ref{lem:compatible} and obtain a stationary electric field law $\bPelec_k$ such that \eqref{compatible} holds. Using the stationarity of $\bPelec_k$ and Fubini's theorem we may write, for any $k$
\[
\Esp_{\bPelec_k} \left[\cW(E)\right] = \Esp_{\bPelec_k} \left[ \int_{C_1} |E_r|^2 \right],
\]
and in fact the left-hand side is equal to $\Esp_{\bPelec_k} \left[ \frac{1}{R^2} \int_{C_R} |E_r|^2 \right]$ for any $R > 0$. 
The sequence of the push-forwards of $\bPelec_k$ by $E \mapsto E_r$ is thus tight in $\probas(L^2_{\rm{loc}}(\R^2, \R^2))$, and using Lemma \ref{lem:L2Lp} we see that $\{\bPelec_k\}_k$ itself is tight in $\probas(\Lploc(\R^2, \R^2))$. 
Using Lemma~\ref{lem:compatibility2} we see that any limit 
point $\bPelec$ is compatible with $\bar{P}$ and that
\[
\liminf_{k \ti} \Esp_{\bPelec_k} \left[ \int_{C_1} |E_r|^2 \right] \geq \Esp_{\bPelec} \left[ \int_{C_1} |E_r|^2 \right], 
\]
and using again the stationarity we see that the right-hand side satisfies
\[
\Esp_{\bPelec} \left[ \int_{C_1} |E_r|^2 \right] = \Esp_{\bPelec}[\cW(E)] \geq \bdW(\bPsts),
\]
which concludes the proof.
\end{proof}

\subsection{Specific relative entropy and large deviations}
\label{sec:ERS}
Let $\Poisson^1$ denote the law of the Poisson point process of intensity $1$ on $\R^2$, and let $\Poissons := \Poisson^1 \otimes \Poisson^1$.
\paragraph{\textbf{Existence and main properties.}}
The following proposition is an adaptation of classical results concerning the existence and properties of the so-called specific relative entropy for stationary point processes. 
\begin{prop} \label{prop:ERS} Let $\Psts \in \probinv(\configs)$.
The following limit exists in $[0, + \infty]$
\begin{equation} \label{def:ERS2}
\ERS[\Psts] := \lim_{R \ti} \frac{1}{R^2} \Ent\left( \Psts_{R} | \Poissons_{R} \right), 
\end{equation}
moreover the functional $\Psts \mapsto \ERS[\Psts]$ is affine and lower semi-continuous on $\probinv(\configs)$.
\end{prop}
\begin{proof}
The proofs of the corresponding results in the non-signed setting extend readily to our context, see e.g. \cite[Section 7.2]{seppalainen}.
\end{proof}

\paragraph{\textbf{Large deviations for signed empirical fields.}}
Let $\bQpN$ be the “reference” empirical field defined as the push-forward by $\iN$ of the Lebesgue measure $\Leb^N \otimes \Leb^N$ on $\Box^{N} \times \Box^{N}$, where we recall that $i_N$ was defined in \eqref{def:iN}.
\begin{prop} \label{prop:Sanov}
For any $A \subset \probinv(\Lambda \times \configs)$, we have (with $\bERS$ defined in \eqref{def:bERSa})
\begin{multline} \label{EqSB}
- \inf_{\bPsts \in \mathring{A} \cap \probas_{\rm{inv},1}} \bERS[\bPsts] \leq \liminf_{N \ti} \f{1}{N} \log \bQpN (A) \\ \leq  \limsup_{N \ti} \f{1}{N} \log \bQpN(A) \leq - \inf_{\bPsts \in \bar{A}} \bERS[\bPsts].
\end{multline}
\end{prop}
The proof of Proposition \ref{prop:Sanov} is almost identical to that of \cite[Proposition 1.6]{LebSer}, see \cite[Section 7]{LebSer}.

\subsection{Tightness and discrepancy estimates}
\paragraph{\textbf{Compactness and exponential tightness.}}
\begin{lem} \label{lem:expotight} The sequence $\{\bPgotNb\}_N$ is exponentially tight.
\end{lem}
\begin{proof}
Let $\Nn_R : 
 (\Box \times \config) \to \R_+$ be the map 
$\Nn_R(x, \Cs) := \Cp(D(0,R)) + \Cm(D(0,R))$ 
which gives the total number of points in the disk $D(0,R)$ of a
 signed point configuration $\Cs = (\Cp, \Cm)$. 
(Although it may seem not to depend on $x$, in applications we
will always use $\Nn_R(x,\Cs_x)$ where $\Cs_x$ is the blow-up 
around $x$ of a signed point configuration.)
By construction it is clear that $\bPgotNb$ is concentrated on 
\[
\left\lbrace \bPsN \in \probas(\Lambda \times \configs), \Esp_{\bPsN}[\Nn_R] \leq 2\pi R^2 \right\rbrace.
\]
This set is easily seen to be compact in $\probas(\Lambda \times \configs)$, see e.g. \cite[Lemma 4.1]{LebSer}.
\end{proof}

\paragraph{\textbf{Discrepancy estimate, equality of intensities.}}
Here we prove that we may control the difference between the number of positive and negative charges in terms of the two-component interaction energy of the signed point configuration. In particular, we show that if $\Psts$ is stationary, the finiteness of $\fbarbeta(\Psts)$ implies that the intensities of positive and negative charges coincide.

\begin{lem} \label{lem:discr} Let $\Psts \in \probinv(\configs)$ be such that $\bdW(\Psts)$ and $\ERS[\Psts]$ are finite. Then we have $\rhop_{\Psts} = \rhom_{\Psts}$.
\end{lem}
\begin{proof} In this proof $C$ denotes a constant depending only on $\Psts$. Let $\Cs$ be a signed point configuration and let $\D_R$ be the discrepancy in the square $C_R$, i.e. the difference between the number of positive and negative charges in $C_R$
$$
\D_R := \int_{C_R} (d\Cp - d\Cm).
$$
Assume that $E$ is an electric field compatible with $\Cs$. Using the relation \eqref{defeeta} and \eqref{eqclam} and  integrating by parts over $C_R$, we have 
$$
\int_{\partial C_R} E_r \cdot \vec{n} = 2\pi \left(\D_R + d_R\right),
$$ with $\vec{n}$ the outer unit normal, 
where the error term $d_R$ is bounded by the number of points of $\Cs$ in a $1$-neighborhood of $\partial C_R$. Let $\psi : \R_+
 \to \R$ be a map such that $\psi(x) = x \log \log x$ for $x$ large and such that $\psi$ is convex, nonnegative, nondecreasing. We have
$$
\psi(|\D_R|) = \psi\left( 
|\frac{1}{2\pi}  \left(\int_{\partial C_R} E_r \cdot \vec{n} \right) - d_R|
 \right) \leq C \psi\left(\left|\int_{\partial C_R} E_r \cdot \vec{n}\right|\right) + C \psi(|d_R|) + C.
$$
Using Jensen's inequality and the stationarity of $\Psts$ we get
\begin{equation*}
\Esp_{\Psts} \left[ \psi\left(\int_{\partial C_R} |E_r|\right) \right] \leq \Esp_{\Psts} \left[\psi\left(4R\ |E_r|\right)\right].
\end{equation*}
By stationarity of $\Psts$, $\Esp_{\Psts}\left[|E_r|^2\right]$ is equal to the positive part energy of $\Psts$, which is finite by assumption. We may thus bound
\begin{equation} \label{discr2}
\Esp_{\Psts} \left[\psi\left(4R\ |E_r|\right)\right] \leq CR\log \log R + CR + C.
\end{equation}
On the other hand, again by stationarity of $\Psts$ and using the convexity of $\psi$ we have
\begin{equation*}
\Esp_{\Psts}\left(\psi\left(|d_R|\right)\right) \leq C R \Esp_{\Psts} \left( d_1 \log \log (Rd_1)\right)
\leq CR \log \log R\,
 \Esp_{\Psts} [ d_1] + CR \,\Esp_{\Psts}[d_1 \log \log d_1].
\end{equation*}
The exponential moments of $d_1$ and $d_1 \log \log d_1$ under a
 Poisson point process are finite, and $\Psts$ has finite entropy, 
hence both expectations under $\Psts$ are finite. We obtain
\begin{equation} \label{discr3}
\Esp_{\Psts}\left(\psi(|d_R|)\right) \leq CR\log \log R + CR + C\,,
\end{equation}
where $C$ depends on the entropy of $P$.
Combining \eqref{discr2} and \eqref{discr3} we get
$\Esp_{\Psts} \left[ \psi(|\D_R|) \right] = o(\psi(R^2)) \text{ as } R \ti$. It implies by Jensen's inequality 
(and the fact that $\psi(x) = x \log \log x$ for $x$ large) 
that $\Esp_{\Psts}[ |\D_R|] = o(R^2)$, but since 
$\Psts$ is stationary we have $\Esp_{\Psts} 
[|\D_R|] = R^2 \Esp_{\Psts} 
[|\D_1|]$, hence we deduce that for all $R > 0$ we have
 $\Esp_{\Psts} 
[|\D_R|] = 0$. In particular the mean density of 
positive and negative charges are equal. 
\end{proof}

\subsection{Scaling relations and optimal intensity}
Let $\Psts$ be a stationary signed point process such that the intensity of $\Pstsp$ and $\Pstsm$ are both equal to $\rho$.
Let $\sigma_{\rho}(\Psts)$ be the stationary point processes obtained as the push-forward of $\Psts$ by $\Cs \mapsto \sqrt{\rho}
\Cs$. It is easy to see that both component of $\sigma_{\rho}\, (\Psts)$ now have intensity $1$.

\begin{lem} \label{lem:scalinﬁ} We have the following scaling relations
\begin{align*}
\tdW(\Psts) & = \rho(\tdW(\sigma_{\rho}(\Psts)))\\
\tDeta(\Psts) & = \rho(\tDeta(\sigma_{\rho}(\Psts)) -  \log \rho)\\
\ERS[\Psts] & = \rho\, \ERS[\sigma_{\rho}(\Psts)] + 1 - \rho + \rho \log \rho
\end{align*}
\end{lem}
\begin{proof}
The first two relations are deduced by a change of variable 
in the definitions, see also the scaling relations in 
\cite[Equation (2.4)]{SS2d}. 
The scaling relation for the entropy is proven in a elementary way as 
in \cite[Lemma 4.2]{LebSer}.
\end{proof}

\begin{lem} \label{lem:minimmemerho}
Let $\bPsts \in \probas_{\rm{inv}, 1}(\Lambda \times \configs)$ be a stationary tagged signed point process such that the intensity measures $\rho_{\bPstsp}$ and $\rho_{\bPstsm}$ are equal to the same function $\rho$ on $\Lambda$ with  $\int_{\Lambda} \rho = 1$.
Then 
$$
\fbarbeta(\bPsts) \geq \inf_{\probas_{\rm{inv}, 1}} \fbarbeta,
$$
with equality only if $\rho(x) = 1$ for Lebesgue-a.e. $x \in \Lambda$.
\end{lem}
\begin{proof}
From Lemma \ref{lem:scalinﬁ} we deduce a scaling relation for $\fbarbeta$:
$$
\fbarbeta(\bPsts) = \frac{\beta}{2} \int_{\Lambda} \rho(x) 
\tWs(\sigma_{\rho(x)}(\bPstsx))
dx + \int_{\Lambda} \rho(x) 
\ERS[\sigma_{\rho(x)}(\bPstsx)] dx + \int_{\Lambda} 
\left[(1-\frac{\beta}{2}) \rho \log \rho - \rho + 1\right]dx.
$$
In particular we have
$$
\fbarbeta(\bPsts) \geq \inf_{\probas_{\rm{inv}, 1}} \fbarbeta + (1-\frac{\beta}{2}) \int_{\Lambda} \rho \log \rho.
$$
The total intensity being fixed, since $\frac{\beta}{2} < 1$ the expression above is minimized only if $\rho = 1$ Lebesgue-a.e. and we get
$$
\fbarbeta(\bPsts) \geq \inf_{\probas_{\rm{inv}, 1}} \fbarbeta.
$$ 
\end{proof}

\section{Study of the Gibbs measure and main conclusions} \label{sec:prelimstudy}

\subsection{Bounds on the partition function}
We rely on the work of Gunson-Panta \cite{GunPan} which gives a 
``classical" (in contrast to the Quantum Field Theory techniques 
of \cite{Frohlich}) approach to the study of the Gibbs measure $\PNbeta$ and 
of the partition function $\ZNbeta$. For the reader's reference, a 
rewriting of their results can be found in Section \ref{appGP} below.

In this subsection
 we mostly re-phrase the key points of their analysis 
in  our notation.

\paragraph{\textbf{Dipole contribution.}}
The analysis in \cite{GunPan}, 
recalled in Section \ref{appGP}, see
\eqref{ZNannex} and \eqref{ZNUB}, yields the following lemma.
\begin{lem}
For any integer $N$ and any $\beta < 2$ we have
\begin{equation} \label{GPDg}
\log \int_{\Lambda^{2N}} \exp\left(- \frac{\beta}{2} \left( \sum_{i=1}^N \log  \ro(x_i) + \sum_{i=1}^N \log  \ro(y_i)  \right)\right) d\XN d\YN \leq \frac{\beta}{2} N \log N + C_{\beta} N,
\end{equation}
with a constant $C_{\beta}$ depending only on $\beta$. \end{lem}

\paragraph{\textbf{Exponential moments.}}
We 
give another consequence of the analysis in \cite{GunPan}. For any 
pair
 of integers $\Np, \Nm$, and real
$R > 0$ we denote by $\B_{\Np, \Nm, R}$ the law of the signed point process 
on $\carr_R$ obtained 
from
 two independent Bernoulli processes with $\Np$ and $\Nm$ points. The following lemma gives a bound on (the exponential moments of) the dipole contribution  in the interaction energy.

\begin{lem} For any $\beta < 2$ and any $R > 0$ we have
\begin{multline} \label{partitionD}
\log \Esp_{\B_{\Np, \Nm, R}}\left[ e^{\frac{\beta}{2}\Deta(\1_{\carr_R}, \Cs)} \right] \leq \frac{\beta}{4} (\Np+\Nm) \log (\Np + \Nm) +  (\Np + \Nm) C_{\beta} \\ - \frac{\beta}{2} (\Np+\Nm)  \log R.
\end{multline}
\end{lem}
\begin{proof}
Scaling the configuration by a factor $R^{-1}$ changes the left-hand side by $\frac{\beta}{2} (\Np+\Nm) \log R$ and then we are left to prove the inequality for $R =1$. With our notation, it reduces to the upper bound on \cite[(2.4)]{GunPan} as expressed in \cite[(2.9)]{GunPan} (cf. \eqref{ZNUB} and \eqref{ZNannex}). Let us emphasize that although the analysis of \cite{GunPan} initially deals with a system such that $\Np = \Nm = N$, the bound on \cite[(2.4)]{GunPan} is not affected by the actual sign of each charge, as it is merely a bound on some given integral on $(\R^2)^{2N}$. We may thus follow the lines of \cite[Section 2.2.]{GunPan} with $\Np + \Nm$ instead of $2N$ and \cite[(2.9)]{GunPan} yields \eqref{partitionD}.
\end{proof}

\subsection{Study of the rate function}
In this subsection we show that $\fbarbeta$ is bounded below and is 
well-defined as a functional  
 $\fbarbeta:\probinv(\Lambda \times \config) \to \R \cup \{+ \infty\}$.

\begin{lem} \label{lem:minoRF}
For any $\beta < 2$, any $\tau, R > 0$ and any $\Psts \in \probinv(\configs)$ such that $\ERS[\Psts]$ is finite, it holds that
\begin{equation} \label{minoRFfix}
- \frac{\beta}{2} \Esp_{\Psts}\left[ \Detat(\1_{\carr_R}, \Cs) \right] + \Ent[\Psts_{R} | \Poissons_{R}] \geq - L_{\beta} R^2
\end{equation}
where $L_{\beta}$ is a constant depending only on $\beta$. Consequently we get as $\tau \t0$
\begin{equation}\label{minoRFfix2}
- \frac{\beta}{2} \Esp_{\Psts}\left[ \Deta(\1_{\carr_R}, \Cs) \right] + \Ent[\Psts_{R} | \Poissons_{R}]  \geq - L_{\beta} R^2
\end{equation}
and finally in the limit $R \to + \infty$,
\begin{equation}\label{minoRFfix3}
- \Esp_{\Psts}\left[\frac{\beta}{2} \Deta(\Cs)\right] + \ERS[\Psts] \geq - L_{\beta}.
\end{equation}
\end{lem}
\begin{proof} By the variational characterization of the relative entropy, we know that
\begin{equation} \label{caracVar}
- \log \Esp_{\Poissons_{R}} \left[e^{\frac{\beta}{2} \Detat(\1_{\carr_R}, \Cs)}\right] \leq \Ent[\Psts_{R} | \Poissons_{R}] - \Esp_{\Psts}\left[\frac{\beta}{2} \Detat(\1_{\carr_R}, \Cs) \right] .
\end{equation}
 We 
 evaluate the left-hand side of \eqref{caracVar}. Combining the definition of a Poisson point process with Lemma \ref{partitionD} we get
\begin{multline*}
\log \Esp_{\Poissons_{R}} \left[e^{\frac{\beta}{2} \Detat(\1_{\carr_R}, \Cs)}\right] = \log \sum_{\Np, \Nm = 0}^{+ \infty} \Poissons_{R}(\Np, \Nm) \Esp_{\B_{\Np, \Nm, R}} \left[e^{\frac{\beta}{2} \Detat(\1_{\carr_R}, \Cs)}\right] \\
\le  \log \sum_{\Np, \Nm = 0}^{+ \infty} e^{-2R^2} \frac{R^{2(\Np + \Nm)}}{\Np! \Nm!} e^{\frac{\beta}{4} (\Np+\Nm) \log (\Np + \Nm) +  (\Np + \Nm) C_{\beta} -\frac{\beta}{2} (\Np+\Nm)  \log R}.
\end{multline*}
Using the elementary inequality $$(\Np+\Nm) \log (\Np + \Nm) \leq (\Np \log \Np + \Nm \log \Nm + \Np + \Nm),$$ we may separate the variables $\Np$ and $\Nm$ (which play a symmetric role) and write, using the fact that $\frac{1}{N!} \leq e^{-N\log N + (C+1)N}$ for a certain constant $C$,
\begin{multline*}
\log \Esp_{\Poissons_{R}} \left[e^{\frac{\beta}{2} \Detat(\1_{\carr_R}, \C)}\right] \leq 2 \log \sum_{N=0}^{+\infty} e^{-R^2} \frac{R^{2N}}{N!} e^{\frac{\beta}{4} (N\log N + N) + N (C_{\beta}+C) - \frac{\beta}{2} N \log R} \\
\leq 2 \log \sum_{N=0}^{+\infty} e^{-R^2} R^{2(1-\frac{\beta}{4}) N} e^{(\frac{\beta}{4}-1) N \log N + (C_{\beta}+C+1)N} \leq L_{\beta} R^{2}
\end{multline*}
for a certain constant $L_{\beta}$ depending only on $\beta$. Inserting this estimate into \eqref{caracVar} yields \eqref{minoRFfix}, \eqref{minoRFfix2} follows by sending $\tau \t0$ and \eqref{minoRFfix3} is obtained by dividing \eqref{minoRFfix2} by $R^2$ and then sending $R \to + \infty$ (together with the definition \eqref{def:ERS1} of $\ERS$).
\end{proof}
In particular if $\bERS[\bPsts]$ is finite, then for Lebesgue-a.e. $x \in \Lambda$ the disintegration measure $\bPstsx$ has finite entropy and satisfies \eqref{minoRFfix3}, hence the functional $\fbarbeta$ is well-defined.

\paragraph{\textbf{Conclusion.}}
\begin{lem} The functional $\fbsc$ is a good rate function.
\end{lem}
\begin{proof}
Lemma \ref{lem:minoRF} shows that $\fbarbeta$ is well-defined as a functional $\fbarbeta : \probinv(\Lambda \times \configs) \to \R \cup \{+ \infty\}$. It also implies that the sub-level sets of $\fbarbeta$ are included in sub-level sets of $\bPsts \mapsto \bERS[\bPsts]$, which are compact.

Thus $\fbsc$ is well-defined and it is lower semi-continuous by definition. Since the sub-level sets of $\fbarbeta$ are pre-compact, those of its lower semi-continuous regularization are compact. It proves that $\fbsc$ is a good rate function.
\end{proof}

\subsection{Properties of the limit objects}
One of the crucial points in order to get Theorem \ref{theo:mesure} 
from Theorem \ref{theo:LDP} is to show, by entropy arguments, that the 
intensities
 of the underlying point processes coincide 
with the limits of the empirical measures, while  by the scaling argument of Lemma \ref{lem:minimmemerho} the rate function is minimized only when these intensities equal $1$. 

In this subsection we use the preliminary bounds on $\ZNbeta$ available for $\beta < 2$ thanks to the analysis of Gunson-Panta to derive some \textit{a priori} properties of the possible limits of $\muNp$ and $\bPsN$. In particular we wish to show that the intensity of the limits of  $\bPsN$ equals, most of the time,  the density of the limits of $\muNp$. 
This is not obvious since, with the topology that we use for the 
convergence of  $\bPsN$, there can be a 
loss of mass
 when taking the limit. 

To overcome this, we will  show below that with overwhelming probability (i.e. up to neglecting events of $\PNbeta$-probability less than $e^{-NT}$ with $T$ arbitrarily large) the limiting objects must have finite entropy, which will yield a uniform integrability of the densities of points, 
which in turn ensures that no loss of mass occurs in the limit.

\subsubsection{A priori bounds on the entropy}
\begin{lem} \label{lem:aprioribounds}
For any $\beta < 2$, the following holds with a constant $C_{\beta}$ depending only on $\beta$. 
\begin{enumerate}
\item For any $\mup \in \probas(\Lambda)$ we have
\begin{equation} \label{borneentmac}
\lim_{\epsilon \t0} \limsup_{N \ti} \frac{1}{N} \log \PNbeta \left( \muNp \in B(\mup, \epsilon) \right) \leq C_{\beta} - \frac{1}{C_{\beta}} \Ent[\mup|\Leb_{\Lambda}].
\end{equation}
\item For any $\bPsts \in \probinv(\Lambda \times \configs)$ we have
\begin{equation} \label{borneentmic1}
\lim_{\epsilon \t0} \limsup_{N \ti} \frac{1}{N} \log \bPgotNb \left( B(\bPsts, \epsilon) \right) \leq C_{\beta} - \frac{1}{C_{\beta}} \bERS[\bPsts].
\end{equation}
\item For any $R, N$ let $\{C_i\}_{i \in I}$ be a partition of $\Lambda$ by squares of sidelength in $(\frac{R}{2\sqrt{N}}, \frac{3R}{2\sqrt{N}})$, and let $n_i = \muNp(C_i)$ be the number of positive charges in the square $C_i$. When $R > 0$ is fixed we have
\begin{equation} \label{controleNbpoints}
\limsup_{N \ti} \frac{1}{N} \log \PNbeta \left( 
\left( \frac{1}{\# I} \sum_{i \in I} \frac{n_i}{R^2} \max\left(1, 
\left(\log  \frac{n_i}{R^2}\right)^{\hal}\right)  \right) \geq M \right) 
\leq f_{\beta}(M ,R), 
\end{equation}
with $\lim f_{\beta}(M ,R) = - \infty$ as $M \ti$.
\end{enumerate}
\end{lem}

\begin{proof} For any $\beta < 2$, let us fix some $\pb > 1$ such that $\pb \beta < 2$ and let $\qb$ be the conjugate exponent of $\pb$.

 Let $\mu \in \probas(\Lambda)$. 
Let $A\subset \probas(\Lambda)$ be measurable.
For any $N$ 
we obtain using H\"{o}lder's inequality that
\begin{multline} \label{Holdermacro}
\PNbeta \left(\Lambda^{2N} \cap \{\muNp \in 
 A \}\right) = 
\frac{1}{\ZNbeta} \int_{\Lambda^{2N} \cap \{\muNp \in 
A\}} e^{-\frac{\beta}{2} \WN(\XN, \YN)} d\XN d\YN \\ 
\leq \frac{1}{\ZNbeta}  \left(\int_{\Lambda^{2N}} e^{-\pb\frac{\beta}{2}  
\WN(\XN, \YN)} d\XN d\YN  \right)^{\frac{1}{\pb}} 
\left(\int_{\Lambda^{2N} \cap \{\muNp \in 
A\}}  d\XN d\YN \right)^{\frac{1}{\qb}} \\ = 
\frac{Z_{N,\pb\beta}^{\frac{1}{\pb}}}{\ZNbeta} \left( 
\int_{\Lambda^{2N} \cap \{\muNp \in A\}}  
d\XN d\YN \right)^{\frac{1}{\qb}} \end{multline}
where $\pb$, $\qb$ are as above. By \eqref{expZN1} we have
\begin{equation} \label{GPpbeta}
\log Z_{N,\pb \beta} = \pb \frac{\beta}{2} N \log N + C_{p\beta}N + o(N) \text{ and } \log Z_{N,\beta} = \frac{\beta}{2} N \log N + C_{\beta}N + o(N)
\end{equation}
where $C_{\beta}, C_{\pb \beta}$ depend only on $\beta$. 
On the other hand we have by Sanov's theorem
\begin{equation} \label{applicationSanov}
\lim_{\epsilon \t0} \limsup_{N \ti} \frac{1}{N} \log \int_{\Lambda^{2N} \cap \{\muNp \in B(\mup, \epsilon)\}}  d\XN 
d\YN= - \Ent[\mu^+|\Leb_{\Lambda}].
\end{equation}
Combining \eqref{Holdermacro}
(with $A=B(\mup, \epsilon)$), 
\eqref{GPpbeta} and \eqref{applicationSanov} yields \eqref{borneentmac}. The proof of \eqref{borneentmic1} is similar, using Proposition \ref{prop:Sanov} instead of Sanov's theorem in the last step, where \eqref{applicationSanov} is replaced by
$$
\lim_{\epsilon \t0} \limsup_{N \ti} \frac{1}{N} \log \int_{\Lambda^{2N} \cap \bPN \in B(\bPsts, \epsilon)}  d\XN d\YN = - \bERS[\bPsts].
$$

To see \eqref{controleNbpoints}, we take $A=A(M,R)$ in 
\eqref{Holdermacro} as the event inside the 
probability in \eqref{controleNbpoints}. Using \eqref{GPpbeta},
the proof of \eqref{controleNbpoints} reduces to proving
\begin{equation} \label{Point3}
\limsup_{N \ti} \frac{1}{N} \log \int_{\Lambda^{2N} \cap A}  d\XN 
d\YN\leq f_\beta(M ,R).
\end{equation}
The proof of \eqref{Point3} is simplified if one uses 
comparison to a Poisson process of intensity $N$ on $\Lambda$, 
denoted $\Poisson_N$. Indeed, extend the event $A$ in a natural way
to apply to any collection of integers $\{n_i\}_{i\in I}$, and note that $A$ is monotone increasing with respect to $K=\sum_{i\in I} n_i$. Since 
there exists a constant $\eta>0$ independent of $N$ so that 
$\Poisson_N(K\geq N)>\eta$, and since conditioned on $K$ the 
points of the Poisson process are independent and uniformly distributed 
in $\Lambda$,
the proof of \eqref{Point3} reduces to proving that
\begin{equation} \label{Point3-a}
\limsup_{N \ti} \frac{1}{N} \log \Poisson_N(A)
\leq f_\beta(M ,R).
\end{equation}
The advantage of working with $\Poisson_N$ is that the random variables $n_i$ are now independent Poisson of parameter in $(R^2/4,9R^2/4)$.
In particular, the random variables $n_i \max(1,(\log (n_i/R^2))^{1/2})$
possess a finite exponential moment. Applying Markov's exponential
inequality
then yields \eqref{Point3-a} and completes the proof of
 \eqref{controleNbpoints}.
\end{proof}
Of course \eqref{borneentmac} and \eqref{controleNbpoints}  also hold  when replacing $\muNp$ by $\muNm$.

\subsubsection{Uniform integrability of the number of points}
The bound \eqref{controleNbpoints} implies that under $\PNbeta$, the random number of points $\muNp \left(B(x,\frac{R}{\sqrt{N}})\right)$ is uniformly (as $N \ti$) integrable on $\Lambda$ with overwhelming probablity. More precisely we have
\begin{lem} \label{lem:UI}
For any $T,R > 0$ and any $\epsilon > 0$ there exists $M' > 0$ (depending on $T, \epsilon$ and on $\beta$) such that for $N$ large enough we have
$$
\int_{\Lambda}
\muNp \left(B(x,\frac{R}{\sqrt{N}})\right) 
{\bf 1}_{\{
\muNp \left(B(x,\frac{R}{\sqrt{N}})\right) \geq M'\}} dx
\leq \epsilon 
$$
with probability $\geq 1 - \exp(-NT)$ under $\PNbeta$.
\end{lem}
\begin{proof}
Indeed from \eqref{controleNbpoints} we control the $L^1(\Lambda)$ norm of the superlinear map $\psi_R$ defined as 
$$\psi_R(x) := \frac{x}{R^2} \max\left(1, \left(\log  \frac{x}{R^2}\right)^{\hal}\right)$$
by $M$ with $\PNbeta$-probability $\geq 1 - \exp(Nf_{\beta}(M))$ with $\lim_{M \ti} f_{\beta}(M) = - \infty$.
\end{proof}

\subsubsection{Microscopic intensity \textit{versus} macroscopic density}
We emphasize the following abuse of notation: in Lemma \ref{lem:micromacro} and its proof, the quantities $\muNp$ and $\bPsN$ are elements of a \textit{deterministic} sequence.
\begin{lem} \label{lem:micromacro} 
Let $\{\XN, \YN\}_N$ be a sequence of points in $\Lambda^{2N}$,  let $\muNp := \frac{1}{N} \sum_{i=1}^N \delta_{x_i}$ and let $\bPsN := \iN(\XN, \YN)$. Assume that up to extraction the sequence $\{(\muNp, \bPsN)\}_N$ converges to $(\mup,\bPsts)$ where $\mup \in \probas(\Lambda)$ and $\bPsts \in \probinv(\Lambda \times \configs)$. Then we have $\rhop_{\bPsts} \leq \mup$ in the sense of nonnegative measures.

Moreover under the assumption that $\mup$ does not charge $\partial \Lambda$ and that for any $R>1$, $x \mapsto \muNp \left(B(x,\frac{R}{\sqrt{N}})\right)$ is uniformly integrable on $\Box$ as $N \ti$ then $\rhop_{\bPsts} = \mup$. 
\end{lem}
Of course the same results hold for the quantities associated to the negative charges as well.
\begin{proof}
Let $\chi$ be a non-negative test function in $C^0(\Lambda)$ and for any $R > 1$ let $f_R$ be a smooth function satisfying
$$
\frac{1_{\carr_{R-1}}}{|\carr_R|} \leq f_R \leq \frac{1_{\carr_{R}}}{|\carr_R|},
$$
we also define $f_{R,N}$ as $f_{R,N}(x) := N f_{R}(\sqrt{N}x)$ for $x \in \Lambda$. Finally we define $\lc f_R, \Cp \rc$ as
\[
\lc f_R, \Cp \rc := \int f_R d\Cp \text{ for any } \Cp \in \config(\R^2).
\]
We compute (with $*$ being the convolution product)
\begin{multline} \label{muverrho} \int \chi * f_{R,N} d\muNp =
 \int \chi f_{R,N} * d\muNp = \int \chi(y) \int f_{R,N}(y-z) 
d\muNp (z)  dy \\
= \int \chi(y) \sum_{i=1}^N f_R(\sqrt{N} (x_i -y)) dy \geq \int \chi(y) \lc f_R, \Cp \rc d\bPN(y, \Cs),
\end{multline}
where the inequality is due to the possible loss of mass at the boundary. 
For any $M > 0$ we may write
\begin{multline*}
\int \chi(y) \lc f_R, \Cp \rc d\bPN(y, \Cs) \geq \int \chi(y) \left(\lc f_R, \Cp \rc \wedge M \right) d\bPN(y, \Cs) \\ \underset{N \ti}{\longrightarrow} \int \chi(y) \left(\lc f_R, \Cp \rc \wedge M\right) d\bPsts(y, \Cs) 
\end{multline*}
and we have, by definition of $f_R$ 
$$ \int \chi(y) \left(\lc f_R, \Cp \rc \wedge M\right) d\bPsts(y, \Cs) \geq  \frac{1}{R^2} \int \chi(y) \left(\Nn(0,R-1)(\Cp) \wedge M\right) d\bPsts(y, \Cs).
$$
By definition of the intensity it holds that
\[
\lim_{R \ti} \lim_{M \ti} \frac{1}{R^2} \int \chi(y) \left(\Nn(0,R-1)(\Cp) \wedge M\right) d\bPsts(y, \Cs) = \int \chi(y) \rhop_{\bPsts}(y).
\]
Moreover, since  $\muNp$ converges to $\mup$ we have 
$$\lim_{N \ti} \int \chi * f_{R,N} d\muNp = \int \chi d\mup + o_{R}(1).$$
Finally, sending $R \ti, M \ti, N \ti$ we get
$$
\int \chi d\mup  \geq \int \chi(y) \rhop_{\bPsts}(y)
$$
for any non-negative continuous test function $\chi$, which proves $\rhop_{\bPsts} \leq \mup$.

We next prove the equality under the additional assumption that $\mup$ does not charge $\partial \Lambda$ and that that $x \mapsto \muNp \left(B(x,\frac{R}{\sqrt{N}})\right)$ is uniformly integrable on $\Box$ as $N \ti$. First, the difference between the last two terms in \eqref{muverrho} is bounded as follows:
\begin{multline*}
\int \chi(y) \sum_{i=1}^N f_R(\sqrt{N} (x_i -y)) dy \\ 
\leq 
\int \chi(y) 
 \lc f_R, \Cp \rc d\bPN(y, \Cs) 
  + ||\chi||_{\infty} \mup_{N} \left( \lbrace x \in \Lambda, \dist(x, \partial \Lambda) \leq 2 \frac{R}{\sqrt{N}} \rbrace \right).
\end{multline*}
Since $\muNp$ converges to $\mup$ which does not charge the boundary, the error term satisfies
\[
 ||\chi||_{\infty} \mup \left( \lbrace x \in \Lambda, \dist(x, \partial \Lambda) \leq 2 \frac{R}{\sqrt{N}} \rbrace \right) = o(1)
\]
as $N \ti$, for any $\chi$ and $R$ fixed.

Moreover, the uniform integrability assumption implies that $\Cs \mapsto \lc f_R, \Cp \rc$ is uniformly integrable against $d\bPN$ as $N \ti$ and we may for any $\delta > 0$ choose $M$ large enough such that 
$$
\int \chi(y) \lc f_R, \Cp \rc d\bPN(y, \Cs) \leq \int \chi(y) \left(\lc f_R, \Cp \rc \wedge M \right) d\bPN(y, \Cs) + \delta
$$
uniformly in $N$. Arguing as above we see that 
$$
\lim_{R \ti} \lim_{N \ti} \int \chi(y) \left(\lc f_R, \Cp \rc \wedge M \right) d\bPN(y, \Cs) \leq \int \chi(y) \rhop_{\bPsts}(y) \leq \int \chi(y) \rhop_{\bPsts}(y).
$$
Eventually we get $\mup \leq \rhop_{\bPsts} + \delta$ and we conclude by letting $\delta \t0$.
\end{proof}

\subsubsection{Total intensity of the limit random point process}
From the previous lemmas we deduce that in the LDP we may restrict ourselves to random point processes with total intensity $1$.
\begin{lem} \label{lem:bonneintensite}
Let $\bPsts$ be the law of a stationary tagged signed point process such that the intensity of positive charges satisfies $\int_{\Box} \rhop_{\bPsts} < 1$. Then we have
\begin{equation} \label{bonneintensite}
\lim_{\epsilon \t0} \limsup_{N \ti} \frac{1}{N} \log \bPgotNb(B(\bPsts, \epsilon)) = - \infty.
\end{equation}
\end{lem}
\begin{proof}
Assume that \eqref{bonneintensite} does not hold and that we have for some $T$
\[
\lim_{\epsilon \t0} \limsup_{N \ti} \frac{1}{N} \log \bPgotNb(B(\bPsts, \epsilon)) \geq - T.
\]
Using the relative compactness of $\iN(\Lambda^{2N})$ we may find a sequence $\{\XN\}_N$ of points in $\Lambda^{2N}$ such that $\iN(\XN)$ converges to some $\bQ \in B(\bPsts, \epsilon)$. Up to extraction we may also assume that $\muNp$ converges to $\mup \in \probas(\Lambda)$ and the point \eqref{borneentmac} of Lemma \ref{lem:aprioribounds} ensures that we may assume that $\mup$ has finite entropy, hence does not charge the boundary $\partial \Lambda$. Then, using Lemma \ref{lem:UI} and Lemma \ref{lem:micromacro} we obtain that $\rhop_{\bQ} = \mup$ and in particular $\rhop_{\bQ}$ has total mass $1$. Thus in any ball $B(\bPsts, \epsilon)$ we may find a random tagged point process $\bQ_{\epsilon}$ such that $\rhop_{\bQ_{\epsilon}}$ has total mass $1$. Moreover, again by Lemma \ref{lem:UI}, we may assume that the number of points in any disk is uniformly integrable under $\bQ_{\epsilon}$ as $\epsilon \t0$. Passing to the limit $\epsilon \t0$, it implies that $\rhop_{P}$ has total mass $1$, which yields a contradiction.
\end{proof}

\subsection{Conclusion}
We now show how Theorems \ref{theo:LDP} and \ref{theo:mesure}  follow once we have proven the following lower and upper bounds:

\begin{prop} \label{prop:LDPLB} Let $\bPsts \in \probas_{\rm{inv},1}(\Lambda \times \configs)$. We have
\begin{multline} \label{LDPLB}
\lim_{\epsilon \t0} \liminf_{N \ti} \frac{1}{N} \log \int_{\iN^{-1} (B(P, \epsilon))} e^{-\frac{\beta}{2}\left( \frac{1}{2\pi} \int_{\R^2} |\nab \VpN|^2  +\sum_{i=1}^N \log \ro(x'_i)+ \log \ro(y'_i) \right)} d\XN d\YN \\
\geq - \fbarbeta(\bPsts). 
\end{multline}
\end{prop}

\begin{prop} \label{prop:LDPUB} Let $\bPsts \in \probas_{\rm{inv},1}(\Lambda \times \configs)$. We have
\begin{multline} \label{LDPUB}
\lim_{\epsilon \t0} \limsup_{N \ti} \frac{1}{N} \log \int_{\iN^{-1} (B(P, \epsilon))} e^{ \left(-\frac{\beta}{2}\left( \frac{1}{2\pi} \int_{\R^2} |\nab \VpN|^2  +\sum_{i=1}^N \log \ro(x'_i)+ \log \ro(y'_i) \right)\right)} d\XN d\YN \\ \leq - \fbarbeta(\bPsts). 
\end{multline}
\end{prop}

\subsubsection{Proof of Theorem \ref{theo:LDP}}
Since we have exponential tightness, the proof of Theorem \ref{theo:LDP} reduces to proving a weak LDP. Thanks to Lemmas \ref{lem:minimmemerho}   and \ref{lem:bonneintensite}, the latter is easily deduced by combining Proposition \ref{prop:LDPLB} and Proposition \ref{prop:LDPUB}. Indeed in view of \eqref{rwP2} it only remains to show that $\log \KNbeta = - \inf \fbarbeta + o(N)$, but combining the upper and lower bounds with the exponential tightness (and Lemma \ref{lem:bonneintensite}) we have
\begin{equation} \label{convergenceKNbeta}
- \inf_{\probas_{\rm{inv},1}} \fbarbeta \leq \liminf_{N \ti} \frac{1}{N} \log \KNbeta
\leq \limsup_{N \ti} \frac{1}{N} \log \KNbeta
\leq - \inf_{\probas_{\rm{inv},1}} \fbarbeta.
\end{equation}
Hence 
$\lim_{N \ti} \frac{1}{N} \log \KNbeta = 
- \inf_{\probas_{\rm{inv},1}} \fbarbeta$, which concludes the proof of Theorem \ref{theo:LDP}.

We also get Corollary \ref{coro:ZN} from \eqref{convergenceKNbeta} and the fact that $\log \KNbeta + \frac{\beta}{2} N \log N = \log \ZNbeta$ as seen in Section \ref{sec:blowup}.

\subsubsection{Proof of Theorem \ref{theo:mesure}}
From Lemmas \ref{lem:discr} and \ref{lem:minimmemerho} we see that minimizers of $\fbarbeta$ are such that the intensity  of both components of $\bPstsx$ are equal to $1$ (for Lebesgue-a.e. $x \in \Lambda$). Thus the limit points of $\{\bPsN\}_N$ have $\PNbeta$-a.s. both intensity measures equal to the uniform measure on $\Lambda$, and by Lemma \ref{lem:micromacro} we see that any limit point of $\{\muNp, \muNm\}_N$ must be the uniform measure on $\Lambda$, almost surely under $\PNbeta$.

The rest of the paper is organized as follows: in Section \ref{appGP} we give for the reader's convenience the proof of the main result of \cite{GunPan}, in Section \ref{secLB} we prove Proposition \ref{prop:LDPLB}, and in Section \ref{secUB} we prove Proposition \ref{prop:LDPUB}.

\section{The method of Gunson-Panta}\label{appGP}
In this section we recall the main steps of the analysis of Gunson-Panta as presented in \cite{GunPan} while keeping our notation when it is in conflict with that of \cite{GunPan}. In \cite{GunPan} the charges have absolute value $q > 0$, and for our concerns $q$ should always be taken equal to $1$.

 Recall that the partition function is defined as
\[
\ZNbeta := \int_{\Box^{2N}} e^{-\frac{\beta}{2} \WN(\XN, \YN)} d\XN d\YN.
\]
This is almost exactly what is denoted by $Q^{*}_{2N}$ in \cite[(2.2)]{GunPan}, up to the fact that the domain of integration $\Lambda$ is a square, in contrast to \cite{GunPan} where it is a disk. We have a factor $\beta/2$ but the definition of $\WN$ counts each pairwise interaction twice, whereas in \cite[(2.2)]{GunPan} the temperature factor is $\beta$ but each pairwise interaction is counted only once.

In \cite[(2.3)]{GunPan} an ``electrostatic inequality"
 is used to bound below the interaction energy in terms of the quantity
\[
\sum_{i=1}^{N} 
(\log r(x_i) + \log r(y_i)),
\] 
this is the same computation as in our Lemma \ref{lem:WIPP}. It yields the bound
\begin{equation} \label{ZNannex}
\ZNbeta \leq \int_{\Box^{2N}} e^{-\frac{\beta}{2} \sum_{i=1}^{N} 
(\log r(x_i) + \log r(y_i) 
}
d\XN d\YN,
\end{equation}
as expressed in \cite[(2.4)]{GunPan} (up to  notation, and the fact that 
in the latter, a minus sign is missing in the exponential). 

Henceforth the signs of the charge will not play any role. For any $M$-tuple of points $\vec{S}_{M} = (S_1, \dots, S_{M})$, let us define the map $F : \{1, \dots, M\}  \to \{1, \dots, M\}$ such that 
\[
|S_i - S_{F(i)}| = \min_{j \in \{1, \dots, M\}} |S_i - S_j|.
\]
With this notation we may rewrite \eqref{ZNannex} as
\begin{equation} \label{rewriteZnann1}
\ZNbeta \leq \int_{\Box^{2N}} e^{-\frac{\beta}{2} \sum_{i=1}^{2N} \log ( \frac{1}{2} |S_i - S_{F_{i}}| )}  d\SN.
\end{equation}
To any $M$-tuple $\vec{S}_M$ we associate the (directed) graph $\hat\gamma(\vec{S}_M)$ of ``nearest-neighbors", whose set of vertices is $\{1, \dots, M\}$ and such that there is a directed arrow from $i$ to $F_i$ for any $i \in \{1, \dots, M\}$. We observe the following
\begin{lem} \label{lem:graphe}
For any $\vec{S}_M$, the associated graph $\hat\gamma(\vec{S}_M)$ has between $1$ and $M/2$ connected components. Each connected component is composed of a cycle of length $2$, together with trees attached to the two vertices of the cycle. 
\end{lem}
The graphs satisfying these properties are called ``functional digraphs''
(or ``functional directed graphs'') such that each connected component 
contains
 a cycle of order $2$. For any 
even
$M \geq 1$ and $1 \leq K \leq M/2$ let us denote by $\mathbf{D}_{M,K}$ the set
 of (isomorphism classes of) 
labeled functional digraphs 
with $M$ vertices and $K$ connected components, each 
possessing
 a cycle of order $2$.
A combinatorial computation (as 
the one leading to  \cite[(2.8)]{GunPan}) shows that
\begin{lem} \label{lem:combin}
For any $M \geq 1$ and $1 \leq K \leq M/2$ the 
cardinality
 of $\mathbf{D}_{M,K}$ is bounded by
\begin{equation}
|\mathbf{D}_{M,K}| \leq \frac{\Gamma(M+1) M^{M-2K}}{2^K \Gamma(K+1)\Gamma(M-2K+1)}.
\end{equation}
\end{lem}

If $\gamma \in \mathbf{D}_{M,K}$ is an isomorphism class, we denote 
by $\hat\gamma(\vec{S}_M) \equiv \gamma$ the event 
``$\hat\gamma(\vec{S}_M)$ is isomorphic to $\gamma$''.
We may then rewrite \eqref{rewriteZnann1} by splitting the domain of integration according to the isomorphism class of $\hat\gamma(\SN)$, this reads
\begin{equation} \label{rewriteZnann2}
\ZNbeta \leq \sum_{K=1}^N \sum_{\gamma \in \mathbf{D}_{N,K}} \int_{\Box^{2N} \cap \{\hat\gamma(\SN) \equiv \gamma \}} e^{-\frac{\beta}{2} \sum_{i=1}^{2N} \log ( \frac{1}{2} |S_i - S_{F_i}| )}  d\SN.
\end{equation}

Let $1 \leq K \leq N$ and $\gamma \in \mathbf{D}_{2N,K}$ be fixed, we turn to evaluating the quantity
\begin{equation} \label{gammafix}
\int_{\Box^{2N} \cap \{\hat\gamma(\SN) \equiv \gamma \}} e^{-\frac{\beta}{2} \sum_{i=1}^{2N} \log ( \frac{1}{2} |S_i - S_{F_{i}}| )}  d\SN.
\end{equation}
Let $L_1, \dots, L_K$ be the $K$ subsets of vertices associated to each connected component of the isomorphism class $\gamma$ of graphs. 
For $k \in \{1, \dots, K\}$ we perform a change of variables on the variables $S_{i}$ for $i \in L_k$. We denote by $c_k := \{i^a_k, i^b_k\}$ the two vertices on the cycle. We let for $i \in L_k$ such that $i \notin c_k$
\[
u_i := \frac{1}{2} (S_i - S_{F_i}),
\]
and we let 
\[
u_{i^a_k} := \frac{1}{2} (S_{i^a_k} - S_{i^b_k}), \quad u_{i^b_k} := \frac{1}{2} S_{i^b_k}.
\]
With respect to the new variables, the integral in \eqref{gammafix} is bounded by
\begin{equation} \label{Diric1}
4^{2N} \prod_{k=1}^K \int_{D_k} e^{-\frac{\beta}{2} \left(\sum_{i \in L_k \backslash c_k} \log u_i + 2 \log u_{i^a_k}\right)} \prod_{i \notin L_k} du_{i}\ du_{i^a_k}\ du_{i^b_k},
\end{equation}
where $D_k$ is a (suitably enlarged) domain of integration for the new variables. It may be observed that the new variables satisfy
\begin{equation*}
\sum_{i \in L_k \backslash c_k} |u_i|^2 + |u_{i^a_k}|^2 + |u_{i^b_k}|^2 \leq C .
\end{equation*}
for a certain universal constant $C$, thus each integral term in \eqref{Diric1} can be viewed as an integral over a simplex i.e. a multiple Dirichlet integral. Using classical results about such integrals following \cite[Equation (2.9)]{GunPan} we have 
\begin{lem}\label{lemdiri}
For any integer $M \ge 1$ and $K\le M/2$, $\gamma \in \mathbf{D}_{M,K}$, we have
\begin{multline} \label{Diric2}
\prod_{k=1}^K \int_{D_k} e^{-\frac{\beta}{2} \left(\sum_{i \in L_k \backslash c_k} \log u_i + 2 \log u_{i^a_k}\right)} \prod_{i \notin L_k} du_{i}\ du_{i^a_k}\ du_{i^b_k}
\\ \le 
\mathit{Diri}_{M,K} := \frac{\left(Y(1-\frac{\beta}{4})\right)^{M-2K} \left(X(1-\frac{\beta}{2})\right)^K}{\Gamma\left((M-K) - \frac{M}{2}\frac{\beta}{2}+1\right)},
\end{multline}
where $X$ and $Y$ are two functions independent of $M$ and $K$ and the bound \eqref{Diric2} depends only on $M,K,
\beta$ and not on the isomorphism class inside $\mathbf{D}_{M,K}$.\end{lem}
 With that lemma we deduce \begin{multline} \label{beforereass}
\ZNbeta \leq \sum_{K=1}^N |\mathbf{D}_{2N,K}| \mathit{Diri}_{2N,K} \\
\leq \sum_{K=1}^N \frac{\Gamma(2N+1) (2N)^{2N-2K}}{2^K \Gamma(K+1)\Gamma(2N-2K+1)}  \frac{\left(Y(1-\frac{\beta}{4})\right)^{2N-2K} \left(X(1-\frac{\beta}{2})\right)^K}{\Gamma\left((2N-K) - N\frac{\beta}{2}+1\right)}.
\end{multline}
The last step is the evaluation of the right-hand side in
 \eqref{beforereass}. From
 \cite[Equation (2.9)]{GunPan},
\begin{equation} \label{reecGP}
\ZNbeta \leq (2N)^{\beta N/2} \sum_{K=1}^N \frac{\Gamma(2N+1) \exp(2N)}{\Gamma(K+1) \Gamma(2N-K+1)} \left( Y(1-\frac{\beta}{4}) \right)^{2N-K} \left( \frac{X(1- \frac{\beta}{2})}{Y(1-\frac{\beta}{4})}\right)^{K},
\end{equation}
which gives in turn, using Newton's formula
\begin{equation}\label{nform}
\ZNbeta \leq N^{\beta N/2} C_{\beta}^N,
\end{equation}
where $C_{\beta}$ depends only on $\beta$. The value of $C_{\beta}$ is not important and it is thus enough to prove \eqref{reecGP} up to a multiplicative constant of order $C^N$. In particular it yields an upper bound on the partition function
\begin{equation} \label{ZNUB}
\log \ZNbeta \leq \frac{\beta}{2} N\log N + C_{\beta}N.
\end{equation}
Passing from \eqref{beforereass} to \eqref{reecGP} (up to a multiplicative constant of order $C^N$) is simple after observing that the summands in \eqref{beforereass} and \eqref{reecGP} differ by a factor
\[
\frac{(2N)^{2N-2K} \Gamma(2N-K+1)}{\Gamma((2N-K) - N\frac{\beta}{2} +1) \Gamma(2N-2K+1)}.
\]
Using Stirling's estimate for the Gamma function we see that the logarithm of the previous expression is equal to
\[
(2N-2K) \log N + (2N-K) \log N - (2N-K- N\frac{\beta}{2}) \log N - (2N-2K) \log N + O(N). 
\]
After simplifying we see that the ratio of the two summands in \eqref{beforereass} and \eqref{reecGP} is bounded by $C_{\beta}^N N^{\frac{\beta N}{2}}$
for some constant $C_{\beta}$ depending on $\beta$, whose precise value is not important here.

The thermodynamic limit for $\log \ZNbeta$ (as expressed in our 
Proposition \ref{expZN1}) is 
proved in \cite[Sections 3 and 4]{GunPan} using an 
interesting ``conjugation" trick. 
In this paper 
we only need to use an upper bound (as \eqref{ZNUB}) and 
more generally to follow the method of \cite[Section 2]{GunPan} that 
we have just recalled. \textit{A posteriori} our large deviation 
principle at scale $N$ implies in particular that 
Proposition \ref{prop:FGP} holds.

\section{Next order large deviations: lower bound}\label{secLB}
In this section, we use the blow-up coordinates as introduced in Section \ref{sec:blowup} and  we prove the LDP lower bound announced in Proposition \ref{prop:LDPLB}.

In the rest of this section $\bPsts$ is a fixed stationary tagged signed point process in $\probas_{\rm{inv},1}(\Lambda~\times~\configs)$ such that $\bERS[\bPsts]$ is finite, otherwise there is nothing to prove.

\subsection{Negative part of the energy}
First we observe that the negative part of the energy is semi-continuous in the suitable direction.
\begin{lem} \label{lem:USCD}
For any sequence $\{(\XN,\YN)\}_N$ such that $\iN(\XN, \YN) \in B(\bPsts, \ep)$, we have
$$ \liminf_{N\to \infty} - \frac{1}{N} \sum_{i=1}^N \left(\log \ro(x'_i) + \log \ro(y'_i)\right) \ge \bDeta(\bPsts)- o_\ep(1)\quad \text{as} \ \ep \to 0.
$$
\end{lem}
\begin{proof}
	We
	fix a family $\{\chi_{\tau}\}_{\tau \in (0,1)}$ of non-negative 
	bounded by $1$
	smooth functions such that $\chi_{\tau} \equiv 1$ on 
	$C_{1-\tau}$, $\chi_{\tau} \equiv 0$ outside $C_{1}$,
	and such that
	for any $x \in \R^2$, $\chi_{\tau}(x)$ is nonincreasing 
	with respect to $\tau$.
	We also set $I_{\tau} := \int \chi_{\tau}$, and we have $I_{\tau} \to 1$ as $\tau \t0$.

For any $M > 0$ we have
\begin{multline*}
- \frac{1}{N} \sum_{i=1}^N \left(\log \ro(x'_i) + \log \ro(y'_i)\right) \ge -\frac{1}{N} \sum_{i=1}^N \left( \log (\ro(x'_i) \vee\tau)+\log (\ro(y'_i) \vee \tau) \right)\\ \geq - \int \left(\Detat(\frac{1}{I_{\tau}} \chi_{\tau}, \Cs) \wedge M\right) d\bPsN.
\end{multline*}
The map $\Cs \mapsto \Detat(\frac{1}{I_{\tau}}  \chi_{\tau}, \Cs) \wedge M$ is continuous for the topology we use,  hence 
$$
- \liminf_{N\to \infty} \frac{1}{N} \sum_{i=1}^N \left(\log \ro(x'_i) + \log \ro(y'_i)\right) \ge \int \left(\Detat(\frac{1}{I_{\tau}} \chi_{\tau}, \Cs) \wedge M\right) d\bPsts - o_\ep(1).
$$
The monotone convergence theorem implies
\[
\lim_{M \ti} \int \left(\Detat(\frac{1}{I_{\tau}} \chi_{\tau}, \Cs) \wedge M\right) d\bPsts = \int \Detat(\frac{1}{I_{\tau}}  \chi_{\tau}, \Cs) d\bPsts.
\]
Let us observe that $\Detat(\cdot, \cdot)$ is linear in the first variable, in particular  $\Detat(\frac{1}{I_{\tau}} \chi_{\tau}, \Cs) = \frac{1}{I_{\tau}} \Deta(\chi_{\tau}, \Cs)$.
The family of functions $\{ \Detat(\chi_{\tau}, \cdot)\}_{\tau \in (0,1)}$ is monotone in $\tau$, and the monotone convergence theorem implies
$$
\lim_{\tau \t0} \int \frac{1}{I_{\tau}} \Detat(\chi_{\tau}, \Cs) d\bPsts = \int \Deta(\mathbf{1}_{C_1}, \Cs) d\bPsts.
$$
By stationarity of $\bPsts$ we have $\int \Deta(\mathbf{1}_{C_1}, \Cs) d\bPsts = \bDeta(\bPsts)$.
Chosing then $\tau$ small enough, $M$ large enough and $\epsilon$ small enough depending on $\tau, M$ yields the result.
\end{proof}

To prove the LDP lower bound, it thus suffices to prove a lower bound for 
$$\int_{\Lambda^{2N} \cap \{\iN(\XN, \YN) \in B(\bar P^s, \ep)\}} e^{-\frac{\beta}{2} \frac{1}{2\pi} \int_{\R^2}|\nab \VpN|^2 } d\XN d\YN$$ (for notation see Section \ref{sec:blowup}).
This then becomes similar to the question treated in \cite{LebSer} and we follow the same strategy: we need to show that there is a large enough volume of configurations (in fact one which is logarithmically very close to the relative entropy of $\bar P^s$) 
 on which the energy $ \int_{\R^2}|\nab \VpN|^2$ is not too large. For that we will split the box $\Lambda$ into squares of microscopic size, and we will draw signed  configurations independently at random in a Poissonian way in these squares, so that the total number of points in each square is the expected one. Sanov's theorem will guarantee  that this generates a set of configurations with the right volume.  Then we need to estimate the energy $\int_{\R^2}|\nab \VpN|^2 $ generated by each such configuration. We in fact need to make these energies restricted to each square  depend only on the configuration  in the square, which is a priori not the case. For that we use the idea of ``screening" the electric field $E_K$ generated in each square $K$, by modifying the configuration in a small neighborhood of the boundary of the square, to make the electric field energies independent and summable. 
 This is accomplished by enforcing the boundary condition $E_K\cdot\vec{n}=0$ on each boundary, which ensures that when pasting together these electric fields, the relation
\begin{equation}\label{deve}
-\div E= 2\pi\mc{C}  \quad \text{in } \R^{2}\end{equation}
holds globally. Indeed,  a vector field which is 
discontinuous across an interface has a distributional divergence 
concentrated on the interface 
and
equal to the jump of the normal derivative.
With the choice \eqref{deve},  there is no extra divergence created across the interfaces between the squares.
Even if the $E_K$'s were gradients, the global $E$ is in general no longer a gradient.  This does not matter however, since the energy of the true electric field $\nab \VpN$ generated by the configuration $\mc{C}$ 
will be shown to be  smaller than that of $E$ by Helmholtz projection. 
 In the case of positive charges with a neutralizing background, this procedure was introduced in \cite{gl13} and further refined in \cite{SS2d,RougSer,PetSer,LebSer}, to which we refer for more detail. 
   We implement this program, with the appropriate
 adaptations needed for controlling dipoles, in the rest of this section.
\subsection{The screening lemma} \label{sec:screening}
\begin{prop}\label{lem:screening}
	Let $R > 0$ and let $\Cs = (\Cp,\Cm)$ be a 
simple
	signed configuration in $\carr_R$, and $E$ an electric field satisfying 
$$-\div E= 2\pi (\Cp - \Cm) \quad \text{in} \ \carr_R.$$ Let $\Er$ be its truncation at nearest-neighbour half-distance as defined in \eqref{defeeta}. 

Let $\np := \Cp(\carr_R)$, $\nm := \Cm(\carr_R)$,  and $n := \np + \nm$, and for any $0 < \epsilon < 1$ let
\[
\nintp := \Cp(\carr_{R(1-\ep)}), \quad \nintm := \Cm(\carr_{R(1-\ep)}), \quad \nint := \nintp + \nintm,
\] 
and define
\begin{equation} \label{def:M}
M := \frac{1}{R^2} \int_{\carr_R} |\Er|^2. 
\end{equation}

Then for any $0<\ep<1$, if $R$ is  large enough,   there exists a (measurable) family of signed  configurations $\Phiscr_{\epsilon, R}(\Cs, E)$ such that for any $\Cscr \in \Phiscr_{\epsilon, R}(\Cs,E)$ we have
 \begin{enumerate}
\item $\Cs$ and  $\Cscr$ coincide in $\carr_{R(1-\ep)}$.
\item There exists a vector field $\Escr$ satisfying
\begin{enumerate}
\item $\Escr$ is compatible with $\Cscr$  in $\carr_R$ and is screened in the sense that 
\begin{equation} \label{Escrcompat}
\left\lbrace\begin{array}{ll}
- \div \Escr = 2\pi \Cscr & \ \text{in} \ \carr_R\\
\Escr \cdot \vec{n}= 0 &\ \text{on} \ \p \carr_R\end{array}\right. .
\end{equation}
\item The energy of $\Escr$ (after truncation) is controlled by that of $E$
 \begin{equation} \label{erreurEcrantage}
\int_{\carr_R} |\Escrr|^2  \le \int_{\carr_R} |\Er|^2 + C \left(\frac{MR}{\ep} +\ep R^2+ \frac{n}{\ep R}\left| \log \frac{n}{\ep R}\right|\right).
\end{equation}
\end{enumerate}
\item We have
\begin{equation} \label{contnbpoints}
\nintp(\C) \leq \np(\Cscr) \leq \np(\C) + C\left(\frac{MR}{\epsilon} + 1 + \frac{n}{\epsilon R} \left|\log \frac{n}{\epsilon R}\right|\right)
\end{equation}
and the same holds for $\nintm, \nm$.
\end{enumerate}
Moreover the map $\Cs \mapsto \Phiscr_{\epsilon,R}(\Cs,E)$ is such that if $A$ is a (measurable) subset of signed point configurations in $\configs(\carr_R)$ such that the quantities $n$ and $\nint$ are constant on $A$, then we have
\begin{multline}
\label{preservVolume} \log \Leb^{\otimes{2R^2}} \left( \bigcup_{\Cs \in A} \Phiscr_{\epsilon, R}(\Cs,E) \right) \geq \log \Leb^{\otimes n}(A) \\ - C \left((n - \nint) \log R + R + \frac{MR}{\ep} +\ep R^2+ \frac{n}{\ep R} \log \frac{n}{\ep R}\right).
\end{multline}
\end{prop}

\begin{proof}
{\it Step 1: choice of a good ``annulus''.}
Consider the disjoint ``annuli'' $A_k= C_{R-8(k+1)}\backslash C_{R-8 k}$ for $k$ ranging from $1$ to 
the integer part of $\frac{1}{8}\ep R$. There are $[\frac{1}{8}\ep R]$ such disjoint sets. 
The proportion of such $k$'s such that 
$$\int_{A_k} |E_r|^2 \le \frac{20 MR^2}{\ep R}$$
is  strictly larger than $\hal.$ Similarly, the proportion of such $k$'s such that
$$|\Cs|(A_k) \le \frac{ 20 n}{\ep R}$$ 
is strictly larger than $\hal$.
We deduce that there exists a $k_0\in [1, [\ep R]]$ such that 
\begin{equation}\label{surk0}
\int_{A_{k_0}} |E_r|^2 \le \frac{20 MR}{\ep }\qquad |\Cs|(A_{k_0}) \le \frac{20 n}{\ep R}.
\end{equation}
For brevity, we set
 $A:= A_{k_0}$ and $A_{\rm{int}}= \{x\in A, \dist (x, \p A) \ge 2\}.$

Set next $\eta= \frac{\ep R}{40 n }$ and 
\begin{equation}
	E_{r,\eta}\textcolor{red}{:}
	= \Er - \sum_{p\in\Cp\cap A_{\rm{int}} }\nabla (f_{\eta \wedge \ro(p) }-f_{\ro(p)})(x-p)+\sum_{p\in \Cm \cap A_{\rm{int}} } \nabla (f_{\eta\wedge r(p)} -f_{\ro(p)})(x-p).
\end{equation}
In other words, we replace $\delta_p^{(r(p))}$ by $\delta_p^{(\eta\wedge \ro(p))}$ for all the points $p$ in $A_{\rm{int}}$.

Computing as in \cite[Lemma 2.3]{PetSer} we may write 
\begin{multline}\label{evai}
\int_A |E_{r, \eta}|^2=\int_A |\Er|^2 + |E_{r,\eta}-\Er|^2 \\+ 2\int_A \sum_{p\in \Cp \cap A_{\rm{int}}} \nab (f_{\eta \wedge \ro(p) }-f_{\ro(p) })(x-p)\cdot \Er +\sum_{p\in \Cm \cap A_{\rm{int}} } \nabla (f_{\eta\wedge \ro(p)} -f_{\ro(p)})(x-p) \cdot \Er
\end{multline}
Integrating by parts, and using that $-\div E_r=2\pi (\sum_{x\in \Cp}\delta_x^{(\ro(x))}-\sum_{x\in \Cm}\delta_x^{(\ro(x))})$, we have 
\begin{multline}\label{evai2}
\int_A \sum_{p\in \Cp\cap A_{\rm{int}}} \nab (f_{\eta \wedge \ro(p) }-f_{\ro(p) })(x-p)\cdot \Er +\sum_{p\in \Cm \cap A_{\rm{int}} } \nabla (f_{\eta\wedge \ro(p)} -f_{\ro(p)})(x-p) \cdot \Er
\\ =2 \pi \sum_{p\in \Cp \cap A_{\rm{int}} } \int_A  (f_{\eta \wedge r(p) }-f_{r(p) })(x-p)\left( \sum_{x\in \Cp}\delta_{x}^{(\ro(x))}-\sum_{x\in \Cm}\delta_x^{(\ro(x))}\right)
\\+ 2\pi \sum_{p\in \Cm \cap A_{\rm{int}} } \int_A (f_{\eta\wedge r(p)} -f_{r(p)})(x-p)\left( \sum_{x\in \Cp}\delta_{x}^{(\ro(x))}-\sum_{x\in \Cm}\delta_x^{(\ro(x))}\right)
\end{multline}
and there are no boundary terms since $f_r $ and $f_{\eta \wedge r}$ vanish 
outside of  $B(p, \ro(p))$ with $\ro(p)\le 1$. By the same argument, and since all the balls of radius $\ro(p)$ are disjoint,  all the terms in the right-hand side of \eqref{evai2} vanish.
We are thus left with 
\begin{multline}
\int_A |E_{r, \eta}|^2=\int_A |\Er|^2 + |E_{r,\eta}-\Er|^2 
=\int_A|\Er|^2 + \sum_{p\in\Cs \cap A_{\rm{int}} } \int_{\R^2} |\nab (f_{\eta\wedge \ro(p)}-f_{\ro(p)})|^2.
\end{multline}
and using again the same integration by parts argument, we 
have
$$\int_{\R^2} |\nab (f_{\eta\wedge \ro(p)}-f_{\ro(p)})|^2=2\pi (\log \ro(p)-\log(\eta\wedge\ro(p))\le 2\pi |\log \eta|.$$ 
It follows with \eqref{surk0} that 
\begin{equation}\label{eer}
\int_A |E_{r, \eta}|^2\le \int_A |\Er|^2 + \frac{40\pi n}{\ep R} \left|\log \frac{\ep R}{40n}\right|.
\end{equation}

{\it Step 2: choice of a good boundary.}
Consider all the $\p C_t$ where $t$ is chosen so that \[
\p C_t \subset \{ x\in A_{\rm{int}}, \dist (x, \p A_{\rm{int}} ) \ge 1\}.\]
By the bound $\ro(x)\le 1$, $\p K_t$ cannot intersect any $B(p, \ro(p))$ for $p \notin A_{\rm{int}}$. Moreover, 
by the
choice of $\eta$ and \eqref{surk0}, we have  $\eta |\Cs|(A) \le \frac12$. Thus, the total perimeter of the balls  $B(p, \eta\wedge \ro(p) )$ with $p\in \Cs \cap A_{\rm{int}}$ is bounded by $\frac12$.
We deduce that there exists  $t$ such that $\p C_t \subset A_{\rm{int}}$ and $\p C_t $ intersects none of the $B(p, \eta)$ for $p\in \Cs \cap A_{\rm{int}}$ and none of the $B(p, \ro(p))$ for $p \in \Cs \backslash A_{\rm{int}}$. 
Applying a mean value argument to the integrand in
\eqref{eer}, and using \eqref{surk0},  we may also assume that $t$ is 
such that the restriction (or trace) of $E_{r, \eta}$ on $\p C_t$ is well defined as an $L^2$ function, and denoting 
\begin{equation}\label{defig}
g:= (E_{r, \eta} \cdot \vec n )\lfloor_{\p C_t}\end{equation} where $\vec{n}$ denotes the inner unit normal,
we have
\begin{equation}\label{ea}
\int_{\p C_t} |g|^2 \le 10 \left(\int_A |\Er|^2 + \frac{40\pi n}{\ep R} \left|\log \frac{\ep R}{40 n}\right|\right)\le  10 \left(\frac{20MR}{\ep} + \frac{40 \pi n}{\ep R} \left|\log \frac{\ep R}{40 n}\right|\right).
\end{equation}

{\it Step 3: construction outside $C_t$}

We take the $C_t$ given by the previous step and  keep the configuration in $C_t$ unchanged, and discard the configuration in $C_R\backslash C_t$. 
This way the first item will be verified.

Consider next
$\p C_t$ and partition each of its sides into  segments $I_l$ of length $\in [\hal, \frac32]$. This yields a natural partition of 
   $C_{t+1}\backslash C_t$ into disjoint rectangles $\mathcal R_l$, $l=1, \cdots, L$.
   Each $\p \mathcal R_l$ has four sides: one is $I_l$ (or does not exist if $\mathcal R_l$ is a corner square), one (or two in case of a corner) belongs to $\p C_{t+1}$, one is adjacent to $\p \mathcal R_{l-1},$ one to $\mathcal R_{l+1}$.
   
   For each $l$, we let $g_l$ denote the restriction of $g$ to $I_l$.
    We also define $n_0=0$ and for each $l \in [1,L]$,
 \begin{equation}\label{defnl}
 n_l= \left[\sum_{k=1}^l \int g_k \right]- \sum_{k=0}^{l-1} n_k,\qquad c_l= \sum_{k=1}^l \int g_k -\left[\sum_{k=1}^l \int g_k \right]
 \end{equation}
 where $[\cdot]$ is the integer part.
 We observe that 
 \begin{equation}\label{proprnl}
 |c_l|\le 1,\qquad  \sum_{k=1}^l n_k= \left[\sum_{k=1}^{l}\int g_k \right] 
 \end{equation}
 and 
 \begin{equation}\label{nknk}n_l= c_l-c_{l-1} -\int g_l.
 \end{equation}

 In each $\mathcal R_l$ we let $\Lambda_l$ be a set of $|n_l|$ points of sign equal to that of $n_l$, and which are a perturbation 
 of a fixed regular square lattice 
 $\Lambda_l^0$ of sidelength $1/\sqrt{|n_l|}$.
 More formally, to each $z_i\in \Lambda_l^0$ which is at distance $\ge 1/\sqrt{|n_l|}$ from $\p \mathcal R_l$, we  associate a 
	 point $x_i$ satisfying $|x_i-z_i|\leq 1/(4\sqrt{|n_l|})$, 
	 and set $\Lambda_l=\{x_i\}_{i=1}^{|n_l|}$.
 We let $h_l$ be the 
 mean zero solution to 
 $$\left\lbrace\begin{array}{ll}
  -\Delta h_l = 2\pi sgn(n_l)\displaystyle \sum_{p \in \Lambda_l} \delta_p & \quad \text{in } \mathcal R_l\\ [2mm]
  \nabla h_l \cdot \vec{n} = g_l& \text{on} \ \p \mathcal R_l \cap I_l\\
  \nabla h_l \cdot \vec{n}= - c_{l-1} & \text{on} \ \p \mathcal R_l \cap \p \mathcal R_{l-1}\\
  \nabla h_l \cdot \vec{n}= c_{l} & \text{on} \ \p \mathcal R_l \cap \p \mathcal R_{l+1}\\
  \nabla h_l \cdot \vec{n} =0 & \text{on} \ \p \mathcal R_l \cap \p C_{t+1}.\end{array}\right.
  $$
  One may check that this equation is solvable, and has a unique solution with mean zero, in view    
   of \eqref{nknk}.
  We may also write $h_l= u_l+v_l$ where 
  $$\left\lbrace\begin{array}{ll}
  -\Delta u_l = 2\pi \frac{n_l}{|\mathcal R_l|} & \quad \text{in } \mathcal R_l\\ [2mm]
  \nabla u_l \cdot \vec{n} = g_l& \text{on} \ \p \mathcal R_l \cap I_l\\
  \nabla u_l \cdot \vec{n}=- c_{l-1} & \text{on} \ \p \mathcal R_l \cap \p \mathcal R_{l-1}\\
  \nabla u_l \cdot \vec{n}= c_{l} & \text{on} \ \p \mathcal R_l \cap \p \mathcal R_{l+1}\\
  \nabla u_l \cdot \vec{n} =0 & \text{on} \ \p \mathcal R_l \cap \p C_{t+1}.\end{array}\right.
  $$
  and 
  $$\left\lbrace\begin{array}{ll}
  -\Delta v_l = 2\pi sgn(n_l)\displaystyle \sum_{p \in \Lambda_l} \delta_p- 2\pi\frac{n_l}{|\mathcal R_l|} & \quad \text{in } \mathcal R_l\\ [2mm]
  \nabla v_l \cdot \vec{n} = 0& \text{on} \ \p \mathcal R_l .\end{array}\right.
  $$
Both equations have a unique solution with zero average.  
Since the points in $\Lambda_l$ are well-separated from the boundary, using the same notation as in \eqref{defeeta},   we have 
$$\left\lbrace\begin{array}{ll}
  -\div (\nab v_l)_r = 2\pi sgn(n_l)\displaystyle \sum_{p \in \Lambda_l} \delta_p^{(r(p))}- 2\pi\frac{n_l}{|\mathcal R_l|} & \quad \text{in } \mathcal R_l\\ [2mm]
  (\nabla v_l)_r \cdot \vec{n} = 0& \text{on} \ \p \mathcal R_l ,\end{array}\right.$$
and we may write $(\nab v_l)_r= \sum_{p\in \Lambda_l} \nab G_p$ with
  $$\left\lbrace\begin{array}{ll}
  -\Delta G_p = 2\pi sgn(n_l)\(\delta_p^{(r(p))}-\frac{1}{|\mathcal R_l|}\) & \quad \text{in } \mathcal R_l\\ [2mm]
  (\nabla G_p)\cdot \vec{n} = 0& \text{on} \ \p \mathcal R_l .\end{array}\right.$$  
  We have that for each $p\in \Lambda_l$,
  $$\int_{\mathcal R_l} |\nab G_p|^2 \le C (|\log r(p)|+1).$$
  This can be proved by comparing $G_p$ to $f_{r(p)}(x-p)$ 
(defined in \eqref{def:feta}), for example as in \cite[proof of (6.23)]{PetSer}. Using that the separation of the points is at least of order $1/\sqrt{|n_l|}$, we may then write 
  $$\int_{\mathcal R_l}  |(\nab v_l)_r |^2  \le C |n_l|\log  |n_l|+ \sum_{p \neq p' \in \Lambda_l} \nab G_p \cdot \nab G_{p'},$$
  and by arrangement of the points near a lattice, the second sum is seen 
  to be comparable to $- \sum_{p\neq p' \in \Lambda_l} \log |p-p'|$. 
  This term is itself known to be  asymptotic as $|n_l| \to \infty$ to 
  $-(n_l)^2\int_{\mathcal{R}_l\times \mathcal{R}_l} \log |x-y|\, d\mu(x) 
  \, d \mu(y)$ where $\mu$ is the limit of 
  $\frac{1}{|n_l|}\sum_{p \in \Lambda_l} \delta_p$ 
  (here the uniform measure on $\mathcal{R}_l$). 
  We thus 
  conclude that whether $|n_l|\to \infty$ or not, we have 
\begin{equation}\label{nrjv}
\int_{\mathcal R_l}  |(\nab v_l)_r |^2  \le  C (n_l)^2.
\end{equation}
On the other hand, by elliptic estimates (for example \cite[Lemma 5.8]{RougSer}), we have
\begin{equation}\label{nrju}
\int_{\mathcal R_l} |\nabla u_l|^2\le C \( \int g_l^2 + c_{l-1}^2 +c_l^2\)
\end{equation}
with $C$ universal.
Since $|c_l|\le 1$ for every $l$, in view of  \eqref{nknk} we have
$(n_l)^2\le 3 +2\int g_l^2  $
and thus, combining \eqref{nrjv} and \eqref{nrju}, we find 
\begin{equation}\label{nerjh}
\int_{\mathcal R_l} |(\nab h_l)_r|^2 \le C \left(1+ \int g_l^2\right).
\end{equation}

We now define $\Escr$ to be $E$ in $C_t$ and to be $\sum_l \indic_{\mathcal R_l} \nab h_l$ in $C_{t+1}\backslash C_t$, and $\Cscr$ to be 
$$\Cscr := (\Cs \cap C_t) \cup (\cup_l \Lambda_l)$$ (with sign). We see that the normal components of $\Escr$ agree on each interface of $\p \mathcal R_l$, 
and thus
  \begin{equation}\label{eqE}
  \left\lbrace \begin{array}{ll}
  -\div \Escr = 2\pi \Cscr 
   & \quad \text{in } \mathcal R_l\\
\Escr \cdot \vec{n} = 0& \text{on} \ \p C_{t+1}. 
\end{array}\right.
  \end{equation}
Also, in view of \eqref{nerjh} we have 
$$\int_{C_t\backslash C_{t-1}} |\Escrr|^2\le C \left( 1+ \int g^2\right).$$
 We conclude with \eqref{ea} that 
 \begin{equation*}\int_{C_t} |\Escrr|^2\le  \int_{\carr_R} |\Er|^2 \\+ C \left(\frac{MR}{\ep} +1+ \frac{n}{\ep R} \left|\log \frac{\ep R}{n}\right|\right).
 \end{equation*}

Next, we extend (if needed) the configuration to $\carr_R\backslash C_{t+1}$ by just adding squares with dipoles. 
More precisely, we partition $\carr_R\backslash C_{t+1}$ into rectangles $\mathcal R$ of sidelengths in $[1,2]$.  In each rectangle we place a  positive charge $p_+$ and a negative charge $p_-$ separated from each other and from the boundary of the rectangle by at least $1/4$. We then solve for 
 $$\left\lbrace\begin{array}{ll}
  -\Delta u = 2\pi (\delta_{p_+}-\delta_{p_-}) & \quad \text{in } \mathcal R\\ [2mm]
  \nabla u \cdot \vec{n} = 0& \text{on} \ \p \mathcal R .\end{array}\right.
  $$We check as above  that $\int_{\mathcal R} |(\nab u)_r|^2 \le C$ for each such rectangle, and pasting together the electric fields $(\nab u)_r$ thus constructed and the one constructed in $C_{t+1}$, we find a vector field and a family of configurations satisfying all the desired conditions, and we see that 
 this extension has added an energy at most proportional to 
the volume of $\carr_R \backslash C_{t+1}$, i.e. $C\ep R^2$. \\

{\it Step 4: control on the number of points}

By construction, the point configuration has not been changed in $C_{R(1-\epsilon)}$, hence the left-hand inequality in \eqref{contnbpoints} holds. The number of points that have been added is given by $\sum_{l=1}^L n_l$. In view of \eqref{nknk} and \eqref{ea} we obtain
\[
\sum_{l=1}^L n_l \leq \sum_{l=1}^L n_l^2 \leq C\left(R +  \frac{MR}{\ep} + \frac{n}{\ep R} \left|\log \frac{\ep R}{n}\right|\right),  
\]
which yields the right-hand inequality in \eqref{contnbpoints}.
\\

{\it Step 5: volume estimate} 

We now turn to the proof of  \eqref{preservVolume}. Since we have discarded the point configuration in $\carr_R \backslash C_t$, in which there were at most $n-\nint$ points, we have lost a  logarithmic volume bounded 
(in absolute value) by 
\begin{equation} \label{volumelost}
(n-\nint) \log |\carr_R \backslash K_t|.
\end{equation}
On the other hand, the points that are constructed in each rectangle $\mc{R}_l$ were allowed to move independently in a small perturbation of the lattice of sidelength $1/\sqrt{|n_l|}$, e.g. they may be chosen arbitrarily in a disk of radius $\frac{1}{4\sqrt{|n_l|}}$ up to a multiplicative constant in the estimates. This allows us to create a volume of configurations of order
\[
\left(\frac{1}{4\sqrt{|n_l|}}\right)^{2 |n_l|}
\]
in each rectangle $\mc{R}_l$. Summing over $l$, we see that the 
(absolute value of the) 
logarithmic volume contribution of the points that are created is bounded  by
$$
\sum_{l =1}^L C| n_l| \log |n_l | \leq C \sum_{l = 1}^L n_l^2 \leq C\left( L + \int g^2\right) 
$$
in view of \eqref{nknk}. We have $L \leq R$ and $\int g^2$ is bounded as in \eqref{ea}, which allows us to bound the previous expression by
\begin{equation} \label{volumecreated}
C\left(R + \frac{MR}{\ep} + \frac{n}{\ep R} \log \frac{n}{\ep R}\right).
\end{equation}
 Combining  \eqref{volumelost} and \eqref{volumecreated} yields \eqref{preservVolume}.
\end{proof} 

\subsection{Screening the best electric field}
For any $R > 0$, for any $\Cs \in \configs(\carr_R)$, we let $\OR(\Cs)$ be the set of electric fields which are compatible with $\Cs$ in $\carr_R$, i.e.~such that $-\div E=2\pi \Cs$ in $\carr_R$.

We may then define $\FR(\Cs)$ to be the “best energy” associated to $\Cs$ in $\carr_R$, i.e. 
\begin{equation}
\label{BSE} \FR(\Cs) := \min \left\lbrace \frac{1}{R^2} \int_{\carr_R} |\Er|^2, E \in \OR(\Cs) \right\rbrace.
\end{equation} 

Since we may always consider the local electric field associated to $\Cs$, the set $\OR(\Cs)$ is non empty. If $\{E^{(k)}\}_k$ is a sequence in $\OR(\Cs)$ such that $\frac{1}{R^2} \int_{C_R} |E_r|^2$ is bounded, then using Lemma \ref{lem:L2Lp} we see that $\{E^{(k)}\}_k$ converges weakly (up to  
a subsequencen extraction) to some $E$ in $\Lploc(C_R, \R^2)$. Moreover we have $E \in \OR(\C)$ and the sequence $\{E_r^{(k)}\}_k$ converges weakly to $E_r$ in $L^2$. By 
the weak lower semi-continuity of the $L^2$ norm we obtain
\[ \int_{C_R} |E_r|^2 \leq \liminf_{k \ti} \int_{C_R} |E^{(k)}_r|^2. \]
This ensures that the minimum in \eqref{BSE} exists.

We next let $\G_R : \configs(\carr_R) \to \A$ be a measurable choice of an optimal electric field, i.e. be such that $\G_R(\Cs) \in \OR(\Cs)$ and 
$$
\frac{1}{R^2} \int_{\carr_R} |\G_{R,r}|^2 = \FR(\Cs).
$$

Let us also make the following observation: if $\bPsts$ is a stationary tagged signed point process, we have, for any  $R > 0$
\begin{equation} \label{FRbeatsbdW}
\Esp_{\bPsts}\left[ \FR(\Cs) \right] \leq \bdW(\bPsts). 
\end{equation}
Indeed from Lemma \ref{lem:compatible} we know that we may find a stationary electric field process $\bPelec$ which is compatible with $\bPsts$ and such that
$$
\Esp_{\bPelec}\left[ \left(\frac{1}{R^2} \int_{\carr_R} |\Er|^2\right) \right] = \bdW(\bPsts),
$$
and the left-hand side is by definition $\geq \int \FR(\Cs) d\bPsts$.

\begin{lem} \label{lem:USCFREP} The map $\FR$ is upper semi-continuous on $\configs(\carr_R)$.
\end{lem}
\begin{proof} The proof goes as in \cite[Lemma 5.8]{LebSer}. First, if $\Cs_1$ is a signed point configuration in $\carr_R$ and $\Cs_2$ is close to $\Cs_1$, then they have the same number of points (for both components). We may then evaluate the energy of the ``Neumann'' electric field $\tilde{E}$ associated to 
	$\Cs_1 - \Cs_2$ (i.e. the solution to 
	$-\div \tilde{E} = 2\pi(\Cs_1 - \Cs_2)$ with zero mean and vanishing 
	Neumann boundary conditions
	and see that it can be made arbitrarily small if $\Cs_2$ is close 
	enough  to
	$\Cs_1$. By adding $\tilde{E}$ to a properly an electric field in $\OR(\Cs_1)$ of almost minimal energy we may construct an element 
	of
	$\OR(\Cs_2)$ whose energy is bounded above by $\FR(\Cs_1) + o(1)$ as $d_{\configs}(\Cs_2, \Cs_1)$ goes to zero.
\end{proof}

Henceforth for any $R > 0$ and any $\epsilon \in (0,1)$, if $\Cs$ is a signed point configuration in $\carr_R$ we let $\Phiscr_{R, \ep}(\Cs)$ be the set of point configurations obtained by applying Proposition \ref{lem:screening} to $\Cs$ and $\G_R(\Cs)$, in other words we let (with a slight abuse of notation)
$$
\Phiscr_{R, \ep}(\Cs) := \Phiscr_{R, \ep}(\Cs, \G_R(\Cs)).
$$ 
To any configuration in $\Phiscr_{R,\ep}(\Cs)$ is associated a compatible electric field $\Escr$ whose energy is bounded in terms of $\FR(\Cs)$ according to the conclusions of Proposition \ref{lem:screening}
\begin{equation} \label{EnergyFR}
\int_{\carr_R} |\Escrr|^2  \le  \FR(\Cs) + C \left(\frac{R}{\ep}\FR(\Cs) +\ep R^2+ \frac{n}{\ep R} \log \frac{n}{\ep R}\right).
\end{equation}

\subsection{Construction of configurations}
We now turn to the construction of configurations by cutting the domain into microscopic squares as announced. 

For any $N \geq 1$ we let $\Box' := \sqrt{N} \Box$.

\paragraph{\textbf{Tiling  the domain.}}
For any integer $N \geq 1$ and any $R > 0$, we let 
\begin{equation} \label{def:barR}
\bar{R} := R(1 + (\log R)^{-1/10}). 
\end{equation}
If $R$ is such that $\sqrt{N}/\bar{R}$ is an integer, we let for convenience $\mNR := N/\bar{R}^2$. We let $\K_{N,R} := \{ \bar{K}_i\}_{i = 1, \dots, \mNR}$ be a collection of closed squares with disjoint interior which tile $\Box'$ by translated copies of $\carr_{\bar{R}}$. For any $i$ we denote by $z_i$ the center of $\bar{K}_i$, and we let $K_i$ be the square of center $z_i$ and sidelength $R$ (and whose sides are parallel to those of $\bar{K}_i$). Finally for any $\epsilon \in (0,1)$ we let $K_{i, \epsilon}$ be the square of center $z_i$ and sidelength $R(1-\epsilon)$. In particular we have $K_{i, \epsilon} \subset K_i \subset \bar{K}_i$.

\subsubsection{Generating approximating microstates} In the following lemma we show how to generate configurations  with enough phase-space volume and ressembling any given $\bPsts$.
\begin{lem} \label{lem:Sanovapprox}
	Let $\left( (\Ccp_1, \Ccm_1), \dots, (\Ccp_{\mNR}, \Ccm_{\mNR})\right)$ be $\mNR$ independent random variables such that  $(\Ccp_i, \Ccm_i)$ is distributed as $\Poissons_{|K_i}$, in other words $\Ccp_i$ and $\Ccm_i$ are the restriction to $K_i$ of a couple of independent Poisson point processes of intensity $1$. We condition this $\mNR$-tuple of random variables so that the total number of points of each sign is  about
	$N \left(\frac{R}{\bar{R}}\right)^{2}$, more precisely
\begin{equation} \label{conditionpoints}
\sum_{i=1}^{\mNR} \Ccp_i = \sum_{i=1}^{\mNR} \Ccm_i =  \lceil{N \left(\frac{R}{\bar{R}}\right)^{2}}\rceil.
\end{equation}
We define $\BM_{N,R}$ as the law of the following random variable in $\Box \times \configs$:
\begin{equation} \label{moyennediscretisee}
\frac{1}{\mNR} \sum_{i=1}^{\mNR} \delta_{(N^{-1/2} z_i, \theta_{z_i} \cdot (\Ccp_i, \Ccm_i))}.
\end{equation}
Moreover let $\Ccs$ be the signed point process obtained as the union of the signed point processes $(\Ccp_i, \Ccm_i)$ i.e.
$$\Ccs := \left( \sum_{i=1}^{\mNR} \Ccp_i, \sum_{i=1}^{\mNR} \Ccm_i\right) $$
and 
define $\HM_{N,R}$ as the law of the random variable in $\Box \times \config$
$$\frac{1}{N} \int_{\Box'} \delta_{(N^{-1/2} z, \theta_z \cdot \Ccs)} dz.$$
Then for any $\bPsts \in \probas_{s,1}(\Box \times \configs)$ the following inequality holds
\begin{equation}\label{EqSB1}
\liminf_{R \ti} \lim_{\nu \t0} \liminf_{N \ti} \frac{1}{\mNR} \log \BM_{N,R} \left( B(\bPsts, \nu) \right) \geq - \bERS[\bPsts]\,,
\end{equation}
moreover for any $\delta > 0$ we have
\begin{equation} \label{EqSB2}
\liminf_{R \ti} \lim_{\nu \t0} \liminf_{N \ti} \frac{1}{\mNR} \log \left(\BM_{N,R}, \HM_{N,R} \right)  \left( B(\bPsts, \nu) \times 
B(\bPsts, \delta) \right) \geq - \bERS[\bPsts]\,,
\end{equation}
where $\left(\BM_{N,R}, \HM_{N,R} \right) $ denotes the joint law of $\BM_{N,R}$ and $\HM_{N,R}$ with the natural coupling.
\end{lem}
\begin{proof} First let us forget about the condition on the number of points (i.e. we consider independent Poisson point processes) and about the tags (i.e. let us replace $\bPsts$ by a  signed point process $\Psts$ in $\probinv(\configs)$). Then for any fixed $R$ there holds a LDP for $\BM_{N,R}$ at speed $\mNR$ with rate function $\Ent[\cdot | \Poissons_{R}]$. This is a consequence of the classical Sanov theorem (see \cite[Section 6.2]{DemboZeitouni}) because the random variables $\theta_{z_i} \cdot (\Ccp_i, \Ccm_i)$ are i.i.d. Taking the limit $R \ti$ yields
$$
\lim_{R \ti} \lim_{\nu \t0} \liminf_{N \ti} \frac{1}{\mNR} \log \BM_{N,R} \left( B(\Psts, \nu) \right) \geq - \ERS[\Pst].
$$
We may extend this LDP to the context of tagged (signed) point processes by arguing as in \cite[Section 7]{LebSer}, where it is also shown that the condition on the number of points does not alter the LDP. This leads to \eqref{EqSB1}. The lower bound \eqref{EqSB2} follows from \eqref{EqSB1} by elementary manipulations as sketched in \cite[Section 7]{LebSer}.
\end{proof}

\paragraph{\textbf{Further conditioning on the points.}}
Let $\epsilon = (\log R)^{-3}$. For any $i \in \{1, \dots, \mNR\}$ we let $n_i$ be the number of points in $K_i$ and $n_{i,\rm{int}}$ be the number of points in $K_{i, \epsilon}$ (we add a $+$ or $-$ superscript in order to restrict ourselves to points with a positive or negative charge).

\begin{lem} \label{lem:SanovV2}
The conclusions of Lemma \ref{lem:Sanovapprox} hold after conditioning the random variables $\left( (\Ccp_1, \Ccm_1), \dots, (\Ccp_{\mNR}, \Ccm_{\mNR})\right)$ to satisfy the following additional conditions:
\begin{equation} \label{condSanov1}
\sum_{i=1}^{\mNR} \frac{n_i}{\epsilon R} \left| \log \frac{n_i}{\epsilon R}\right| \leq (\epsilon R)^{-1/4} N,
\end{equation}
\begin{equation} \label{condSanov2}
\sum_{i=1}^{\mNR} (n^{\pm}_i - n^{\pm}_{i,\rm{int}}) \log R \leq (\log R)^{-1/2} N.
\end{equation}
\end{lem}
\begin{proof}
It is enough to show that both events occur with probability $1- \exp(-NT)$ with $T$ tending to $\infty$ as $R \ti$, when throwing points as a signed Poisson point process of intensity $1$ in $\bigcup_{i =1}^{\mNR} K_i$. The result then follows from standard large deviations estimates. Indeed, to prove that we may assume \eqref{condSanov1} without changing the volume of microstates, we may observe that the exponential moments of 
\[
\Cs \mapsto \sqrt{\epsilon R}  \frac{1}{\epsilon R} n \left| \log \frac{n}{\epsilon R}\right|
\]
under the law of a Poisson point process in $\carr_R$ are bounded by $O(R^2)$ as $R \ti$. In particular, 
\begin{multline*}
\log \Poisson^s \left(\sum_{i=1}^{\mNR} \frac{n_i}{\epsilon  R} \left| \log \frac{n_i}{\epsilon R}\right| > (\epsilon R)^{-1/4} N \right) = \log \Poisson^s \left(\sum_{i=1}^{\mNR} \sqrt{\epsilon R} \frac{n_i}{\epsilon R} \left| \log \frac{n_i}{\epsilon R}\right| >(\epsilon R)^{1/4} N \right) \\
\leq - (\epsilon R)^{1/4} N + \log \Esp_{\Poisson^s} \exp\left(\sum_{i=1}^{\mNR} \sqrt{\epsilon R} \frac{n_i}{\epsilon R} \left| \log \frac{n_i}{\epsilon R}\right|\right) \\
\leq - (\epsilon R)^{1/4} N + \frac{N}{R^2} \log \Esp_{\Poisson^s_{|\carr_R}} \exp \sqrt{\epsilon R}  \frac{1}{\epsilon R} n \left| \log \frac{n}{\epsilon R}\right| \leq - (\epsilon R)^{1/4} N  + O(N).
\end{multline*}
In particular \eqref{condSanov1} indeed occurs with probability of order $1- \exp(-(\epsilon R)^{1/4} N)$. 

To prove that we may assume \eqref{condSanov2} we may argue similarly, by first observing that the exponential moments of
\[
\C \mapsto  (n_i - n_{i,\rm{int}}) (\log R)^{2}
\]
under a standard Poisson point process (of intensity $1$) in $\carr_R$ are bounded by $O(R^2)$ as $R \ti$. Indeed the quantity $n_i - n_{i,\rm{int}}$ is nothing but the number of points in the thin layer $\carr_R \backslash \carr_{R(1-\epsilon)}$ which has an area of order $R^2 \epsilon = R^2 (\log R)^{-3}$. We then deduce that
\[
\log \Poisson \left(\sum_{i=1}^{\mNR} (n^{\pm}_i - n^{\pm}_{i,\rm{int}}) \log R > (\log R)^{-1/2} N\right) \leq -C(\log R)^{1/4} N, 
\]
which implies that \eqref{condSanov2} indeed occurs with large enough probability.
\end{proof}

\subsubsection{Screening microstates}
\begin{lem} \label{lem:scrMicro} Let $\bPsts \in \probas_{s,1}(\Box \times \configs)$ and $\delta > 0$ be fixed. Let $N, R, \bar{R}$ be as above. For any $\nu > 0$ there exists a set $\Amod$ of signed point configurations in $\Box'$ of the form $\Cmods = \sum_{i=1}^{\mNR} (\Cmodp_i, \Cmodm_i)$ where $(\Cmodp_i, \Cmodm_i)$ is a signed point configuration in $K_i$, satisfying
\begin{equation} \label{bonnombrepoints}
\sum_{i=1}^{\mNR} \Cmodp_i(\Box') = \sum_{i=1}^{\mNR} \Cmodm_i(\Box') = N
\end{equation}
and such that the following holds
\begin{enumerate}
\item If $R$ is large enough, $\nu$ small enough and $N$ large enough we have for any $\Cmods \in \Amod$
\begin{equation} \label{moycontinue}
\frac{1}{N} \int_{\Box'} \delta_{(N^{-1/2} z, \theta_z \cdot \Cmods)} dz \in B(\bPsts, \frac{3\delta}{4}).
\end{equation}
\item For any $\Cmods$ in $\Amod$, there exists an electric field $\Emod$ satisfying
\begin{enumerate}
\item $\Emod$ is compatible with $\Cmods$
\begin{equation}\label{Ecompatible}
\left\lbrace\begin{array}{ll}
\div(\Emod) = 2\pi \Cmods  & \text{ in } \Box' \\
 \Emod \cdot \vec{n}= 0 & \text{on} \ \partial \Box'
\end{array}.\right.
\end{equation}
\item The energy of $\Emod$ is bounded by 
\begin{equation} \label{EnergieEmod}
\limsup_{R \ti, \nu \t0, N \ti} \frac{1}{2\pi} \int_{\R^2} |\Emod_r|^2 \leq \bdW(\bPsts) + \delta,
\end{equation}
uniformly on $\Cmods \in \Amod$.
\end{enumerate}
\item There is a good volume of such microstates
\begin{equation}
\label{goodvolume}\liminf_{R \ti, \nu \t0, N \ti}  \frac{1}{N} \log \frac{\Leb^{2N}}{|\Box'|^{2N}} \left( \Amod \right) \geq - \bERS[\bPsts].
\end{equation}
\end{enumerate}
\end{lem}
\begin{proof}
For any $\nu > 0$, let us write the conditions for a signed point configuration $\Cs := \sum_{i=1}^{\mNR} \Cs_i$
\begin{equation} \label{moy1}
\frac{1}{\Box'} \int_{\Box'} \delta_{(N^{-1/2}z, \theta_{z} \cdot \Cs)} dx \in  B(\bPsts, \delta)
\end{equation}
 \begin{equation}\label{moy2}
\frac{1}{\mNR} \sum_{i=1}^{\mNR} \delta_{(N^{-1/2} z_i, \theta_{z_i} \cdot \Cs_i)}\in B(\bPsts, \nu).
\end{equation}
By Lemma \ref{lem:Sanovapprox} we know that given $\delta > 0$, for any $N,R,\nu$  (such that $N/R \in \N$) there exists a set $\Aabs$ (“abs” as “abstract” because we generate them abstractly - and not by hand - using Sanov's theorem as explained in the previous section) of configurations $\Cabss = \sum_{i=1}^{\mNR} \Cabss_i$ with $N$ points, where $\Cabss_i$ is a point configuration in the square $K_i$, such that for any $\Cabss \in \Aabs$ \eqref{moy1} and \eqref{moy2} hold and satisfying
\begin{equation} \label{boundbelvolreg}
\liminf_{R\ti, \nu \to 0, N\ti} 
\frac{1}{2 N} \log \frac{\Leb^{2 N}}{|\Box_{N,R}|^{2 N}} ( \Aabs ) \ge -\int_{\Box} \ERS[\bPstsx] dx.
\end{equation}
To see how Lemma \ref{lem:Sanovapprox} yields \eqref{boundbelvolreg} it suffices to note that the law of the $2 N$-points signed point process $\Ccs$ of Lemma \ref{lem:Sanovapprox} coincides with the law of the point process induced by the $2 N$-th  product of the normalized Lebesgue measure on $\Box'$, and with this observation \eqref{EqSB2} gives \eqref{boundbelvolreg}.

We let $\Amod$ be the set of configurations obtained after applying the screening procedure described in Section \ref{sec:screening} with the parameter $\epsilon$ chosen as 
\begin{equation} \label{choixepsilon}
\epsilon = \frac{1}{\log^3 R}.
\end{equation}
More precisely, for each $\Cabss$ in $\Aabs$ we decompose $\Cabss$ as $\sum_{i=1}^{\mNR} \Cabss_i$ where $\Cabss_i$ is a signed point configuration in $K_i$, and for any $i= 1, \dots, \mNR$ we let
$\Phiscr_i(\Cabss)$ be the set of signed point configurations obtained after screening the configuration $\Cabss_i$ in $K_i$. Combining \eqref{contnbpoints} with \eqref{conditionpoints} and \eqref{condSanov1}, \eqref{condSanov2} we see that any signed point configuration $\Cmods$ has a total number of points with positive charge between $N - o(N)$ and $N$. We may then complete the point configurations by just adding squares with dipoles in the remaining layers  $\bigcup_{i = 1, \dots, \mNR} (\bar{K}_i \backslash K_i)$, in such a way that \eqref{bonnombrepoints} is satisfied.

 We then let $\Phimod(\Cabss)$ be the set of signed point configurations in $\Box'$ obtained as the cartesian product of the $\Phiscr_i(\Cabss)$
$$
\Phimod(\Cabss) := \prod_{i=1}^{\mNR} \Phiscr_i(\Cabss)
$$
and $\Amod$ (“mod” as “modified”) is defined as the image of $\Aabs$ by $\Phimod$. Since $\bPsts$ and $\delta$ are given,  the set $\Amod$ depends on the parameters $N,R,\epsilon, \nu$.

Let us now check that $\Amod$ satisfies the conclusions of Lemma \ref{lem:scrMicro}.

\textbf{Distance to $\bPsts$.}
To prove the first item we claim that the screening procedure preserves the closeness of  the continuous average to $\bPsts$ as expressed in \eqref{moy1} (however in general it does not preserve that of the discrete average \eqref{moy2}). The proof of such a claim was already given (in a slightly different setting) in \cite[Section 6.3.2]{LebSer}.

Since the topology on $\config$ is \textit{local}, when comparing two random signed point processes we can localize the configurations to a square of fixed size $R_0$, up to a uniform error which goes to $0$ as $R_0 \ti$. Now, the main point is that the screening procedure only modifies the configurations in a thin layer of size $\epsilon R$, (where $\epsilon$ has been chosen in \eqref{choixepsilon}) in each square $K_i$. In particular, when $R$ is large (hence $\epsilon$ is small), the vast majority of the translates of a given square $R_0$ by $z \in \Box'$ does not intersect any such thin layer, so that the configurations in them have not been modified when passing from $\Cabss$ to $\Cmods$. Fixing $R_0$ large enough and sending $R \ti$ we may thus bound the distance between the continuous average for $\Cabss$ and the one for $\Cmods$ by 
at most $\delta/4$.

\textbf{Energy.}
First we associate to any $\Cmods \in \Amod$ a screened electric field $\Emod$. We know by definition of the map $\Phiscr$ that for $\Cmods \in \Amod$, for any $i =1 \dots \mNR$ there exists an electric field $\Emod_i$ such that 
$$\left\lbrace\begin{array}{ll}  \div(\Emod_i)= 2\pi \Cmods_i & \text{in} \ K_i \\
 \Emod_i \cdot \vec{n} = 0 & \text{ on  }  \partial K_i
 \end{array}\right.$$ An electric field $\bar E^{\mathrm{mod}}_i$ satisfying  the analogous relation on $\bar K_i \backslash K_i$ can also easily be constructed, and its energy can be bounded be the number of dipoles added in that layer.  
Setting $\Emod := \sum_i \Emod_i \indic_{K_i} +\bar E^{\mathrm{mod}}_i \indic_{\bar K_i \backslash K_i}$ provides an electric field satisfying \eqref{Ecompatible}. We now turn to bounding its energy. Let $\Cabss \in \Aabs$ be such that $\Cmods$ is obtained from $\Cabss$ after screening. For any $i =1, \dots, \mNR$, the energy of $\Emod_i$ is bounded as in \eqref{EnergyFR} in terms of the “best energy” associated to $\Cabss_i$. We thus have, by summing the contributions of each square $K_i$  and  $\bar K_i \backslash K_i$ 
\begin{multline} \label{estimEnergyEmod}
\int_{\Box'} |\Emod_r|^2 = \sum_{i=1}^{\mNR} \int_{K_i} |\Emod_{i,r}|^2  + \int_{\bar K_i \backslash K_i} |\bar E^{\mathrm{mod}}_{i,r}|^2  \leq \sum_{i=1}^{\mNR}  \FR(\Cabss_i) \\ + C \sum_{i=1}^{\mNR} \frac{R}{\ep}\FR(\Cabss_i) + \sum_{i=1}^{\mNR} \frac{n_i}{\ep R} \left|\log \frac{n_i}{\ep R} \right|+ N \ep+o(N)
\end{multline} 
where $n_i$ is the number of points $|\Cabss_i|(K_i)$.

We now use the fact that the discrete average of the configurations in the square $K_i$ is close to $\bPsts$ and that $\FR$ is upper semi-continuous (see Lemma \ref{lem:USCFREP}). We thus have
\begin{equation} \label{FRmodvsbdW}
\limsup_{R \ti, \nu \t0} \frac{1}{\mNR} \sum_{i=1}^{\mNR}  \FR(\Cabss_i) \leq \int \FR(\Cs) d\bPsts \leq \bdW(\bPsts)
\end{equation}
where the last inequality follows from \eqref{FRbeatsbdW}. Combining \eqref{estimEnergyEmod}, \eqref{FRmodvsbdW} and \eqref{condSanov1} proves \eqref{EnergieEmod} .

\textbf{Volume.}
We now wish to bound the volume loss between the set $\Aabs$ of microstates generated “abstractly”  and the set $\Amod$  of configurations obtained after modification by the screening procedure. From the conclusions Proposition \ref{lem:screening} we may estimate the cost (in logarithmic volume) of screening the signed point configurations. According to \eqref{preservVolume} the loss can be controlled by
\begin{equation*}
\int_{\Cabss \in \Aabs} C \left(\sum_{i=1}^{\mNR} \left((n_i - n_{i,\rm{int}}) \log R + R\right) + \frac{R}{\ep} \sum_{i=1}^{\mNR} \left(\FR(\Cabss_i) +\ep R^2\right)+ \sum_{i=1}^{\mNR} \frac{n_i}{\ep R} \left|\log \frac{n_i}{\ep R}\right|\right).
\end{equation*}
We have $m_{N,R} R = o(N)$, and $\lim_{R, N \ti} \frac{1}{N} m_{N,R} \epsilon R^2 = 0$ due to the choice \eqref{choixepsilon}. 

We also have, from \eqref{condSanov1}
\[
\lim_{R \ti} \lim_{N \ti} \frac{1}{N} \sum_{i=1}^{\mNR} \frac{n_i}{\ep R} \left|\log \frac{n_i}{\ep R} \right|= 0.
\]

Since \eqref{FRmodvsbdW} holds we obtain
$$
\lim_{R \ti} \lim_{\nu \t0} \lim_{N \ti} \frac{1}{N} \frac{R}{\epsilon} \sum_{i=1}^{m_{N,R}} \FR(\Cabss_i) = 0.
$$
Finally from \eqref{condSanov2} we get that
\[
\sum_{i=1}^{\mNR} (n_i - n_{i,\rm{int}}) \log R = o(N)
\]
where $n_i = |\Cabss_i|(\carr_R)$ and $n_{i,\rm{int}} = |\Cabss_i|(\carr_{R(1-\ep)})$, which concludes the proof of \eqref{goodvolume}.

\end{proof}

Combining Lemmas \ref{lem:USCD} and \ref{lem:scrMicro} yields Proposition \ref{prop:LDPLB}.

\section{Next order large deviations: upper bound}\label{secUB}
In this section we prove the  Large Deviations upper bound stated in Proposition
\ref{prop:LDPUB}.

\subsection{Positive part of the energy}
First we observe that the positive part of the energy is semi-continuous in the suitable direction.
\begin{lem} \label{lem:LSCW}
Let $\bPsts \in \probas_{\rm{inv},1}(\Lambda \times \configs)$. 
For any sequence $\{(\XN,\YN)\}_N$ such that $\iN(\XN, \YN) \in B(\bPsts, \ep)$, we have
$$ \liminf_{N\to \infty} \frac{1}{2\pi N} \int_{\R^2} |\nab \VpN|^2 \ge \bdW(\bPsts) - o_{\ep}(1).
$$
\end{lem}
\begin{proof}

Define $ \bar{P}_{(\XN,\YN)}^{\mathrm{elec}}$ as
\begin{equation}
 \bar{P}_{(\XN,\YN)}^{\mathrm{elec}}:= \displaystyle \int_{\Box} \delta_{(x,\VN (\sqrt{N} x+ \cdot) )} dx,
\end{equation}
i.e. $\bar{P}_{(\XN,\YN)}^{\mathrm{elec}}$ is the push-forward of the Lebesgue measure on $\Box$ by $x \mapsto \VN ( \sqrt{N} x+ \cdot)$, where $\VN$ was defined  in \eqref{defvn}. 
Assume that $\int_{\R^2} |\nab \VpN|^2\le C N$ (otherwise there is nothing to prove). For any $m > 0$ we have
\begin{equation} \label{controlegeneral}
\int_{\Box} \frac{1}{|C_m|} \int_{C_m} |E_r|^2 d\bar{P}_{(\XN,\YN)}^{\mathrm{elec}}(x, E) \leq \frac{1}{N}\int_{\R^2} |\nab \VpN|^2.
\end{equation}
The bound \eqref{controlegeneral} implies that the push-forward of $\bar{P}_{(\XN,\YN)}^{\mathrm{elec}}$ by $(z,E) \mapsto (z,E_r)$ is tight in $\probas(\Box \times L^2_{\rm{loc}}(\R^2, \R^2))$. Using Lemma \ref{lem:L2Lp} we also get the tightness of $\bar{P}_{(\XN,\YN)}^{\mathrm{elec}}$ in $\probas(\Box \times \Lploc(\R^2, \R^2))$. The associated random tagged signed point process $\bar{P}_{(\XN,\YN)}$ (i.e. the push-forward of $\bar{P}_{(\XN,\YN)}^{\mathrm{elec}}$ by $(z,E) \mapsto (z, \mathrm{Conf}(E))$ is also tight in $\probas(\Box \times \configs)$, arguing as in Lemma \ref{lem:expotight}. Up to 
a
subsequence extraction, we may thus find $\bar{P}_0^{\rm{elec}} \in \probas(\Box \times \Lploc(\R^2, \R^2))$ and $P_0 \in \probas(\Box \times \configs)$ such that
\begin{enumerate}
\item $\bar{P}_{(\XN,\YN)}^{\mathrm{elec}}$ converges to $\bar{P}_0^{\rm{elec}}$ as $N \ti$
\item $\bar{P}_{(\XN,\YN)}$ converges to $\bar{P}_0$ as $N \ti$.
\end{enumerate}
It is not hard to check that the first marginal of $\bar{P}_0$ is the Lebesgue measure on $\Box$, that its second marginal is stationary, and that $\bPsts_0 \in B(\bPsts, 2\ep)$. 

Using Lemma \ref{lem:compatibility2},
 we obtain from \eqref{controlegeneral} that for any $m > 0$
\[
\int_{\Box} \frac{1}{|C_m|} \int_{C_m} |E_r|^2 d\bar{P}_{0}^{\mathrm{elec}}(x, E) \leq \liminf_{N \ti} \frac{1}{N}\int_{\R^2} |\nab \VpN|^2.
\]
In particular, letting $m \ti$ and using the definition of $\bdW$ we obtain 
\[
\bdW(\bar{P}_0) \leq \liminf_{N \ti} \frac{1}{N} \int_{\R^2} |\nab \VpN|^2.
\]
We conclude by letting $\epsilon \t0$ and using the lower semi-continuity of $\bdW$ among random stationary tagged processes, as stated in Lemma \ref{lem:LSCIWo}.
\end{proof}

\subsection{Bound on the nearest neighbor contributions} \label{sec:boundNN}
We are now left to bound from above 
$$\int_{\Lambda^{2N}\cap i_N^{-1}(B(\bPsts,\ep))} e^{-\frac{\beta}{2}\sum_{i=1}^N \log r(x_i')+\log r(y_i')} d\XN d\YN.$$
For $0 < \tau <1$ we will distinguish between the points whose nearest neighbor is at distance $\geq \tau$ and those with a very close neighbor, at distance $\leq \tau$. We thus write
\begin{multline} \label{prochepasproche}
\sum_{i=1}^N \log r(x_i')+\log r(y_i') = \sum_{i=1}^N \log (r(x_i') \vee \tau) + \log (r(y'_i) \vee \tau) \\
+ \sum_{i=1}^N \log (\frac{r(x_i')}{\tau}\wedge 1) + \log (\frac{r(y_i')}{\tau}  \wedge 1).
\end{multline}

\paragraph{\textbf{Points at distance $\geq \tau$.}}
The contributions due to the interactions of points at distance $\geq \tau$ is continuous, as expressed by the following 
\begin{lem} \label{lem:Cgeqtau}
Let $\bPsts \in \probas_{\rm{inv},1}(\Lambda \times \configs)$. 
For any sequence $\{\XN, \YN\}_N$ such that $\iN(\XN, \YN) \in B(\bPsts, \epsilon)$ and such that $\Nn(\Cs, \carr_1)$ is uniformly integrable against $d\bPsN$ as $N \ti$, we have
$$\liminf_{N \ti} - \sum_{i=1}^N \log (r(x_i') \vee \tau) + \log (r(y'_i) \vee \tau) \geq - \Detat(\bPsts) + o_{\epsilon}(1).$$
\end{lem}
\begin{proof}
For any $t > 0$ let $\chi_t$ be a smooth nonnegative function such that $\chi_t \equiv 1$ in $C_{1-t}$ and $\chi_{t} \equiv 0$ outside $C_{1+t}$ and such that $\int \chi_t = 1$. Let $\Lambda_t$ be the square $\{x \in \Lambda, d(x, \partial \Lambda) \geq t\}$. For $N$ large enough we have
$$
- \frac{1}{N} \sum_{i=1}^N \log (\ro(x'_i) \vee \tau)+\log (\ro(y'_i) \vee \tau)  \geq - \int \1_{\Lambda_t}(x) \left(\Detat(\chi_t, \Cs) \right) d\bPsN.
$$
The map $\Cs \mapsto - \Detat(\chi_t, \Cs)$ is continuous and bounded by $(- \log (\tau)) \Nn(\Cs, \carr_{1+1})$, and since $\Nn(\Cs, \carr_{1+t})$ is uniformly integrable against $d\bPsN$ we have
$$
\lim_{N \ti} - \int \1_{\Lambda_t}(x) \left(\Detat(\chi_t, \Cs) \right) d\bPsN  = - \int \1_{\Lambda_t}(x) \left(\Detat(\chi_t, \Cs) \right) d\bPsts + o_{\epsilon}(1).
$$
Letting $t\to 0$ yields the result.
\end{proof}

\paragraph{\textbf{Contribution of close dipoles.}}
We now turn to bounding the contributions to the Boltzmann factor $e^{-\beta \WN}$ due to pairs of points which are at distance $\leq \tau$ from each other, and see that this quantity is negligible when $\tau \t0$.
Using \eqref{prochepasproche} we see that we are left with bounding 
$$ \int_{i_N^{-1}(B(\bPsts,\ep))} e^{- \frac{\beta}{2} \sum_{i=1}^N \log \frac{r(x_i')}{\tau}\wedge 1 + \log \frac{r(y_i')}{\tau}  \wedge 1} \, d \XN\, d\YN.$$
We prove
\begin{lem}
We have 
\begin{equation}
\label{uppbdD}\limsup_{\tau \to 0, \ep \to 0, N \ti} \frac{1}{N} \log \int_{i_N^{-1}( B(\bPsts,\ep))} e^{- \frac{\beta}{2} \sum_{i=1}^N \log \frac{r(x_i')}{\tau}\wedge 1 + \log \frac{r(y_i')}{\tau}\wedge 1 } \, d \XN\, d\YN\le -\bERS(\bPsts).
\end{equation}
\end{lem} This will rely on the method of \cite{GunPan} described in Section \ref{appGP}.
\begin{proof}
For each configuration, we 
denote by $n$ the number of points of any sign for which $r(z_i') \le \tau$, and separate over the value of $n$. Without loss of generality, we may assume that these points are the first $n_x$ ones for the $x$'s and $n_y$ ones for the $y$'s, with $n_x+n_y=n$. 
 We may write 
\begin{multline} \label{decompodeltalike}
\int_{i_N^{-1}( B(\bPsts,\ep))} e^{- \frac{\beta}{2} \sum_{i=1}^{2N} \log \frac{r(x_i')}{\tau} \wedge 1+ \log \frac{r(y_i')}{\tau} \wedge 1} \, d \XN \, d\YN \\
\le 
\sum_{n =0}^{2N} \sum_{n_x+n_y=n}\left( \begin{array}{c} N \\ n_x\end{array}\right) \left( \begin{array}{c} N \\ n_y\end{array}\right)\int_{\Lambda^n} e^{-\frac{\beta}{2} \sum_{i=1}^n \log \frac{r(z_i') }{\tau } \wedge 1}\, dz_1\dots dz_n \\ \times \int_{ i_N^{-1}( B(\bPsts, \ep))} dx_{n_x+1}\dots dx_N \, dy_{n_y+1} \dots dy_{2N}.
\end{multline}

For any $n \geq 0$, $1 > \tau  > 0$ and $N \geq 1$,  define $\BNtau$ as 
the set of $n$-tuples of points in $\Box$ such that all the 
nearest-neighbor distances at blown-up (by $\sqrt{N}$) 
scale are smaller than $\tau$. We define 
\begin{equation}
\label{def:Zntaubeta} Z(n,\tau,\beta, N) = \int_{\BNtau} e^{-\frac{\beta}{2} \sum_{i=1}^n \log \frac{r(z'_i)}{\tau }}\, dz_1 \dots dz_n,
\end{equation}
with $z'_i = z_i \sqrt{N}$.

Fixing the parameter $\delta= 1/\sqrt{|\log \tau|}$, we may rewrite \eqref{decompodeltalike} as 
\begin{multline*}
\int_{i_N^{-1}( B(\bPsts,\ep))} e^{- \frac{\beta}{2} \sum_{i=1}^{2N } \log \frac{r(x_i')}{\tau} \wedge 1+ \log \frac{r(y_i')}{\tau} \wedge 1} \, d \XN \, d\YN 
\le  \sum_{n=\lfloor \delta N\rfloor+1}^{2N} Z(n,\tau, \beta,N) |\Lambda|^{2N-n}\\ +
 \sum_{n=0}^{\lfloor \delta N\rfloor}
 \sum_{n_x+n_y=n} \left( \begin{array}{c} N \\ n_x\end{array}\right) \left( \begin{array}{c} N \\ n_y\end{array}\right)
 Z(n,\tau, \beta,N) 
  \int_{\Lambda^{2N-n} \cap i_N^{-1}( B(\bPsts,\ep))} dx_{n_x+1}\dots dx_N \, dy_{n_y+1} \dots dy_{2N}.
\end{multline*}
Next, we claim that if $n \leq N$ (which is the case here), we have
\begin{equation}
\label{estzktb}
Z(n,\tau, \beta, N) \le \tau^{n} C^n  .
\end{equation}
Assuming this claim is true, and 
  observing that $i_N(\XN,\YN)\in B(\bPsts, \ep)$ implies that 
  $$i_{N-n} (x_n, \dots, x_N, y_n, \dots, y_N) \in B(\bPsts, 2\ep)$$  with obvious notation (for $\delta$ small enough depending on $\epsilon$), so we may then write
  \begin{multline}\label{dr1}
\int_{i_N^{-1}( B(\bPsts,\ep))} e^{- \frac{\beta}{2} \sum_{i=1}^{2N } \log \frac{r(x_i')}{\tau} \wedge 1+ \log \frac{r(y_i')}{\tau} \wedge 1} \, 
d \XN \, d\YN \\ \le  \sum_{n=\lfloor \delta N\rfloor +1}^{2N} 
\sum_{n_x+n_y=k} \left( \begin{array}{c} N \\ n_x\end{array}\right) \left( \begin{array}{c} N \\ n_y\end{array}\right)
 Z(n,\tau,\beta,N)\\ + 
 \sum_{n=0}^{\lfloor \delta N\rfloor }
 \sum_{n_x+n_y=n} \left( \begin{array}{c} N \\ n_x\end{array}\right) \left( \begin{array}{c} N \\ n_y\end{array}\right)
 Z(n,\tau, \beta,N) 
  \int_{ i_N^{-1}( B(\bPsts,\ep))} dx_{n_x+1}\dots dx_N \, dy_{n_y+1} \dots dy_{2N}.
\end{multline}

  The first term in the right-hand side of \eqref{dr1} is easily bounded above by $C^N \tau^{\delta N}$ in view of \eqref{estzktb}. For the second one, we note that by direct computations one may show that for $n\le \delta N$, 
  $$\sum_{n_x+n_y=n} \left( \begin{array}{c} N \\ n_x\end{array}\right) \left( \begin{array}{c} N \\ n_y\end{array}\right)\le C_\delta^N$$ with 
 $\lim_{\delta \to 0} C_\delta=1.$
 Combining this with \eqref{estzktb} and inserting into \eqref{dr1} we thus are led to 
 \begin{multline}\label{dr2}
\int_{i_N^{-1}( B(\bPsts,\ep))} e^{- \frac{\beta}{2} \sum_{i=1}^{2N } \log \frac{r(x_i')}{\tau} \wedge 1+ \log \frac{r(y_i')}{\tau} \wedge 1} \, d \XN \, d\YN \\ \le C^N \tau^{\delta N} + \delta N C_\delta^N C^{\delta N}     \int_{ \Lambda^{2N-2\delta N} \cap i_N^{-1}( B(\bPsts,\ep))} dx_{\delta N+1}\dots dx_N \, dy_{\delta N+ 1} \dots dy_{2N}.
\end{multline}
By the
choice of $\delta= 1/\sqrt{|\log \tau|} $ we find that the first term is logarithmically negligible and thus 
\begin{multline}\label{dr3} \limsup_{N\to \infty} \frac{1}{N} \log\int_{\Lambda^{2N} \cap i_N^{-1}(B(\bPsts,\ep))} e^{- \frac{\beta}{2} \sum_{i=1}^{2N } \log \frac{r(x_i')}{\tau} \wedge 1+ \log \frac{r(y_i')}{\tau} \wedge 1} \, d \XN \, d\YN \\ \le \log C_\delta +C \delta+  \limsup_{N\to \infty} \frac{1}{N}\log   \int_{ \Lambda^{2N-2\delta N} \cap i_N^{-1}(B(\bPsts,\ep))} dx_{\delta N+1}\dots dx_N \, dy_{\delta N +1} \dots dy_{2N}.
\end{multline}
But we have, from Proposition \ref{prop:Sanov}\begin{multline*}\limsup_{\ep \to 0} \limsup_{N\to \infty} \frac{1}{N}\log   \int_{ \Lambda^{2N-2\delta N} \cap i_N^{-1}( B(\bPsts,\ep))} dx_{\delta N+1}\dots dx_N \, dy_{\delta N +1} \dots dy_{2N}\\ 
\le - N(1-\delta ) \bERS(\bPsts)\end{multline*} so letting  $\ep \to 0$ and $\tau \to 0 $  (hence $\delta\to 0$) in \eqref{dr3}, we obtain the conclusion.

We now check \eqref{estzktb}, following  Section \ref{appGP}.  We will establish the more general bound
\begin{equation} \label{estimZNtaubetan}
Z(n,\tau,\beta,N) \leq \left(\frac{C \tau n}{N}\right)^n,
\end{equation}
which implies \eqref{estzktb} when $n \leq N$. We may assume that 
$n \geq 1$, otherwise \eqref{estimZNtaubetan} is obvious 
(both sides are equal to $1$).
First, we rewrite $Z(n,\tau,\beta,N)$ as 
\begin{equation*}
Z(n,\tau, \beta,N) = \int_{\BNtau} e^{-\frac{\beta}{2} \sum_{i=1}^n \log \frac{r(S_i)\sqrt{N}  }{\tau }}\, d\vec{S}_n.
\end{equation*}
As in  Section \ref{appGP},  for each configuration $(S_1, \dots, S_n)$, we form the nearest-neighbor function $i \mapsto F_i$, and associate to the configuration a directed graph with $K \in [1,n]$ connected  components, each component comprising a cycle of length 2 together with trees attached to the two vertices of the cycle.  We know moreover  that the distances to the nearest neighbor  $|S_i -S_{F_i}|$  are bounded by $\tau/\sqrt{N}$. 

Splitting between isomorphisms classes of graphs and 
following the  computations and using the notation of Section \ref{appGP}, we now have to bound 
\begin{equation}\label{45} \int_{\BNtau \cap \{ \hat\gamma (\vec{S_n})  \equiv \gamma \}} e^{-\frac{\beta}{2} \sum_{i=1}^{n} \log \( \frac{\hal |S_i - S_{F_{i}}|\sqrt{N} }{\tau} \) }  d\vec{S_n}.
\end{equation}
Let $L_1, \dots, L_K$ be the $K$ subsets of indices associated to $\gamma$. 
We make the change of variables: 
 for $i \in L_k$ such that $i \notin c_k$
\[
u_i := \sqrt{\frac{N}{n}} \frac{1}{2\tau } (S_i - S_{F_i}),
\]
and 
\[
u_{i^a_k} := \sqrt{\frac{N}{n}} \frac{1}{2\tau} (S_{i^a_k} - S_{i^b_k}), \quad u_{i^b_k} := \sqrt{\frac{N}{n}} \frac{1}{2} S_{i^b_k}.
\]
With respect to the new variables, the integral in \eqref{45} is bounded by
\begin{equation*} \label{Diric3}
 \left(\sqrt{\frac{n}{N}}\right)^{2n}  e^{-\frac{\beta}{2} n 
 \log \sqrt{\frac{n}{N}}} 4^{p_k} \tau^{2(p_k-1)} \left(\frac{1}{\sqrt{N}}\right)^{\hal \beta p_k}   \int_{D_k'} e^{-\frac{\beta}{2} \left(\sum_{i \in L_k \backslash c_k} \log u_i + 2 \log u_{i^a_k}\right)} \prod_{i \notin L_k} du_{i}\ du_{i^a_k}\ du_{i^b_k},
\end{equation*}
where $D_k'$ is a (suitably enlarged) domain of integration for the new variables where they satisfy
\begin{equation*}
\sum_{i \in L_k \backslash c_k} |u_i|^2 \leq 1,
\qquad |u_{i^a_k}|\le \frac{1}{2\sqrt{n}},
\qquad |u_{i^b_k}|\le 1
\end{equation*}
and $p_k$ denotes the cardinality of $L_k$.  
Clearly $D_k'$ is included 
in the set that was denoted $D_k$ in Section \ref{appGP}, 
so we deduce using Lemma \ref{lemdiri}
that 
 \begin{multline*}
 \int_{\BNtau \cap \{ \hat\gamma (\vec{S_n})  \equiv \gamma \}} e^{-\frac{\beta}{2} \sum_{i=1}^{n} \log \left( \frac{\hal |S_i - S_{F_{i}}|\sqrt{N} }{\tau} \right) }  d\vec{S_n} \le \left(\frac{n}{N}\right)^{n(1-\frac{\beta}{4})} 
  \frac{\tau^{2\sum_{k=1}^K (p_k-1)} } {N^{\beta/4 \sum_{k=1}^K p_k}}  \mathit{Diri}_{n, K} \\ \le \left(\frac{n}{N}\right)^{n(1-\frac{\beta}{4})} \frac{\tau^{2(n-K)}}{N^{n\beta/4}} \mathit{Diri}_{n,K}.
  \end{multline*}
  Summing over all isomorphism classes of $\gamma$, we deduce that 
  \begin{multline} \label{resumedescourses}
  Z(n, \tau, \beta, N) \le \left(\frac{n}{N}\right)^{n(1-\frac{\beta}{4})} \sum_{K=1}^{n/2} |\mathbf{D}_{n,K}| \mathit{Diri}_{n,K} \frac{\tau^{2(n-K)}}{N^{n\beta/4}}\\
  \le \left(\frac{n}{N}\right)^{n(1-\frac{\beta}{4})} \frac{\tau^n}{N^{n\beta/4}} \sum_{K=1}^{n/2} |\mathbf{D}_{n,K}| \mathit{Diri}_{n,K}
  \le \left(\frac{n}{N}\right)^{n(1-\frac{\beta}{4})} \frac{\tau^n}{N^{n\beta/4}} \left(\frac{n}{2}\right)^{\beta n/4} C_\beta^{n/2},
  \end{multline}
 which yields \eqref{estimZNtaubetan}.
 \end{proof}

\section{Leading order large deviations}
\label{sec-first-order}
In this section, we prove a large deviations upper bound at the leading order 
($N^2$) 
the joint law of the empirical measures $\muNp$ and $\muNm$, with rate function
given by \eqref{def:H}. 
\paragraph{\textbf{The limiting energy.}}
For any two probability measures $\mup, \mum$ in $\probas(\Lambda)$ we let 
\begin{equation} \label{def:H}
\H(\mup, \mum) := \min \left\lbrace \int_{\R^2} |E|^2 , E \in L^2(\R^2, \R^2) \text{ s.t. } - \div E = 2\pi(\mup - \mum) \right\rbrace,
\end{equation}
which represents the electrostatic interaction energy between 
$\mup$ and $\mum$. The infimum in \eqref{def:H} is achieved because 
the $L^2$-norm is coercive and lower semi-continuous for 
the weak-$L^2$ topology. In fact it is not difficult to check that if $\H(\mup, \mum)$ is finite, it is equal to $\int_{\R^2} |\Eloc|^2$ where $\Eloc$ is the “local electric field” defined as $\Eloc(x) := \int_{\R^2} - \nabla \log |x-t| (d\mup(t) - d\mum(t))$. 

The functional $\H$ takes value in $[0, + \infty]$ and we have
\begin{equation} \label{Hnulsamemu}
\H(\mup, \mum) = 0 \iff \mup = \mum.
\end{equation}

\paragraph{\textbf{Large deviations upper bound.}}
We give a large deviation upper bound at speed $N^2$ for the joint law of the empirical measures $\muNp$ and $\muNm$.
\begin{prop}\label{LDP1}
For any $\mup, \mum \in \probas(\Lambda)$ we have
\begin{equation}
\label{LDP1UB}
\limsup_{\epsilon \t0} \limsup_{N \ti} \frac{1}{N^2} \log \PNbeta \left( \{ (\muNp, \muNm) \in B((\mup, \mum), \epsilon) \} \right) \leq - \frac{\beta}{2} \H(\mup, \mum).
\end{equation}
\end{prop}
 Together with \eqref{Hnulsamemu} this essentially says that we must have $\muNp \approx \muNm$ for $N$ large, except with very small (of order $e^{-N^2}$) probability.
\begin{proof} 
Let $\mup, \mum$ be in $\probas(\Lambda)$ and $\epsilon > 0$. Using Lemma \ref{lem:WIPP} we may write
\begin{multline} \label{LDP1a}
\PNbeta \left( \left\lbrace (\muNp, \muNm) \in B( (\mup, \mum), \epsilon) \right\rbrace \right) \\ \leq \frac{1}{\ZNbeta} \int_{\Lambda^{2N} \cap \{(\muNp, \muNm) \in B((\mup, \mum), \epsilon) \}} \hspace{-1cm}d\XN d\YN \exp\left(-\frac{\beta}{2} \left(\int |\nabla \VNr|^2 +\sum_{i=1}^N \log(\ro(x_i)) + \sum_{i=1}^N \log(\ro(y_i) \right)\right).
\end{multline}

We claim that $\int |\nabla \VNr|^2$ is lower-semi continuous in the following sense: if $\muNp \to \mup$ and $\muNm \to \mum$ then
\begin{equation} \label{lscinabHN}
 \liminf_{N \ti} \frac{1}{N^2} \int_{\R^2} |\nabla \VNr|^2 \geq \H(\mup, \mum).
\end{equation}
Indeed, a uniform bound on the $L^2$-norm implies that $\{\frac{1}{N} \nabla \VNr\}_N$ is tight in $L^2(\R^2, \R^2)$, let us denote by $E$ a limit point. For any $N$ we have by definition 
$$- \div (\nabla \VNr) = 2 \pi \left(\sum_{i=1}^N \delta^{\ro(x_i)}_{x_i} - \sum_{i=1}^N \delta^{\ro(y_i)}_{y_i}\right) $$
and it easily implies that
$$
- \div E  = 2 \pi(\mup - \mum).
$$
On the other hand, by lower semi-continuity of the $L^2$-norm with respect to weak convergence we have
$$
\liminf_{N \ti} \frac{1}{N^2} \int_{\R^2} |\nabla \VNr|^2  \geq \int_{\R^2} |E|^2,
$$
thus \eqref{lscinabHN} holds, and together with \eqref{LDP1a} it yields
\begin{multline} \label{LDP1b}
\limsup_{\epsilon \t0} \limsup_{N \ti} \frac{1}{N^2} \log \PNbeta \left( \{ (\muNp, \muNm) \in B((\mup, \mum), \epsilon) \} \right) \\ \leq -\frac{\beta}{2} \H(\mup, \mum) - \liminf_{N \ti} \frac{1}{N^2} \log \ZNbeta \\ +\limsup_{N \ti}  \frac{1}{N^2} \log \int_{\Lambda^{2N}} \exp\left(\frac{\beta}{2} \left(\sum_{i=1}^N \log(\ro(x_i)) + \sum_{i=1}^N \log(\ro(y_i) \right) \right)  d\XN d\YN.
\end{multline}
Combining \eqref{LDP1b} with the control \eqref{expZN1} on the partition function and with \eqref{GPDg} we get
$$
\limsup_{\epsilon \t0} \limsup_{N \ti} \frac{1}{N^2} \log \PNbeta \left( \{ (\muNp, \muNm) \in B((\mup, \mum), \epsilon) \} \right) \leq -\frac{\beta}{2} \H(\mup, \mum), 
$$
which concludes the proof of the proposition.
\end{proof}

\paragraph{\textbf{Large deviations lower bound (sketch).}}
We remark (without providing details)
that a complementary large deviations lower bound, with the same rate function, 
can be derived by adapting the approximation
constructions in \cite{bz,serfatyZur}. Indeed, to construct an approximate 
configuration of points with enough volume (in the exponential
in $N^2$ scale), one may proceed as follows. First, one
cancels the common 
parts of $\mu_+$ and $\mu_-$ by positioning pairs of positive and negative
charges (henceforth referred to as dipoles) with intra-dipole separation of
$N^{-10}$ say and inter-dipole separation of at least $N^{-1}$. Then, one may 
use the construction of e.g. \cite[Theorem 2.3]{serfatyZur} and
construct a sequence of ``well-separated'' configurations (separation at least
$\eta N^{-1/2}$ between points with small enough $\eta$)
so that $\limsup_{N\to\infty} N^{-2}
w_N(\XN,Y_N)\leq H(\mu_+,\mu_-)+C(\eta)$ where
$C(\eta)\to_{\eta\to 0} 0$; the argument in \cite{bz} shows that the volume
of such configurations is large enough.

\begin{appendix}
\section{Appendix: tail estimates for the complex Gaussian multiplicative chaos,
 by Wei Wu} 
In this appendix we apply Theorem \ref{theo:LDP} and Corollary \ref{coro:ZN}
to obtain the tail asymptotics of subcritical complex Gaussian
multiplicative chaos on $\mathbb{R}^{2}$. Let $h$ be an instance of the
Gaussian free field (GFF) on $\mathbb{R}^{2}$, which we define below.
Let $D\subset \mathbb{R}^{2}$ be a bounded domain, we are interested in the
distribution of $\left\vert \int_{D}^{{}}e^{i\beta h\left( x\right)
}dx\right\vert $, for $\beta \in (0,\sqrt{2})$. This object appears in
different contexts, such as the partition function of complex random energy
type models, the scaling limit of the compactified height function in two
dimensional dimer model \cite{Dub}, and the electric vertex operator in
conformal field theory \cite{Gaw}. Another motivation comes from the
conjecture that this object describes the scaling limit of the magnetization
of XY model in the plasma phase \cite{FS}. The Lee-Yang Theorem was 
proved
for the XY model (see \cite{LS}), therefore one may further conjecture that the
characteristic function of $\int_{D}e^{i\beta h\left( x\right) }dx$ has pure
imaginary zeros. Here we focus on another perspective, which is the tail
behavior of $\int_{D}e^{i\beta h\left( x\right) }dx$, and our approach is
based on an identity relating the moments of $\left\vert \int_{D}e^{i\beta
h\left( x\right) }dx\right\vert $ to the partition function of a two
component Coulomb gas (Lemma \ref{2kmom} below). For simplicity we set $%
D=\Lambda $, the unit cube.

Formally, let $H_{0}(\mathbb{R}^{2})=\left\{ \varphi :\varphi \in
C_{0}^{\infty }(\mathbb{R}^{2}),\text{ s.t.}\int \varphi \left( x\right)
dx=0\right\} $, and denote by $H(\mathbb{R}^{2})$ the completion of $H_{0}(%
\mathbb{R}^{2})$ in $L^{2}(\mathbb{R}^{2})$. The GFF is defined as a random
distribution in $(H(\mathbb{R}^{2}))^{\prime }$, such that for any $\rho
_{1},\rho _{2}\in H(\mathbb{R}^{2})$, the covariances are given by%
\begin{equation*}
\text{Cov}\left( \left\langle h,\rho _{1}\right\rangle ,\left\langle h,\rho
_{2}\right\rangle \right) =-\int \int \rho _{1}\left( x\right) \rho
_{2}\left( y\right) \log \left\vert x-y\right\vert dxdy,
\end{equation*}%
where $\left\langle h,\rho \right\rangle \dot{=}\int h\left( x\right) \rho
\left( x\right) dx$. Therefore, to give a mathematical definition of $%
e^{i\beta h}$, one immediately runs into the issue that $h$ is only defined
up to an additive constant, so that one can only hope to define $e^{i\beta h}
$ up to a multiplicative constant. However, as we explain below, $\left\vert
\int_{D}e^{i\beta h\left( x\right) }dx\right\vert ^{2}$ can be uniquely
defined, and we then simply set $\left\vert \int_{D}e^{i\beta h\left(
x\right) }dx\right\vert =\sqrt{\left\vert \int_{D}e^{i\beta h\left( x\right)
}dx\right\vert ^{2}}$.

Let $h_{r}$ denote the GFF on a disk $D_{r}=\left\{ z:\left\vert
z\right\vert \leq r\right\} $ with zero boundary condition, and 
let $h^{\left(
m\right) }$ be the massive GFF on $\mathbb{R}^{2}$ with mass $m$. $h_{r}$, $%
h^{\left( m\right) }$ are Gaussian processes with covariances given by
\begin{eqnarray*}
\text{Cov}\left( h_{r}\left( x\right) ,h_{r}\left( y\right) \right) &\dot{=}%
&g_{r}\left( x,y\right) \\
&=&-\log \left\vert x/r-y/r\right\vert +\log \left\vert 1-x\bar{y}%
/r^{2}\right\vert ,
\end{eqnarray*}%
and Cov$\left( h^{\left( m\right) }\left( x\right) ,h^{\left( m\right)
}\left( y\right) \right) =\left( -\Delta +m^{2}\right) ^{-1}\left( x,y\right) $,
respectively. Then, as was explained in Section 6 of \cite{LRV}, both $%
e^{i\beta h_{r}}$ and $e^{i\beta h^{\left( m\right) }}$ are well defined as
random distributions in $(C_{0}^{\infty }(\mathbb{R}^{2}))^{\prime }$ for $%
\beta \in (0,\sqrt{2})$. Therefore we can define the random variable $%
\left\vert \int_{D}e^{i\beta h\left( x\right) }dx\right\vert ^{2}$ as the
distributional weak limit of $\left\vert \int_{D}e^{i\beta h_{r}\left(
x\right) }dx\right\vert ^{2}$ as $r\rightarrow \infty $, or the
distributional weak limit of $\left\vert \int_{D}e^{i\beta h^{\left(
m\right) }\left( x\right) }dx\right\vert ^{2}$ as $m\rightarrow 0$. Indeed,
given $\rho _{1},\rho _{2}\in H(\mathbb{R}^{2})$, since
\begin{equation*}
\lim_{r\rightarrow \infty }\text{Cov}\left( \left\langle h_{r},\rho
_{1}\right\rangle ,\left\langle h_{r},\rho _{1}\right\rangle \right)
=\lim_{m\rightarrow 0}\text{Cov}( \langle h^{\left( m\right) },\rho
_{1}\rangle ,\langle h^{\left( m\right) },\rho _{1}\rangle ) =\text{Cov}%
\left( \left\langle h,\rho _{1}\right\rangle ,\left\langle h,\rho
_{1}\right\rangle \right) ,
\end{equation*}%
it is possible to prove that the two constructions lead to the same limiting
object. We will not give a proof of the equivalence in this appendix, but
instead construct the limiting object using the first approach.

As was explained in Section 6.2 of \cite{LRV}, one needs a regularization
procedure to define $e^{i\beta h_{r}}$. For $\varepsilon >0$ and $x\in D_{r}$%
, let $h_{r}^{\varepsilon }\left( x\right) $ denote the average of $h$ on
the circle of radius $\varepsilon $ centered at $x$ (we assume $h$ to be
identically zero outside $D_{r}$). $h_{r}^{\varepsilon }$ is a Gaussian
process with covariances given by
\begin{eqnarray*}
\text{Cov}\left( h_{r}^{\varepsilon }\left( x\right) ,h_{r}^{\varepsilon
}\left( y\right) \right) &\dot{=}&g_{r}^{\varepsilon }\left( x,y\right) \\
&=&\frac{1}{\left( 2\pi \varepsilon \right) ^{2}}\int_{\left\vert
z-x\right\vert =\varepsilon }\int_{\left\vert w-y\right\vert =\varepsilon
}g_{r}\left( z,w\right) dzdw.
\end{eqnarray*}%
It was shown in \cite{LRV} that as $\varepsilon \rightarrow 0$,
\begin{equation*}
\varepsilon ^{-\beta ^{2}/2}\int_{D}e^{i\beta h_{r}^{\varepsilon }\left(
x\right) }dx\rightarrow \int_{D}e^{i\beta h_{r}\left( x\right) }dx\text{ in }%
L_{p}\text{, for }p>1.
\end{equation*}%
Some properties of the circle averaged field $h_{r}^{\varepsilon }$ were
summarized in Section 3.1 of \cite{DS}. In particular,
\begin{equation}
\text{Var}h_{r}^{\varepsilon }\left( x\right) =\log C\left( x;D_{r}\right)
-\log \varepsilon ,  \label{CR}
\end{equation}%
where $C\left( x;D_{r}\right) =r(1-\left\vert x\right\vert ^{2}/r^{2})$ is
the conformal radius of $D_{r}$ from $x$. Also notice that
 when $\left\vert
x-y\right\vert >2\varepsilon $, $g_{r}^{\varepsilon }\left( x,y\right)
=g_{r}\left( x,y\right) $, and $g_{r}^{\varepsilon }\left( x,y\right)
\rightarrow 0$ if either $x$ or $y$ tends to $\partial D_{r}$. Moreover, for
all $r>2$ (which we will assume for the rest of the argument) and $%
\varepsilon <1$,
\begin{equation*}
\sup_{x,y\in D_{r}}\left\vert g_{r}^{\varepsilon }\left( x,y\right)
\right\vert \leq -\log \varepsilon +C_{1}\text{, for some }C_{1}<\infty
\text{.}
\end{equation*}%
We may therefore compute, for $k\in \mathbb{N}$,%
\begin{eqnarray}
\label{prelim}
&&\mathbb{E}\left\vert \int_{D}e^{i\beta h_{r}\left( x\right) }dx\right\vert
^{2k}   \\
&=&\lim_{\varepsilon \rightarrow 0}\mathbb{E}\left\vert \varepsilon ^{-\beta
^{2}/2}\int_{D}e^{i\beta h_{r}^{\varepsilon }\left( x\right) }dx\right\vert
^{2k}  \notag \\
&=&\lim_{\varepsilon \rightarrow 0}\varepsilon ^{-\beta ^{2}k}\mathbb{E}%
\int_{D^{\otimes 2k}}\prod_{i=1}^{k}e^{i\beta h_{r}^{\varepsilon }\left(
x_{i}\right) }\prod_{j=1}^{k}e^{-i\beta h_{r}^{\varepsilon }\left(
y_{j}\right) }d\vec{x}d\vec{y}  \notag \\
&=&\lim_{\varepsilon \rightarrow 0}\varepsilon ^{-\beta
^{2}k}\int_{D^{\otimes 2k}}e^{-\frac{\beta ^{2}}{2}\text{Var}\left(
\sum_{i=1}^{k}h_{r}^{\varepsilon }\left( x_{i}\right)
-\sum_{i=1}^{k}h_{r}^{\varepsilon }\left( y_{i}\right) \right) }d\vec{x}d%
\vec{y}  \notag \\
&=&\lim_{\varepsilon \rightarrow 0}\int_{D^{\otimes
2k}}\prod_{i=1}^{k}\left( \frac{1}{C\left( x_{i};D_{r}\right) C\left(
y_{i};D_{r}\right) }\right) ^{\beta ^{2}/2}\left( \frac{%
\prod_{1<i<j<k}e^{-g_{r}^{\varepsilon }\left( x_{i},x_{j}\right)
}e^{-g_{r}^{\varepsilon }\left( y_{i},y_{j}\right) }}{\prod_{i,j}e^{-g_{r}^{%
\varepsilon }\left( x_{i},y_{j}\right) }}\right) ^{\beta ^{2}}d\vec{x}d\vec{y%
},  \nonumber
\end{eqnarray}%
where we applied (\ref{CR}) to obtain the last equation. We now show 
that the
limit above equals%
\begin{eqnarray}
\label{coulomb}
&&\int_{D^{\otimes 2k}}\prod_{i=1}^{k}\left( \frac{1}{C\left(
x_{i};D_{r}\right) C\left( y_{i};D_{r}\right) }\right) ^{\beta ^{2}/2}\left(
\frac{\prod_{1<i<j<k}e^{-g_{r}\left( x_{i},x_{j}\right) }e^{-g_{r}\left(
y_{i},y_{j}\right) }}{\prod_{i,j}e^{-g_{r}\left( x_{i},y_{j}\right) }}%
\right) ^{\beta ^{2}}d\vec{x}d\vec{y}   \\
&=&\int_{D^{\otimes 2k}}\left( \frac{\prod_{1<i<j<k}\left\vert
x_{i}-x_{j}\right\vert \left\vert y_{i}-y_{j}\right\vert }{%
\prod_{i,j}\left\vert x_{i}-y_{j}\right\vert }\right) ^{\beta
^{2}}F_{r}\left( \vec{x},\vec{y},\beta \right) d\vec{x}d\vec{y}\text{,}
\nonumber
\end{eqnarray}%
where
\begin{equation*}
F_{r}\left( \vec{x},\vec{y},\beta \right) =\prod_{i=1}^{k}\left( 1-\frac{%
\left\vert x_{i}\right\vert ^{2}}{r^{2}}\right) ^{-\frac{\beta ^{2}}{2}%
}\left( 1-\frac{\left\vert y_{i}\right\vert ^{2}}{r^{2}}\right) ^{-\frac{%
\beta ^{2}}{2}}\left[ \frac{\prod_{i,j}\left\vert 1-x_{i}\bar{y}%
_{j}/r^{2}\right\vert }{\prod_{1<i<j<k}\left\vert 1-x_{i}\bar{x}%
_{j}/r^{2}\right\vert \left\vert 1-y_{i}\bar{y}_{j}/r^{2}\right\vert }\right]
^{\beta ^{2}}.
\end{equation*}
To show the equality of \eqref{prelim} and \eqref{coulomb},
 given $\varepsilon >0$, set 
\begin{equation*}
D_{\varepsilon }=\left\{ \left( \vec{x},\vec{y}\right) \in D^{\otimes
2k}:\min_{i,j}\left\vert x_{i}-x_{j}\right\vert \geq 2\varepsilon
,\min_{i,j}\left\vert y_{i}-y_{j}\right\vert \geq 2\varepsilon
,\min_{i,j}\left\vert x_{i}-y_{j}\right\vert \geq 2\varepsilon \right\} .
\end{equation*}%
For $\left( \vec{x},\vec{y}\right) \in D_{\varepsilon }$, the integrand in (%
\ref{prelim}) and (\ref{coulomb}) coincide. Since the integrand in (\ref%
{coulomb}) has integrable singularities when $\beta \in \left( 0,\sqrt{2}%
\right) $, it suffices to prove
\begin{equation*}
\lim_{\varepsilon \rightarrow 0}\int_{D^{\otimes 2k}\backslash
D_{\varepsilon }}\prod_{i=1}^{k}\left( \frac{1}{C\left( x_{i};D_{r}\right)
C\left( y_{i};D_{r}\right) }\right) ^{\beta ^{2}/2}\left( \frac{%
\prod_{1<i<j<k}e^{-g_{r}^{\varepsilon }\left( x_{i},x_{j}\right)
}e^{-g_{r}^{\varepsilon }\left( y_{i},y_{j}\right) }}{\prod_{i,j}e^{-g_{r}^{%
\varepsilon }\left( x_{i},y_{j}\right) }}\right) ^{\beta ^{2}}d\vec{x}d\vec{y%
}=0.
\end{equation*}%
Given $\left( x_{1},...,x_{k}\right) ,\left( y_{1},...,y_{k}\right) \in
\mathbb{R}^{k}$, one can define the Gale-Shapley matching $\sigma \left(
\vec{x},\vec{y}\right) \in S_{k}$ of $\vec{x}$ with $\vec{y}$, by

\begin{enumerate}
\item Find $x_{i}$ and $y_{j}$, such that%
\begin{equation*}
\forall i^{\prime },j^{\prime }\text{, }\left\vert x_{i}-y_{j}\right\vert
<\left\vert x_{i}-y_{j^{\prime }}\right\vert \text{ and }\left\vert
x_{i}-y_{j}\right\vert <\left\vert x_{i^{\prime }}-y_{j}\right\vert \text{,}
\end{equation*}%
and set $\sigma \left( i\right) =j$.

\item Delete the points that have been matched in Step 1.

\item Iterate the procedure until all the points have been matched.
\end{enumerate}

Set
\begin{equation*}
B_{\varepsilon }=( D^{\otimes 2k}\backslash D_{\varepsilon }) \cap \left\{
\left( \vec{x},\vec{y}\right) :\sigma \left( \vec{x},\vec{y}\right) =\text{Id%
}\right\} .
\end{equation*}%
And by symmetry, it suffices to prove the integration over $B_{\varepsilon }$
vanishes as $\varepsilon \rightarrow 0$.

Using the properties of $g_{r}^{\varepsilon }\left( \cdot ,\cdot \right) $,
we have%
\begin{equation*}
\sup_{\varepsilon <e^{-C_{1}}}\sup_{\substack{ x,y,z\in D \\ \left\vert
x-y\right\vert \leq 2\left\vert x-z\right\vert }}\frac{e^{-g_{r}^{%
\varepsilon }\left( x,y\right) }}{e^{-g_{r}^{\varepsilon }\left( x,z\right) }%
}<\infty \text{.}
\end{equation*}%
Therefore, by the same argument as  \cite[Lemma A.2]{LRV}, for $\left( \vec{%
x},\vec{y}\right) \in B_{\varepsilon }$ we have the upper bound
\begin{equation*}
\left( \frac{\prod_{1<i<j<k}e^{-g_{r}^{\varepsilon }\left(
x_{i},x_{j}\right) }e^{-g_{r}^{\varepsilon }\left( y_{i},y_{j}\right) }}{%
\prod_{i,j}e^{-g_{r}^{\varepsilon }\left( x_{i},y_{j}\right) }}\right)
^{\beta ^{2}}\leq C\left( k,\beta \right) \prod_{i=1}^{k}e^{\beta
^{2}g_{r}^{\varepsilon }\left( x_{i},y_{i}\right) },
\end{equation*}%
thus%
\begin{eqnarray}
&&\int_{B_{\varepsilon }}\prod_{i=1}^{k}\left( \frac{1}{C\left(
x_{i};D_{r}\right) C\left( y_{i};D_{r}\right) }\right) ^{\beta ^{2}/2}\left(
\frac{\prod_{1<i<j<k}e^{-g_{r}^{\varepsilon }\left( x_{i},x_{j}\right)
}e^{-g_{r}^{\varepsilon }\left( y_{i},y_{j}\right) }}{\prod_{i,j}e^{-g_{r}^{%
\varepsilon }\left( x_{i},y_{j}\right) }}\right) ^{\beta ^{2}}d\vec{x}d\vec{y%
}  \notag \\
&\leq &C\left( k,\beta \right) \prod_{i=1}^{k}\int_{B_{\varepsilon }}\left(
\frac{1}{C\left( x_{i};D_{r}\right) C\left( y_{i};D_{r}\right) }\right)
^{\beta ^{2}/2}e^{\beta ^{2}g_{r}^{\varepsilon }\left( x_{i},y_{i}\right)
}dx_{i}dy_{i}.  \label{decouple}
\end{eqnarray}

Now, for $\left( \vec{x},\vec{y}\right) \in B_{\varepsilon }$, if for some $%
i $, $\left\vert x_{i}-y_{i}\right\vert <2\varepsilon $, then because $%
g_{r}^{\varepsilon }\left( x_{i},y_{i}\right) \leq -\log \varepsilon
+O\left( 1\right) $, the integrand in (\ref{decouple}) is bounded by $%
O(\varepsilon ^{-\beta ^{2}})$. The volume of the point configuration is at
most $O\left( \varepsilon ^{2}\right) $, thus (\ref{decouple}) has
integrable singularities. Since the volume of $B_{\varepsilon }$ goes to
zero as $\varepsilon \rightarrow 0$, we conclude that (\ref{prelim}) equals (%
\ref{coulomb}).

From the explicit expression of $F_{r}$, we see that $F_{r}\left( \vec{x},%
\vec{y},\beta \right) \rightarrow 1$ uniformly for all $\left( \vec{x},\vec{y%
}\right) \in D^{\otimes 2k}$, as $r\rightarrow \infty $. Finally, noting
Corollary \ref{coro:ZN}, we can send $r\rightarrow \infty $ and apply
the
dominated convergence theorem to obtain

\begin{lem}
\label{2kmom}For $\beta \in (0,\sqrt{2})$, $k\in \mathbb{N}$, and $%
\left\vert \int_{D}e^{i\beta h\left( x\right) }dx\right\vert $ defined in
the sense above, we have%
\begin{equation*}
\mathbb{E}\left\vert \int_{D}e^{i\beta h\left( x\right) }dx\right\vert
^{2k}=Z_{k,2\beta ^{2}},
\end{equation*}%
where $Z_{k,2\beta ^{2}}$ is defined as \eqref{def:ZNbeta}.
\end{lem}

\begin{coro}
\label{tail}For $\beta \in (0,\sqrt{2})$, 
\begin{equation*}
\mathbb{P}\left( \left\vert \int_{D}e^{i\beta h\left( x\right)
}dx\right\vert >x\right) =\exp \left( -c^{\ast }\left( \beta \right) x^{%
\frac{2}{\beta ^{2}}}+o(x^{\frac{2}{\beta ^{2}}})\right) ,
\quad
\mbox{as $x\to\infty$},
\end{equation*}%
and
\begin{equation*}
c^{\ast }\left( \beta \right) =\beta ^{2}\exp \left( -1+\beta ^{-2}\inf_{%
\mathcal{P}_{\text{inv,}1}\left( \Lambda ,\mathcal{X}^{s}\right) }\mathcal{%
\bar{F}}_{\beta }\right) .
\end{equation*}%
\end{coro}
(See \eqref{def:fbarbeta} for the definition of $\mathcal{\bar{F}}_{\beta }$.)

\begin{proof}
By Chebyshev's inequality and Lemma \ref{2kmom},%
\begin{equation*}
\log \mathbb{P}\left( \left\vert \int_{D}e^{i\beta h\left( x\right)
}dx\right\vert >x\right) \leq \log \frac{\mathbb{E}\left\vert
\int_{D}e^{i\beta h\left( x\right) }dx\right\vert ^{2k}}{x^{2k}}=\log
Z_{k,2\beta ^{2}}-2k\log x.
\end{equation*}%
Apply Corollary \ref{coro:ZN}, choose $k^{\ast }$ to optimize the right hand
side while neglecting the $o\left( k\right) $ term in $\log Z_{k,2\beta
^{2}} $, we obtain $k^{\ast }=\lfloor \beta ^{-2}c^{\ast }\left( \beta
\right) x^{\frac{2}{\beta ^{2}}}\rfloor $, and%
\begin{equation*}
\mathbb{P}\left( \left\vert \int_{D}e^{i\beta h\left( x\right)
}dx\right\vert >x\right) \leq \exp \left( -c^{\ast }\left( \beta \right) x^{%
\frac{2}{\beta ^{2}}}+o(x^{\frac{2}{\beta ^{2}}})\right) .
\end{equation*}

For the lower bound, fix constants $C_{1},d_{1}$, such that $%
C_{1}d_{1}+\int_{C_{1}}^{\infty }\exp \left( -c^{\ast }\left( \beta \right)
x^{\frac{2}{\beta ^{2}}}\right) dx=1$. Let $Y$ be the non-negative random
variable whose p.d.f is given by%
\[
f\left( x\right) =\left\{
\begin{array}{cc}
d_{1} & \text{if }0\leq x\leq C_{1} \\
\exp \left( -c^{\ast }\left( \beta \right) x^{\frac{2}{\beta ^{2}}}\right)
& \text{if }x>C_{1}%
\end{array}%
\right. .
\]
An explicit computation gives $\log \mathbb{E}Y^{2k}=\log \mathbb{E}%
\left\vert \int_{D}e^{i\beta h\left( x\right) }dx\right\vert ^{2k}+o\left(
k\right) $. Given $\delta >0$, denote $Y_{\delta }=\left( 1-\delta \right)
Y^{2}$. Then there exists $k_{0}\left( \delta \right) $, such that for all $%
k>k_{0}\left( \delta \right) $, $\mathbb{E}Y_{\delta }^{k}\leq \mathbb{E}%
\left\vert \int_{D}e^{i\beta h\left( x\right) }dx\right\vert ^{2k}$. Take $%
C_{2}=C_{2}\left( \delta \right) <\infty $, such that for all $k\in \mathbb{N%
}$, $\mathbb{E}\left( Y_{\delta }-C_{2}\right) ^{k}\leq \mathbb{E}\left\vert
\int_{D}e^{i\beta h\left( x\right) }dx\right\vert ^{2k}$. Therefore the tail
of $\left\vert \int_{D}e^{i\beta h\left( x\right) }dx\right\vert $ dominates
the tail of $Y_{\delta }-C_{2}$. Then for $x>\sqrt{C_{1}}$,
\begin{eqnarray*}
\mathbb{P}\left( \left\vert \int_{D}e^{i\beta h\left( x\right)
}dx\right\vert ^{2}>x^{2}\right)  &\geq &\mathbb{P}\left( Y_{\delta
}-C_{2}>x^{2}\right)  \\
&=&\mathbb{P}\left( Y>\sqrt{\frac{x^{2}+C_{2}}{1-\delta }}\right)  \\
&=&\exp \left( -c^{\ast }\left( \beta \right) \left( \frac{x^{2}+C_{2}}{%
1-\delta }\right) ^{1/\beta ^{2}}\right) .
\end{eqnarray*}%
Taking $x$ sufficiently large yields  the lower bound.
\end{proof}

\end{appendix}

\bibliographystyle{alpha}
\bibliography{2D2CP}

\end{document}